\newif\ifms
\msfalse

\ifms

\documentclass[mnsc]{informs4}

\OneAndAHalfSpacedXI

\usepackage{amsmath,amssymb,amsfonts}
\usepackage{tikz}
\allowdisplaybreaks
\usetikzlibrary{graphs}

\usepackage{mathtools}
\usepackage{notations}
\graphicspath{{./figures/}}

\usepackage{natbib}
 \bibpunct[, ]{(}{)}{,}{a}{}{,}%
 \def\bibfont{\small}%

\usepackage{rotating}
\usepackage{fancyvrb}

\usepackage{algorithm}
\usepackage{algpseudocode}
\usepackage{multirow}
\usepackage{url}
\def\UrlBreaks{\do\/\do-}

\TheoremsNumberedThrough     
\ECRepeatTheorems
\JOURNAL{Management Science}

\newtheorem{prop}{Proposition}

\newtheorem{coro}[prop]{Corollary}
\newtheorem{lemm}[prop]{Lemma}
\newtheorem{theo}[prop]{Theorem}

\newtheorem{exam}{Example}
\newtheorem{defi}{Definition}
\newtheorem{assu}{Assumption}

\newtheorem{rema}{Remark}

\EquationsNumberedThrough    

\MANUSCRIPTNO{}

\begin{document}

\RUNAUTHOR{Hu and Wager}

\RUNTITLE{Switchback Experiments under Geometric Mixing}

\TITLE{Switchback Experiments under Geometric Mixing}


\ARTICLEAUTHORS{%
\AUTHOR{Yuchen Hu}

\AFF{Industrial Engineering and Operations Research, Columbia University, \EMAIL{yh4013@columbia.edu}}

\AUTHOR{Stefan Wager}

\AFF{Graduate School of Business, Stanford University, \EMAIL{swager@stanford.edu}}

}

\ABSTRACT{
The switchback is an experimental design that measures treatment effects by repeatedly turning
an intervention on and off for a whole system. Switchback experiments are a robust way to overcome
cross-unit spillover effects; however, they are vulnerable to bias from temporal carryovers. In this
paper, we consider properties of switchback experiments in Markovian systems that mix at a geometric
rate. We find that, in this setting,
standard switchback designs suffer considerably from carryover bias: Their estimation
error decays as $T^{-1/3}$ in terms of the experiment horizon $T$, whereas in the absence of carryovers a
faster rate of $T^{-1/2}$ would have been possible. We also show, however, that judicious use of
burn-in periods can considerably improve the situation, and enable errors that decay as $\log(T)^{1/2}T^{-1/2}$.
Our formal results are mirrored in an empirical evaluation.
}

\FUNDING{
This research was supported by the National Science Foundation under Grant SES-2242876.
}


\maketitle

\else

\documentclass{article}

\usepackage{amsmath, amsthm, amssymb}
\usepackage{graphicx}
\usepackage{verbatim}
\usepackage{natbib}
\usepackage{caption}
\usepackage{subcaption}
\usepackage{fancyvrb}
\usepackage{enumerate}
\usepackage{relsize}
\usepackage{algorithm}
\usepackage{algpseudocode}
\usepackage{enumitem}
\usepackage{mathtools}
\usepackage{multirow}
\usepackage{url}
\def\UrlBreaks{\do\/\do-}

\usepackage[breaklinks=true]{hyperref}
\usepackage[margin=1.5in]{geometry}
\usepackage{xurl}
\hypersetup{colorlinks,citecolor=blue,urlcolor=blue,linkcolor=blue}

\usepackage{tikz}
\allowdisplaybreaks
\usetikzlibrary{graphs}

\usepackage{notations}
\graphicspath{{./figures/}}


\theoremstyle{plain}
\newtheorem{prop}{Proposition}

\newtheorem{conj}[prop]{Conjecture}
\newtheorem{coro}[prop]{Corollary}
\newtheorem{lemm}[prop]{Lemma}
\newtheorem{theo}[prop]{Theorem}

\theoremstyle{definition}
\newtheorem*{exam}{Example}
\newtheorem{defi}{Definition}
\newtheorem{assu}{Assumption}
\newtheorem{proo}{Proof}
\newtheorem{model}{Model}

\theoremstyle{remark}
\newtheorem{comm}{Comment}
\newtheorem{rema}{Remark}

\newcommand{\argmin}{\operatorname{argmin}}
\newcommand{\argmax}{\operatorname{argmax}}

\title{Switchback Experiments under Geometric Mixing}
\author{Yuchen Hu \\ Columbia University \\  \texttt{yh4013@columbia.edu} \and
Stefan Wager \\ Stanford University \\ \texttt{swager@stanford.edu}}
\date{Draft version \ifcase\month\or
January\or February\or March\or April\or May\or June\or
July\or August\or September\or October\or November\or December\fi \ \number%
\year\ \  }

\begin{document}

\maketitle

\begin{abstract}
The switchback is an experimental design that measures treatment effects by repeatedly turning
an intervention on and off for a whole system. Switchback experiments are a robust way to overcome
cross-unit spillover effects; however, they are vulnerable to bias from temporal carryovers. In this
paper, we consider properties of switchback experiments in Markovian systems that mix at a geometric
rate. We find that, in this setting,
standard switchback designs suffer considerably from carryover bias: Their estimation
error decays as $T^{-1/3}$ in terms of the experiment horizon $T$, whereas in the absence of carryovers a
faster rate of $T^{-1/2}$ would have been possible. We also show, however, that judicious use of
burn-in periods can considerably improve the situation, and enable errors that decay as $\log(T)^{1/2}T^{-1/2}$.
Our formal results are mirrored in an empirical evaluation.
\end{abstract}

\fi

\section{Introduction}

Switchback experiments involve repeatedly toggling a treatment of interest on and off. There are
several reasons to consider such experiments. In early work, \citet{brandt1938tests} and \citet{cochran1941double}
studied the effect of diet on milk yield in dairy cows by alternating different diets for the same cow.
Here, the motivation for using switchbacks was variance reduction: Different cows may have vastly different
baseline yields, and so using a switchback can improve precision relative to designs that only give a single
diet to each animal.
In clinical settings, a closely related paradigm is that of N-of-1 trials, where a single individual is repeatedly randomized to different treatments to estimate personalized causal effects (see, e.g., \citealt{liang2025randomization} and references therein).
More recently, there has been an explosion of interest in using
switchbacks for A/B testing in online marketplaces, where a target intervention is toggled on and off at the market level
\citep{bojinov2020design,chamandy2016experimentation,glynn2020adaptive,kastelman2018switchback,kohavi2020trustworthy,xiong2022bias}.
When applied at the market level, switchbacks help mitigate bias due to cross-unit interference or spillovers:
For example, if the intervention involves a new pricing scheme, then using a market-level switchback avoids market
distortions that could arise from simultaneously using two different pricing schemes at the same time. 

A key challenge in using switchback experiments, however, is the problem of temporal carryovers or lag effects
\citep{bojinov2020design,cochran1941double}. Treatments assigned at any specific time point can have not only an immediate effect
on the current outcome but also a longer-term effect due to the change of the system's (potentially latent) state \citep{glynn2020adaptive}.
And any approach to estimation and inference in switchbacks that does not account for temporal carryovers is prone to bias.

As a concrete example of a problem setting where carryovers are likely to matter, consider a switchback used
to compare greedy vs.~optimized matching strategies in a two-sided market in which jobs arrive sequentially
and need to be matched to available workers (consider, e.g., a ride-sharing or grocery delivery platform).
Suppose that the greedy algorithm matches each incoming job to the nearest available worker, while the optimized
strategy seeks to preserve available resources when possible (e.g., if a job could reasonably be matched to one of two
available workers and one of them is likely to be in high demand in the future, then the job could be matched
to the other one---even if they are slightly further). The optimized strategy promotes good positioning of available workers,
and so may cause abnormally good initial performance for the greedy algorithm right after the switch; conversely, the greedy
algorithm may lead to anomalously poor initial performance for the optimized strategy. Thus, a switchback analysis that
ignores carryovers may, in this setting, severely underestimate the benefit of the optimized strategy relative to the greedy one.

In this paper, we study switchbacks under an assumption that the system we are intervening on is a (non-stationary)
Markov decision process with mixing time $t_{\mix}$. Under this assumption, actions taken at time $t$ may affect the
state of the system at all times $t' > t$, but the strength of these effects decays as $\exp[-(t' - t) / t_{\mix}]$.
We let the state evolution of the system be arbitrarily non-stationary (e.g., the system may respond arbitrarily to
the time of the day or exogenous shocks like weather); however, we assume the mixing rate of the system to be uniformly
bounded from above (see Section \ref{sec:model} for a formal model).

We first consider the behavior of a standard switchback experiment, i.e., one that toggles a binary treatment at selected
time points and then estimates the treatment effect by taking the difference of the average outcome in periods where treatment
is on and the average outcome when treatment is off. Under our model, we show that this standard switchback is severely
affected by carryover bias: Given a single time series observed for $T$ time periods, the error of the standard switchback
cannot be made to decay faster than $1/\sqrt[3]{T}$. This is markedly slower than the $1/\sqrt{T}$ rate of convergence one could
have achieved with $T$ time periods and no carryovers.

We also find, however, that judicious use of burn-in periods each time the treatment assignment is switched can alleviate this issue.
Specifically, we consider a switchback design that, given a pre-specified burn-in time $b$, throws out all observations that are within
$b$ time periods of the last treatment switch. We then show that, using this design, we can estimate the global treatment effect for non-burn-in periods
with errors decaying as $\sqrt{\log(T)/T}$ with a simple difference-in-means estimator. Furthermore, we propose a bias-corrected, weighted estimator that can use data from a switchback experiment to estimate the global treatment effect across the entire experiment with errors also decaying as $\sqrt{\log(T)/T}$.

We provide central limit theorems for both estimators that remain valid in non-stationary environments.
Finally, in numerical experiments, we also find that---as expected---the use of burn-in periods
considerably helps improve the behavior of switchbacks under our Markovian model for carryover effects.

\subsection{Related work}

The problem of carryover effects in switchbacks has been considered by a number of authors, including \citet{cochran1941double},
\citet{bojinov2020design}, \citet{glynn2020adaptive}, and \citet{liang2025randomization}. \citet{cochran1941double} propose addressing the issue using a regression
model, with a regression coefficient that captures the lagged effect of past treatment on future periods. \citet{bojinov2020design}
allow for the presence of carryovers, but assume that there is a fixed (known) time horizon $m$ such that all carryover effects
of an action taken at time $t$ disappear by time $t + m$. \citet{glynn2020adaptive} consider a Markov model related to the one
used here (although their model is stationary); however, they address the problem of temporal carryovers by fitting the Markov
model by maximum likelihood rather than by adapting the switchback. Similar to us, \citet{liang2025randomization} also allow for infinite carryover effects, but in their framework the carryovers are encoded through an impulse-response outcome model, rather than through nonparametric state transitions.

Relative to existing results, we believe our approach may be helpful in settings where researchers cannot assume that carryover effects
will fully vanish after a finite amount of time (as is the case in generic Markovian systems) and do not trust a regression model to capture
all carryovers, but still want to use a switchback design to estimate treatment effects. We also emphasize that, although we make
Markovian assumptions for analytic purposes, we do not require the researcher to be able to observe the state of the system;
we only depend implicitly on this modeling assumption via the mixing time. In contrast, the maximum likelihood approach of
\citet{glynn2020adaptive} and the  Differences-In-Q approach of \cite{farias2022markovian} require observing the full state in order to compute the estimators.
The problem of treatment effect estimation in Markov decision processes under general designs
with sequential ignorability is further considered in \citet{liao2021off, liao2022batch},
\citet{kallus2020double,kallus2022efficiently} and \citet{mehrabi2024off}, but these results again require observing the
full state of the system.

Our approach to modeling non-stationarity builds on the well-known Neyman model for finite-population causal
inference \citep{neyman1923applications}. We assume that our switchback is run over $t = 1, \, \ldots, \, T$
time periods and---like in the Neyman model where each study participant can be different from the others---we allow
for each time period $t$ to be different from the others.
All we assume is that the system is Markovian (i.e., memoryless), and that
it mixes over time (i.e., the effect of past events washes out over time); see Assumptions \ref{assu:mdp} and
\ref{assu:mixing} for details.
Closely related to this finite-population viewpoint is the framework of \cite{bojinov2019time}, which provides the first pure design-based analysis of time-series experiments, whereas our model additionally incorporates tools from dynamic processes that enable the estimation of global effects.
Our central limit theorems build on statistical tools originally built
for finite-population causal inference
\citep[e.g.,][]{aronow2017estimating, bojinov2019time, bojinov2021panel, bojinov2020design, leung2022causal, li2017general, lin2013agnostic};
however, to our knowledge this paper is the first to use these tools to study non-stationary Markov
decision processes.

Finally we note that, at a high level, the switchback design can be regarded as a special case of the cluster randomized experiment \citep{imbens2015causal} in the temporal setting.
From this perspective, our work is relevant to the strand of the literature focusing on the optimal design of a clustered randomized experiment without the presence of well-defined clusters
\citep{athey2018exact,harshaw2023design,leung2021rate,leung2022causal,savje2021causal,ugander2013graph,ugander2020randomized,viviano2020experimental}. Much of this literature relies on the existence of some exposure mapping function
that fully characterizes the interference structure \citep{aronow2017estimating,manski2013identification}, with some recent exceptions including \citet{leung2021rate,leung2022causal} and \citet{savje2021causal} that allow for
misspecified exposure mappings and instead only assume, e.g., a rate of decay
on the spillover effects. In our Markovian model, the carryover bias never disappears---it only decays over time---and as such our results are closer to these recent papers allowing approximate or misspecified exposure mappings
than to analyses that depend on a well-specified exposure mapping.
In a recent advance, \citet{jia2024clustered} build on a preprint version of this paper to develop
a clustered switchback design that allows for both spatial spillovers and Markovian carryovers.

\section{A Markovian Model for Temporal Carryovers}
\label{sec:model}

We collect data from following a single system for a time period of length $T$. At each time point $t=1,\dots,T$, we assign a binary treatment $W_t\in\{0,1\}$ to the system and observe an immediate outcome $Y_t\in \RR$. In switchback experiments, the horizon is divided into blocks of equal length $l$,
and treatments are assigned such that $W_t=W_s$ if $t$ and $s$ are from the same block. For simplicity, we assume that there exists $k\in \mathbb{Z}^+$ such that $kl=T$.

We assume that we are in a setting where there exists a (potentially unobserved) state variable $S_t\in \mathcal{S}$,
and that the triples $(S_t, W_t,Y_t)$ form a non-stationary Markov decision process with transition operators $P_t(\cdot | \cdot)$. 
This transition operator can vary arbitrarily across time and captures
the influence of all exogenous events on the system. The MDP
assumption is illustrated in Figure \ref{fig:mdp}.

\begin{figure}
\centering
\resizebox{0.8\linewidth}{!}{
\begin{tikzpicture}
\node (ts1) at (3.1,0.35) {$P_0$};
\node (ts2) at (5.1,0.35) {$P_1$};
\node (ts3) at (7.1,0.35) {$P_2$};
\node (tst) at (9.1,0.35) {$P_{T-1}$};
\node (tmu1) at (4.25,-1) {$\mu_1$};
\node (tmu2) at (6.25,-1) {$\mu_2$};
\node (tmut) at (10.25,-1) {$\mu_T$};
\node (wminus) at (0,2) {$\cdots$};
\node[black, draw, circle] (w0) at (2,2) {$W_{0}$};
\node[black, draw, circle] (w1) at (4,2) {$W_{1}$};
\node[black, draw, circle] (w2) at (6,2) {$W_{2}$};
\node (wc) at (8,2) {$\cdots$};
\node[black, draw, circle] (wt) at (10,2) {$W_{T}$};
\node (sminus) at (0,0) {$\cdots$};
\node[black, draw] (s0) at (2,0) {$S_{0}$};
\node[black, draw] (s1) at (4,0) {$S_{1}$};
\node[black, draw] (s2) at (6,0) {$S_{2}$};
\node (sc) at (8,0) {$\cdots$};
\node[black, draw] (st) at (10,0) {$S_{T}$};
\node (yminus) at (2,-2) {$\cdots$};
\node[black, draw, circle] (y1) at (4,-2) {$Y_{1}$};
\node[black, draw, circle] (y2) at (6,-2) {$Y_{2}$};
\node (yc) at (8,-2) {$\cdots$};
\node[black, draw, circle] (yt) at (10,-2) {$Y_{T}$};
\graph {
(s0)->(s1), (w0)->(s1),
(s1)->(s2), (s1)->[dashed] (y1), (w1)->[bend left, dashed] (y1), (w1)->(s2),
(s2)->(sc), (s2)->[dashed] (y2), (w2)->[bend left, dashed] (y2), (w2)->(sc),
(sc)->(st), (sc)->[dashed] (yc), (wc)->[bend left, dashed] (yc), (wc)->(st),
(st)->[dashed] (yt), (wt)->[bend left, dashed] (yt)
};
\end{tikzpicture}
}
\caption{An illustration of a length-$T$ trajectory. Circles: observable variables; Rectangles: unobservable variables; Solid arrows: state transition mechanism; Dashed arrows: outcome generation mechanism.}
\label{fig:mdp}
\end{figure}

\begin{assu}
For all time points $t = 1, \, \ldots, \, T$, the evolution of the underlying state $S_t$ is governed by a pair
of (deterministic) state transition distributions $P_t^{0}(\cdot | \cdot)$ and $P_t^{1}(\cdot | \cdot)$ such that
\begin{equation}
S_{t+1} \cond \ff_t, Y_t  \sim P_t^{W_t}\p{\cdot | S_t},
\end{equation}
where $\ff_t =\sigma\p{S_t, \, W_t, \,Y_{t-1}, \,  S_{t-1}, \, W_{t-1}, \, \cdots} $ contains all information available from time $-\infty$ to time $t$.
Furthermore, there exists a set of functions $\mu_t:\mathcal{S}\times\mathcal{W}\to\RR$ such that $Y_t$ is a noisy observation of $\mu_t(S_t, W_t)$, i.e., $Y_t=\mu_t(S_t, W_t)+\epsilon_t$
for some mutually independent noises $\epsilon_t$ satisfying $\EE{\epsilon_t\cond \ff_t}=0$ and $\Var{\epsilon_t\cond \ff_t}\le \sigma^2<\infty$.
\label{assu:mdp}
\end{assu}

Although the system is not stationary in time, we assume that the system
``mixes'' at a rate of at least $1/t_{\mix}$ at each time point.
Mixing assumptions are common in the literature in contextual bandits and reinforcement learning, and
a mixing assumption of the form \eqref{eq:mixing} is used and discussed by, e.g., \citet{van1998learning},
\citet{even2005experts} and \citet{hu2021off}. In particular, \citet{hu2021off} use such an assumption
for off-policy evaluation in a partially observed Markov decision process. Unlike us, however,
these papers consider mixing assumptions in a stationary environment (i.e., where $P_t^w = P^w$
is the same for all time periods $t$), whereas here we allow the problem dynamics to change arbitrarily from
one period to the next.
Intuitively, this mixing assumption implies that, for any realized sequence of environments (or formally, transition kernels $\{P_t\}_t$), the impact of past treatment assignments decays exponentially fast. In this sense, the system gradually ‘forgets’ its history. This excludes dynamics in which the effects of early actions can accumulate or persist.

\begin{assu}
There exists a mixing time $t_{\mix}<\infty$ such that,
for all time points $t = 1, \, \ldots, \, T$ and $w\in\{0,1\}$, 
\begin{equation}
\label{eq:mixing}
    \Norm{f'P_t^{w}-fP_t^{w}}_{TV}\le \exp\p{-1/t_{\mix}}\Norm{f'-f}_{TV}
\end{equation}
for any pair of distributions $f$ and $f'$ on $S_t$.
\label{assu:mixing}
\end{assu}

\ifms
\begin{exam}[Example: Ridesharing Platform]
Consider a ridesharing marketplace in which the platform repeatedly dispatches drivers to rider requests:
\begin{itemize}
    \item Treatment $W_t$: At each time, the platform selects a dispatching policy, for instance, a greedy nearest-driver rule ($W_t=0$) or an optimized rule that accounts for future demand patterns ($W_t=1$).
    \item State $S_t$: The latent system state collects all operational conditions that influence future outcomes. This may include, but is not limited to, all online drivers’ locations and routes, outstanding requests, drivers’ beliefs, and accumulated constraints from earlier dispatches. This state is high-dimensional, difficult to track in practice, and may be only partially observed or even fully unobserved by the analyst.
    \item Outcome $Y_t$: The observed outcome at time $t$ may be a performance metric such as revenue, profit, fulfillment rate, or rider wait time.
\end{itemize}
This setting naturally exhibits nonstationarity: exogenous conditions such as weather, time of day, day of week, special events, and holidays can alter rider arrival patterns and driver supply in ways that vary arbitrarily across time. These environmental factors may not be constant or mixing at all.

At the same time, the mixing assumption captures the idea that, conditional on a fixed treatment policy and a fixed sequence of environments, the operational state of the system gradually forgets past actions. For example, the impact of a suboptimal dispatch decision dissipates as drivers complete their trips and new requests arrive. Thus, the latent state can mix at a geometric rate even though the exogenous environment introducing nonstationarity may evolve unpredictably. This combination of arbitrary environmental shifts and geometric mixing in the operational state is characteristic of many online marketplaces.
\end{exam}

\else

\begin{exam}[Ridesharing Platform]
Consider a ridesharing marketplace in which the platform repeatedly dispatches drivers to rider requests:
\begin{itemize}
    \item Trearment $W_t$: At each time, the platform selects a dispatching policy, for instance, a greedy nearest-driver rule ($W_t=0$) or an optimized rule that accounts for future demand patterns ($W_t=1$).
    \item State $S_t$: The latent system state collects all operational conditions that influence future outcomes. This may include, but is not limited to, all online drivers’ locations and routes, outstanding requests, drivers’ beliefs, and accumulated constraints from earlier dispatches. This state is high-dimensional, difficult to track in practice, and may be only partially observed or even fully unobserved by the analyst.
    \item Outcome $Y_t$: The observed outcome at time $t$ may be a performance metric such as revenue, profit, fulfillment rate, or rider wait time.
\end{itemize}
This setting naturally exhibits nonstationarity: exogenous conditions such as weather, time of day, day of week, special events, and holidays can alter rider arrival patterns and driver supply in ways that vary arbitrarily across time. These environmental factors may not be constant or mixing at all.

At the same time, the mixing assumption captures the idea that, conditional on a fixed treatment policy and a fixed sequence of environment, the operational state of the system gradually forgets past actions. For example, the impact of a suboptimal dispatch decision dissipates as drivers complete their trips and new requests arrive. Thus, the latent state can mix at a geometric rate even though the exogenous environment introducing nonstationarity may evolve unpredictably. This combination of arbitrary environmental shifts and geometric mixing in the operational state is characteristic of many online marketplaces.
\end{exam}

\fi

With the model in place, we can now formalize the causal effects we aim to learn. In simple randomized controlled trials
without interference, it is customary to focus on estimating the average treatment effect \citep{imbens2015causal},
$\text{ATE} = \frac{1}{T} \sum_{t = 1}^T \EE{Y_t(1) - Y_t(0)}$, where $Y_t(w)$ denotes the potential outcome we would
observe in the $t$-th period by setting $W_t = w$.
In the presence of temporal carryovers, however, these potential outcomes (and the induced average
treatment effect) are no longer meaningful
because the distribution $Y_t$ doesn't just depend on $W_t$, but can also depend on past actions $\{W_{t-1},W_{t-2},\dots\}$.
\begin{figure}[t]
\centering
\includegraphics[width=\linewidth]{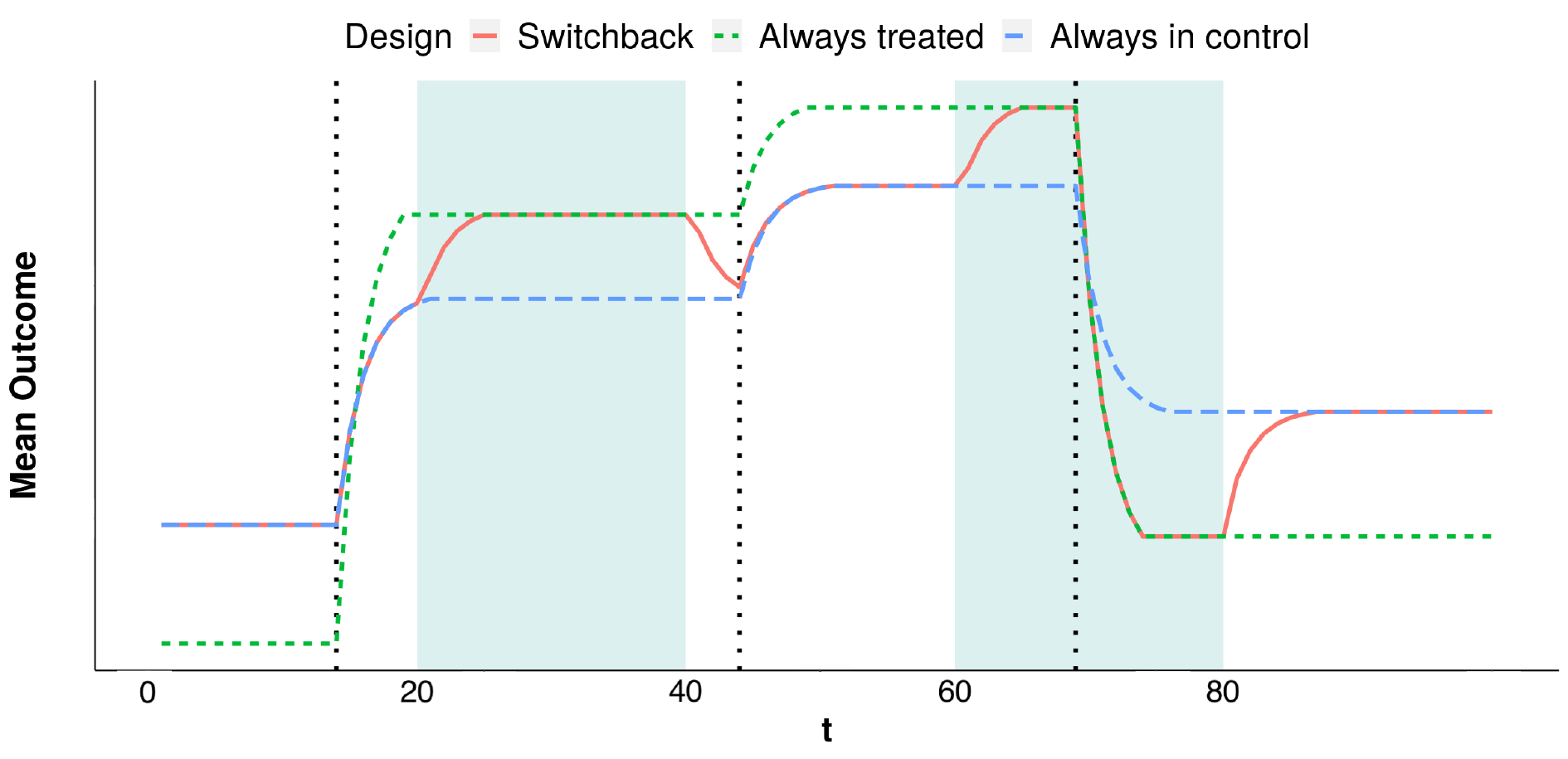}
\caption{An example of how the mean outcome can react to the switches in treatment and exogenous shock.
The blue shaded regions denote treatment blocks in the switchback experiment,
while the dotted vertical lines indicate (exogenous) time points where the underlying
market conditions (formally the $P_t^w$) change.}
\label{figure:illustration}
\end{figure}

Figure \ref{figure:illustration} illustrates difficulties associated with non-stationarity and temporal carryovers in our model.
Shaded regions indicate times when treatment is on $(W_t = 1)$ while clear regions indicate times when treatment is off ($W_t = 0$).
Meanwhile, dashed vertical lines indicate exogenous shocks to the system unrelated to the treatment; these shocks cause the functions
$P^t(\cdot|\cdot)$ and $\mu_t(\cdot)$ in Assumption \ref{assu:mdp} to change. Given this setting, the dashed blue line shows the average
behavior we would get if the system were given control assignment throughout, the dotted green one shows average behavior with treatment
throughout, and the solid red line shows average behavior with treatment toggled as in a switchback. Thanks to our mixing condition
(Assumption \ref{assu:mixing}), the red line eventually converges to the blue or green lines after each time we toggle the treatment---but this
change is not instantaneous.

Intuitively, we can think of the average gap between the green and blue curves in Figure \ref{figure:illustration} as quantifying
an average effect we want to estimate. To formalize this notion, for any treatment sequence $w \in \{0, 1\}^{\mathbb{Z}_+}$, we write
\begin{equation}
\law_w^t = \text{ distribution of } Y_t | W_t = w_0, W_{t-1} = w_1, W_{t-2} = w_2, \cdots
\end{equation}
\sloppy{i.e., $\law_w^t$ is the long-term outcome distribution at time $t$ under the treatment assignment sequence $w$, which is conditioned implicitly on the sequence of transitional kernels $\{P^1_t,P^0_t,P^1_{t-1},P^0_{t-1},\cdots\}$.
We use short-hand $\law_0^t$ and $\law_1^t$ for distributions indexed by ``pure'' histories, i.e., where $w$ is all $0$s or all $1$s. 
Given this notation, we define the following estimands.}

\begin{defi}
Under Assumption \ref{assu:mdp}, the stable treatment effect at time $t$ is\footnote{Notice that this is equivalent to a potential outcome specification, where $Y_t = Y_t(w)$ and $w$ indexes the full history.}
\begin{equation}
\tau_t= \EE[\law_1^t]{Y_t}-\EE[\law_0^t]{Y_t},
\end{equation}
and the global average treatment effect (GATE) is
\begin{equation}
\tau_{\GATE} = \frac{1}{T}\sum_{t=1}^T \tau_t.
\label{eq:estimand_gape}
\end{equation}
\label{defi:estimands}
\end{defi}

The global average treatment effect is a direct analogue to the estimand considered in \cite{ugander2013graph} and \cite{xiong2022bias} for the Markov setting.
In a cross-sectional randomized trial with spillover effects between the units, GATE is the average difference in outcomes when all units are exposed
to treatment versus when all units are exposed to control, conditionally on the group of recruited subjects. Here in the longitudinal setting, our estimand
of interest is the average difference in outcomes when the system is always exposed to treatment versus when the system is always exposed to control,
conditionally on the time period during which the experiment is conducted. 
We emphasize that this estimand is only well-defined conditionally on environments before and during the experiment, i.e., on the sequence of transition kernels from time $-\infty$ to time $T$.
This is also related to the long-run average reward studied in
\citet{glynn2020adaptive}, except that in our setting the system is never stationary due to the exogenous and arbitrary transition operator.

\ifms
\begin{exam}[Example: Ridesharing Platform Continued]
To further understand the conditional estimand in Definition \ref{defi:estimands}, it is helpful to revisit the ridesharing example introduced above. Recall that the state $S_t$ captures operational conditions that evolve as a consequence of the dispatching policy, which a decision-maker will typically expect to behave similarly in future periods if the same treatment policy were adopted. For example, if the platform were to switch permanently to an optimized dispatching policy, under similar environmental conditions, the steady-state distribution of driver locations and queue lengths would be determined by that policy.

By contrast, the exogenous environmental conditions may change unpredictably and need not repeat in future deployment environments. These environmental factors are embedded in the time-varying transition operators $\cb{P_t}$. Because they are external to the system’s operational response and may vary arbitrarily over time, it is natural for the estimand to be conditional on these environmental variables for robustness.

Consequently, the GATE captures the expected difference in outcomes between always using the treatment versus always using the control in the same sequence of environmental conditions under which the experiment was conducted. For instance, if the experiment is run throughout the month of December 2025, then GATE represents the effect of adopting the treatment for the month of December 2025 (with its own pattern of holidays and weather), but not necessarily the effect for an entire future year.

This distinction is important for decision-making. Managers are typically interested in understanding how a new policy would perform in future periods with similar environmental profiles to those observed during the experiment. Meanwhile, extrapolating the treatment effect to environments not represented in the experiment (e.g., inferring July effects from December data) would require additional modeling assumptions about how $\cb{P_t}$ might change outside the experimental window. By conditioning on the realized sequence of environmental conditions but averaging over the operational factors that the policy itself influences, GATE aligns naturally with the causal question that switchback experiments are designed to answer.\footnote{Mathematically, the random outcome noise $\epsilon_t$ could be absorbed into the time-indexed outcome function without affecting the numerical value of the estimand. However, doing so changes the philosophical interpretation of the estimand, and specifically, what sources of randomness are being averaged over. This also alters the variance structure of the estimator, which may in turn lead to a different optimal design.}
\end{exam}

\else

\begin{exam}[Ridesharing Platform Continued]
To further understand the conditional estimand in Definition \ref{defi:estimands}, it is helpful to revisit the ridesharing example introduced above. Recall that the state $S_t$ captures operational conditions that evolve as a consequence of the dispatching policy, which a decision-maker will typically expect to behave similarly in future periods if the same treatment policy were adopted. For example, if the platform were to switch permanently to an optimized dispatching policy, under similar environmental conditions, the steady-state distribution of driver locations and queue lengths would be determined by that policy.

By contrast, the exogenous environmental conditions may change unpredictably and need not repeat in future deployment environments. These environmental factors are embedded in the time-varying transition operators $\cb{P_t}$. Because they are external to the system’s operational response and may vary arbitrarily over time, it is natural for the estimand to be conditional on these environmental variables for robustness.

Consequently, the GATE captures the expected difference in outcomes between always using the treatment versus always using the control in the same sequence of environmental conditions under which the experiment was conducted. For instance, if the experiment is run throughout the month of December 2025, then GATE represents the effect of adopting the treatment for the month of December 2025 (with its own pattern of holidays and weather), but not necessarily the effect for an entire future year.

This distinction is important for decision-making. Managers are typically interested in understanding how a new policy would perform in future periods with similar environmental profiles to those observed during the experiment. Meanwhile, extrapolating the treatment effect to environments not represented in the experiment (e.g., inferring July effects from December data) would require additional modeling assumptions about how $\cb{P_t}$ might change outside the experimental window. By conditioning on the realized sequence of environmental conditions but averaging over the operational factors that the policy itself influences, GATE aligns naturally with the causal question that switchback experiments are designed to answer.\footnote{Mathematically, the random outcome noise $\epsilon_t$ could be absorbed into the time-indexed outcome function without affecting the numerical value of the estimand. However, doing so changes the philosophical interpretation of the estimand, and specifically, what sources of randomness are being averaged over. This also alters the variance structure of the estimator, which may in turn lead to a different optimal design.}
\end{exam}

\fi

\begin{rema}
Outcomes in temporal settings are inherently dynamic, but using a fully dynamic framework typically forces the decision-maker to impose strong assumptions about how the system evolves over time. In contrast, traditional finite-population, design-based approaches treat potential outcomes as fixed and attribute all randomness to the treatment assignment; this ensures robustness to adversarial changes in the potential outcomes \citep{neyman1923applications}.

We extend this idea to a dynamic setting by conditioning on the sequence of environmental conditions (equivalently, the sequence of transition kernels) while still allowing randomness to arise from the latent state transitions and outcome generation. In this way, the framework remains robust to adversarial changes in the environment while capturing the stochastic temporal dependence inherent in dynamic systems. As will become clearer later, this hybrid perspective enables robust design-based inference while accommodating the flexibility needed in non-stationary operational settings.

Our semi-design-based formulation is conceptually related to the finite-population time-series framework of \cite{bojinov2019time} and \cite{bojinov2021panel}, in that both condition on an exogenous sequence (potential outcomes versus environment transitions). However, the goals differ largely: their lag-p dynamic effect compares specific, finite windows of treatment paths, whereas the GATE characterizes the long-run distributional effect of policy-level interventions. Under our framework, when the system mixes sufficiently quickly and $p$ is chosen large enough, the two estimands can be numerically close.
\end{rema}

\section{Difference-in-Means Estimation}
\label{sec:regular_estimation}

In this section, we develop estimation guarantees for the estimands in Definition \ref{defi:estimands} using switchback experiments. We begin by formalizing the switchback design and the associated natural treatment effect estimator.

\begin{defi}
The \emph{regular Bernoulli switchback design} is characterized by a block-length $l > 1$ and a time horizon $T$, and
assigns treatment as
\begin{equation}
\label{eq:bern}
W_t \cond W_{t-1} = \begin{cases}
\text{Bernoulli}(0.5) & \text{if $t = (i-1)l + 1$ for some $i = 1, \, \ldots, \, \lfloor T/l \rfloor$,} \\
W_{t-1} & \text{else.}
\end{cases}
\end{equation}
For convenience, we write $Z_i = W_{(i-1)l + 1}$ for the treatment given to the $i$-th block,
and write $k = \lfloor T/l \rfloor$ for the total number of blocks.
The \emph{regular switchback estimator} takes data from a regular switchback along with an (optional) burn-in
length $0 \leq b < l$, and estimates the  overall treatment effect via a difference-in-means that
discards burn-in periods (if a non-zero burn-in length is used),\footnote{We follow the convention to set $0/0$ to be $0$.}
\begin{equation}
\label{eq:DM}
\begin{split}
&\htau_{\text{DM}}^{(l,b)}=\frac{1}{k_1(l - b)}\sum_{\cb{i : Z_i = 1}} \sum_{s = b + 1}^l Y_{(i - 1)l + s} - 
\frac{1}{k_0(l - b)}\sum_{\cb{i : Z_i = 0}} \sum_{s = b + 1}^l Y_{(i - 1)l + s}, \\
&k_w = \abs{\cb{i = 1, \, \ldots, \, \lfloor T/l \rfloor : Z_i = w}}.
\end{split}
\end{equation}
\label{defi:switchback}
\end{defi}

Our definition of the regular switchback design and estimator is closely related to the one used in \citet{bojinov2020design},
in that we both consider experiments that may randomly switch treatment according to Bernoulli draws at pre-determined time-points,
and both consider estimators that discard observations right after switches to mitigate bias.
However, our specification of the estimator \eqref{eq:DM} differs from the one used in \citet{bojinov2020design} in that it
is purely algorithmic: the estimator \smash{$\htau_{\text{DM}}^{(l,b)}$} does not depend explicitly on the model
(Assumptions \ref{assu:mdp} and \ref{assu:mixing}); rather, it is a function of the block length $l$ and burn-in period $b$
that one could seek to justify from a number of perspectives. In contrast, \citet{bojinov2020design} start by specifying conditions
under which carryover effects are guaranteed to vanish, and then consider estimation using the natural Horvitz-Thompson estimator
induced by their modeling assumptions. We note that in our setting, i.e., under Assumptions \ref{assu:mdp} and \ref{assu:mixing},
carryover effects may never fully vanish---they simply decay at an exponential rate. Thus, a literal application of the construction of
\citet{bojinov2020design} would not be consistent, because no data could ever be used after the first switch.
Beyond the work of \citet{bojinov2020design}, we are not aware of formal studies of switchback experiments prior to the initial version of this paper that consider using
burn-in periods or other analogous estimation techniques.\footnote{Such techniques have
been regularly deployed in applied work, even going back to \citet{brandt1938tests} who
eliminates data right after switches to escape ``transition effects''.}
 
Our first result is an error decomposition for regular switchback estimators under our geometric mixing assumptions.
The estimator \smash{$\htau_{\text{DM}}^{(l,b)}$} has two sources of bias. First, it has bias due to the long-term effect
carryover from the past treatments. As can be observed from Figure \ref{figure:illustration}, when the treatment condition
switches from $w$ to $1-w$, the curve will take time to converge from the mean outcome curve under $w$ to the mean outcome curve under $1-w$.
Thus, there is always a bias due to the mixing of the process after each switchback, and the closer it is to the switch point, the larger the bias
we will encounter in estimating $\tau_t$. The use of burn-in periods can mitigate carryover bias---but this results in a different
source of potential bias due to ignoring outcomes in the burn-in periods.


\begin{theo}
Under Assumptions \ref{assu:mdp} and \ref{assu:mixing}, suppose in addition that the conditional expectations are
bounded almost surely by a constant $\Lambda$ such that
\[\abs{\EE{Y_t|S_t,W_t}}= \abs{\mu_t(S_t, W_t)}\le_{a.s.}\Lambda, \]
and that the heterogeneity in treatment effects is bounded by a constant $\Psi$, i.e.,  
\[\abs{\tau_t-\tau_s}\le \Psi \]
for arbitrary $t,s\in\{1,\dots,T\}$. Given the block length $l$ and the burn-in period $b$, define $k=k_1+k_0=\lfloor T/l \rfloor$, and let
\begin{equation}
\label{eq:lbfate}
\tau_{\FATE}^{(l, b)} := 
\frac{1}{k(l-b)}\sum_{i=1}^k \sum_{s=b+1}^l \tau_{(i-1)l+s}
\end{equation}
denote the Filtered Average Treatment Effect (FATE) that only considers non-burn-in periods. Then,
the bias of the regular switchback estimator $\htau_{\text{DM}}^{(l,b)}$ can be bounded as
\begin{equation}
\abs{\EE{\htau_{\text{DM}}^{(l,b)}}-\tau_{\FATE}^{(l, b)}}\leq  \underbrace{\frac{4\Lambda}{1-\exp\p{-1/t_\mix}}\cdot\frac{\exp\p{-b/t_{\text{mix}} }}{l-b}}_{\text{mixing bias}} + \oo\p{2^{-k}},
\end{equation}
and furthermore
\begin{equation}
\abs{\tau_{\FATE}^{(l, b)} - \tau_{\GATE}} \leq \underbrace{\Psi \, b/l}_{\text{burn-in bias}}.
\end{equation}
Meanwhile, the variance of \smash{$\htau_{\text{DM}}^{(l,b)}$} can be bounded as
\begin{equation}
\begin{split}
\Var{\htau_{\text{DM}}^{(l,b)}}&\le
\frac{16\Lambda^2}{k}
+\frac{4\sigma^2}{k(l-b)}
+\frac{16\Lambda^2\exp\p{-b/t_\mix}}{k(l-b)(1-\exp\p{-1/t_\mix})}
+\oo\p{\frac{1}{k^2}},
\end{split}
\end{equation}
where the three leading terms in the above bound represent the variance associated with
the clustered treatment assignment,
the unpredictable outcome noise,
and covariance due to carryover effects respectively.
\label{theo:GATE_mse}
\end{theo}

Theorem \ref{theo:GATE_mse} highlights the tradeoffs introduced by burn-in periods in the presence of carryover effects.
The upper bound on bias reflects how burn-ins simultaneously mitigate mixing bias, bias caused by switchbacks generating treatment histories inconsistent with the always-treat and never-treat regimes we aim to evaluate, while at the same time inducing burn-in bias, since discarding early observations implicitly replaces $\tau_{t}$ for $t\in\text{burn-in}$
with $\tau_{t}$ for $t\notin\text{burn-in}$.
The existence of the carryover effect also inflates the variance by inducing a positive correlation between mean outcomes observed at different time points.
In addition, the variance depends on the noise level in the outcomes and on how clustered the treatment assignments are.
Nevertheless, as long as $l(T)-b(T)$ is bounded away from zero, the variance will be dominated by the term from the clustered assignment; this mirrors
results available in the context of generic cluster-randomized experiments \citep{leung2021rate}.

Having characterized how burn-in periods and block length affect the bias and variance of the DM estimator, we now turn to the practical question of choosing $l$ and $b$ to optimize estimation accuracy. One might expect burn-in periods to be beneficial; indeed, they are widely used in industry to mitigate carryover effects in switchback experiments. However, for the purpose of estimating GATE, we show that this intuition is misleading: A rate-optimal strategy is to set $b=0$ and rely exclusively on long blocks to control mixing bias, except in cases where the burn-in bias is negligible.

\begin{coro}
Under the conditions of Theorem \ref{theo:GATE_mse}, with the choice that
\begin{equation}
l \asymp T^{1/3}\qquad \text{and} \qquad b=0,
\end{equation}
$\htau_{\text{DM}}^{(l,b)}$ achieves the error bound
\begin{equation}
\EE{\p{\htau_{\text{DM}}^{(l,b)}-\tau_{\GATE}}^2} = \oo\p{T^{-2/3}}.
\end{equation}
Furthermore, no regular switchback estimator as given in Definition \ref{defi:switchback} can guarantee
a faster-than-$\oo\p{ T^{-2/3}}$ rate of convergence in this setting.
\label{coro:order_combined1}
\end{coro}

\begin{coro}
Under the conditions of Theorem \ref{theo:GATE_mse}, if the treatment effect is sufficiently stable so that
\begin{equation}
\abs{\tau_{\FATE}^{(l, b)} - \tau_{\GATE}} = \smallO(\sqrt{\log T\cdot T^{-1}} )
\label{eq:stable_te}
\end{equation}
for the chosen design $(l,b)$,
then for any bounded constant $C_1>0$, with the choice that
\begin{equation}
l= b+C_1 \quad \text{and} \quad b = \frac{t_\mix}{2} \log T,
\label{eq:choices_filtered}
\end{equation}
$\htau_{\text{DM}}^{(l,b)}$ achieves the error bound
\begin{equation}
\EE{\p{\htau_{\text{DM}}^{(l,b)}-\tau_\GATE}^2}= \oo\p{ \log T\cdot T^{-1}}.
\end{equation}
\label{coro:order_combined2}
\end{coro}

Given side by side, these two results present a comprehensive characterization of the role of burn-in periods in 
regular Bernoulli switchbacks as specified in Definition \ref{defi:switchback}. When the treatment effects inside and outside burn-in periods differ substantially, the burn-in bias from using a non-trivial $b$ grows faster than the carryover bias decays, thus making the choice of $b = 0$ (i.e., no burn-in period) rate optimal. Consequently, the optimal strategy for estimating $\tau_{\GATE}$ using regular Bernoulli switchbacks is then to use a very long block length of $T^{1/3}$, resulting in an $\oo(T^{-1/3})$ rate of convergence in root-mean squared error. 
In contrast, when the environment is sufficiently stable so that the burn-in bias is negligible (e.g., when the environment inside and outside the burn-in period behaves similarly, or when the decision-maker adopts a super-population perspective in which environmental variation is regarded as random rather than systematic), using a non-trivial burn-in period becomes beneficial: it removes the leading carryover bias and allows the estimator to achieve an $\oo(T^{-1/2})$ rate up to logarithmic factors.
In the supplemental material, we provide further characterization of optimal designs for the DM estimator and discuss settings where $\tau_{\FATE}^{(l, b)}$ is close to $\tau_{\GATE}$ or can be regarded as a reasonable estimand for downstream decision-making.

We end this section by giving conditions where the error bounds from Theorem \ref{theo:GATE_mse}
can be sharpened into a central limit theorem. This provides an asymptotic theory for semi-design-based inference in dynamic systems, which has remained largely unexplored in the literature. To state our result, we use the following
notation: For all blocks $i=1,\dots,k$, let
\begin{equation}
\label{eq:bY_def}
\bY_i=\frac{1}{l-b}\sum_{t=b+1}^l Y_{(i-1)l+t},
\end{equation}
where as usual $l$ denotes the block length.
Furthermore, 
for $w = 0, \, 1$, let $\bY_i(w)$ 
denote
the potential outcome that would be observed for block $i$ if its treatment assignment were set to $Z_i=w$
while holding the rest of the treatment history unchanged.
We also write
\begin{equation}
\label{eq:mu_M_def}
\bmu_i(w) = \frac{1}{l-b}\sum_{t=b+1}^l\EE[\law_w^{(i-1)l+t}]{Y_{(i-1)l+t}}, \ \ \ \ \bM_i(w) = \EE{\bY_i(w)},
\end{equation}
i.e., $\bmu_i(w)$ would be the expectation of $\bY_i$ in a system
that always receives treatment $w$, and $\bM_i(w)$ is
the expectation of $\bY_i$ in our switchback given $Z_i = w$.
We note that, following \eqref{eq:lbfate}, we have
\begin{equation}
\label{eq:FATE_recall}
\tau_{\FATE}^{(l, b)} = \frac{1}{k} \sum_{i = 1}^k \p{\bmu_i(1) - \bmu_i(0)}.
\end{equation}
Finally, for simplicity of exposition we assume that, although our system is non-stationary, relevant empirical variances converge in large samples; similar assumptions are often made in finite population central limit theorems \citep[e.g.,][]{aronow2017estimating,lin2013agnostic}.

\begin{assu}
\label{assu:mom_limits}
We consider a regular switchback with $l$ and $b$ chosen such that
the following limits exist:
\begin{equation}
\label{eq:mom_limits}
\begin{split}
&\frac{1}{k} \sum_{i = 1}^k \p{\bmu_i(0) - \frac{1}{k} \sum_{j = 1}^k \bmu_j(0)}^2 \rightarrow V_0,\quad\frac{1}{k} \sum_{i = 1}^k \p{\bmu_i(1) - \frac{1}{k} \sum_{j = 1}^k \bmu_j(1)}^2 \rightarrow V_1, \\
&\frac{1}{k} \sum_{i = 1}^k \p{\bmu_i(1)-\frac{1}{k} \sum_{j = 1}^k \bmu_j(1)}\p{\bmu_i(0)-\frac{1}{k} \sum_{j = 1}^k \bmu_j(0)}
\rightarrow V_{01}.
\end{split}
\end{equation}
We also assume that the average within-block variance converges to a finite limit:
\begin{equation}
\Sigma_{\Delta,k}:=\frac{2}{k}\sum_{i=1}^k\p{\Var{\bY_i(1)}+\Var{\bY_i(0)}}\to \Sigma_{\Delta}
\label{delta_variances}
\end{equation}
\end{assu}

\begin{theo}
Under the conditions in Theorem \ref{theo:GATE_mse}, suppose in addition that Assumption \ref{assu:mom_limits} holds, 
that $Y_t$ are
upper bounded such that $\abs{Y_t}\le \Gamma_0$ for some
$\Gamma_0<\infty$, and that there exists a constant $\sigma_0>0$ such
that $\Var{\epsilon_t\cond \ff_t}\ge \sigma_0^2$.
Then provided that we choose the burn-in length $b$ so that $l-b=\oo\p{1}$ and
\begin{equation}
 k \exp\p{-2b/t_{\text{mix}}} \rightarrow 0,
\end{equation}
as $k\to\infty$, we have
\begin{equation}
\sqrt{k} \p{\htau_{\text{DM}}^{(l,b)}-\tau_{\FATE}^{(l, b)}} \to_d \nn\p{0, \, V}, \ \ \
V = V_0 + V_1 +2V_{01} + \Sigma_\Delta.
\label{eq:clt_var}
\end{equation}
\label{theo:clt}
\end{theo}

As in the classical Neyman finite-population causal inference framework \citep{imbens2015causal}, the variance expression in \eqref{eq:clt_var} does not in general admit an unbiased estimator. The challenge arises from the cross-term $V_{01}$, whose estimation would require observing both treatment assignments within the same block; this is impossible by design and therefore unidentifiable from the realized data. Nevertheless, conservative confidence intervals can still be constructed by bounding the cross-term via the Cauchy–Schwarz inequality, which ensures that the resulting intervals have at least the nominal coverage probability. In Section~\ref{sec:jackknife} of the supplemental material, we present a jackknife variance estimator \citep{miller1974jackknife} and show that, under our semi-design-based framework, it converges to such a conservative variance bound.

\section{A Bias-Corrected Estimator}
\label{sec:bias_corrected}

In the previous section, we found that a regular switchback estimator can only achieve an error rate of $\oo\p{T^{-1/3}}$ in estimating $\tau_{\GATE}$.
We also showed that, when considering a design with burn-in periods, it is possible to overcome the challenge brought
by carryover bias and achieve a rate of $\oo(\sqrt{\log T/T})$ if the environment doesn't vary much inside and outside the burn-in periods. 
However, because burn-in periods deterministically exclude observations, introducing burn-ins can also create substantial bias for estimating $\tau_{\GATE}$ in general settings.
This tension raises a natural question: Is it fundamentally hard to get good estimates of $\tau_{\GATE}$
using data collected from a switchback design, or could we achieve better results by analyzing the data
produced by the switchback using a different estimator?

The answer to this last question turns out to be affirmative: We can also achieve an $\oo(\sqrt{\log T/T})$ rate of
convergence for $\tau_{\GATE}$ by using an inverse-propensity-type bias-corrected estimator that takes
data collected from a regular Bernoulli switchback design as its input.
Our proposed bias-corrected estimator utilizes the fact that, with Bernoulli randomization, the assigned
treatment does not actually change at the beginning of every block, and thus observations within burn-in periods need
not always be thrown away. Using shorthand $k_{ww'}=\sum_{i=2}^kI(Z_i = w, \, Z_{i-1} = w')$, we then estimate $\tau_{\GATE}$ as
\begin{equation}
\label{eq:IPW}
\begin{split}
\htau_{\text{BC}}^{(l,b)}=&
\frac{l-b}{l}
\htau_{\text{DM}}^{(l,b)}+
\frac{1}{k_{11}}\sum_{\cb{i: Z_i = Z_{i-1} = 1, i\ge 2}}
\frac{1}{l}\sum_{t = 1}^b Y_{(i - 1)l + t}\\
&\qquad\qquad-\frac{1}{k_{00}}\sum_{\cb{i: Z_i = Z_{i-1} = 0, i\ge 2}}
\frac{1}{l}\sum_{t = 1}^b Y_{(i - 1)l + t}.
\end{split}
\end{equation}
The intuition behind our proposed estimator is as follows.
In a Bernoulli switchback, a switch at the beginning of each block occurs randomly with a probability of $0.5$. Each block in a switchback is divided into two parts: a burn-in part and a focal part.
The observations from the focal part are always included in the estimation with weight $1$, while the observations from
the burn-in part are included with a probability of $0.5$ (when a switch does not occur), and are given a weight of approximately $2$.
Hence, on average, observations from all periods are weighted equally, but those from the initial few time periods after a realized switch are still discarded.
The following result establishes the claimed rate of convergence for $\tau_{\GATE}$ using this bias-corrected estimator.


\begin{theo}
Under the assumptions of Theorem \ref{theo:GATE_mse}, the bias of $\htau_{\text{BC}}^{(l,b)}$ can be bounded as 
\begin{equation}
   \abs{\EE{{\htau_{\text{BC}}^{(l,b)}}}-\tau_{\GATE}}\le \frac{4\Lambda}{1-\exp\p{-1/t_\mix}}\cdot\frac{\exp\p{-b/t_{\text{mix}} }}{l}+\oo\p{\frac{1}{k}}.
\end{equation}
Meanwhile, the variance of $\htau_{\text{BC}}^{(l,b)}$ can be bounded as
\begin{equation}
\begin{split}
\Var{\htau_{\text{BC}}^{(l,b)}}&\le\frac{88\Lambda^2}{k}+\frac{16\sqrt{2}\Lambda^2\exp\p{-l/t_{\rm mix}}}{kl\p{1-\exp\p{-1/t_{\rm mix}}}}+\frac{8\sigma^2}{kl}+\oo\p{\frac1{k^2}}.
\end{split}
\end{equation}
Furthermore, for any bounded constant $C_2>0$, with the choice that 
\begin{equation}
 l=b+C_2 \quad \text{and} \quad b=\frac{t_{\mix}}{2}\log T
\end{equation}
$\htau_{\text{BC}}^{(l,b)}$ achieves the error bound
\begin{equation}
\EE{\p{\htau_{\text{BC}}^{(l,b)}-\tau_\GATE}^2}\le 44\Lambda^2t_\mix\cdot \log T\cdot T^{-1} +\smallO\p{ \log T\cdot T^{-1}}
\end{equation}
in estimating $\tau_\GATE$. 
\label{theo:IPW_mse}
\end{theo}

One question that is left open by Theorem \ref{theo:IPW_mse} is how to choose $l$ and $b$ in practice.
Qualitatively, the insight is that we should use $l \sim \log(T)$, and that ``most'' periods should be devoted to burn-in, i.e., $b = l - C$ for some
constant $C$ that does not scale with $T$. In practice, however, $\log(T)$ may not be materially larger than relevant constants for reasonable values of $T$,
and so the optimal choice of $b$ may be a non-trivial multiple of $l$ (e.g., $b = l/2$ may be reasonable). Furthermore, optimal choices of both $b$ and $l$ depend
on $t_\mix$, which may be challenging to estimate from data; this mirrors the findings of \citet{xiong2022bias}, which highlight the role of prior information
in designing an efficient switchback experiment.
In practice, our main recommendation is that analysts using switchbacks to estimate treatment effects in systems that may
exhibit carryovers should consider using burn-in periods to mitigate carryover bias. Optimal choices of $l$ and $b$ are likely to depend on specifics of the
application setting, and we recommend using a mix of experimental and semi-synthetic validation to pick $l$ and $b$ on a case-by-case basis.

We also provide a central limit theorem analogous to the one given in Theorem \ref{theo:clt}
for the regular switchback estimator. The bias-corrected estimator has a more complex form, and so
we need to require convergence of further moments in addition to Assumption \ref{assu:mom_limits}.


\begin{assu}
For $i=1,\dots,k$, let $\bmu_i(w)$ denote average ``pure treatment'' potential outcomes for
the focal period as in Assumption \ref{assu:mom_limits}, and let
\begin{equation}
\label{eq:mub_def}
\bmu^b_i(w) = \frac{1}{b}\sum_{t=1}^b\EE[\law_w^{(i-1)l+t}]{ Y_{(i-1)l+t}}
\end{equation}
denote the analogous quantity for the burn-in periods.
We assume that the centered second moment of the block-averaged potential outcomes in the burn-in periods ($b$)
and focal periods ($f$) converge as the number of blocks $k$ gets large,
\begin{equation}
\frac{1}{k} \sum_{i = 1}^k
\begin{pmatrix}
\bmu^b_i(w) - \frac{1}{k} \sum_{j = 1}^k \bmu^b_j(w) \\
\bmu_i(w) - \frac{1}{k} \sum_{j = 1}^k \bmu_j(w)
\end{pmatrix}^{\otimes 2}
\rightarrow
\begin{pmatrix}
V_{w}^b & V_{w}^{bf} \\
V_{w}^{bf} & V_{w}^{f}
\end{pmatrix},
\end{equation}
for both $w = 0, \, 1$. We also assume that the associated cross-terms converge in the same limit:
\begin{equation}
\frac{1}{k} \sum_{i = 1}^k
\begin{pmatrix}
\bmu^b_i(0) - \frac{1}{k} \sum_{j = 1}^k \bmu^b_j(0) \\
\bmu_i(0) - \frac{1}{k} \sum_{j = 1}^k \bmu_j(0)
\end{pmatrix}
\begin{pmatrix}
\bmu^b_i(1) - \frac{1}{k} \sum_{j = 1}^k \bmu^b_j(1) \\
\bmu_i(1) - \frac{1}{k} \sum_{j = 1}^k \bmu_j(1)
\end{pmatrix}'
\rightarrow
\begin{pmatrix}
V_{01}^b & V_{01}^{bf} \\
V_{10}^{bf} & V_{01}^{f}
\end{pmatrix}.
\label{eq:clt_var_cross}
\end{equation}
Finally, we assume that all cross products between the block-averaged potential outcomes in one block
and those in the burn-in period of the subsequent block ($s$) also converge,
\begin{equation}
\frac{1}{k-1} \sum_{i = 1}^{k-1}
\begin{pmatrix}
\bmu^b_i(w) - \frac{1}{k} \sum_{j = 1}^k \bmu^b_j(w) \\
\bmu_i(w) - \frac{1}{k} \sum_{j = 1}^k \bmu_j(w)
\end{pmatrix}
\p{\bmu^b_{i+1}(w') - \frac{1}{k} \sum_{j = 1}^k \bmu^b_j(w')}
\rightarrow
\begin{pmatrix}
V_{ww'}^{bs} \\
V_{ww'}^{fs}
\end{pmatrix},
\end{equation}
for all pairs $w, \, w' \in \cb{0,1}$.
\label{assu:cross_limits}
\end{assu}

\begin{theo}
Under the assumptions of Theorem \ref{theo:GATE_mse}, suppose in addition 
that $Y_t$ are
upper bounded such that $\abs{Y_t}\le \Gamma_0$ for some
$\Gamma_0<\infty$, 
and that Assumption \ref{assu:cross_limits} holds.
Then provided that we choose $b$ and $l$ so that $l\to\infty$,
\begin{equation}
\frac{k}{l^2} \exp\p{-2b/t_{\text{mix}}} \rightarrow 0,
\end{equation}
and there exists a constant $\beta\in [0,1]$ such that $b/l\to\beta$
as $k\to\infty$, 
we have
\begin{equation}
\sqrt{k} \p{\htau_{\text{BC}}^{(l,b)}-\tau_{\GATE}} \to_d \nn\p{0, \, \Tilde{V}}
\end{equation}
where 
\begin{equation}
\begin{split}
\label{eq:clt_var_bc}
\Tilde{V}
&= (1-\beta)^2\p{V^f_0+V^f_1+2V^f_{01}}\\
&\qquad\qquad +\beta^2 \p{3V^b_0+3V^b_1+2V^b_{01}+2V^{bs}_{01}+2V^{bs}_{10}+2V^{bs}_{00}+2V^{bs}_{11}}\\
&\qquad\qquad +2\beta(1-\beta)
\p{{V}^{bf}_{0}+{V}^{bf}_{1}+{V}^{bf}_{01}+{V}^{bf}_{10}+{V}^{fs}_{00}+{V}^{fs}_{11}+{V}^{fs}_{01}+{V}^{fs}_{10} }.
\end{split}
\end{equation}
\label{theo:clt_bc}
\end{theo}

Although the form of the asymptotic variance in \eqref{eq:clt_var_bc} may appear complex at first glance, it is in fact the natural analogue of \eqref{eq:clt_var} once the burn-in periods and their inverse-propensity reweighting are incorporated.
As in \eqref{eq:clt_var}, the first term in \eqref{eq:clt_var_bc} corresponds to the variation arising from the focal-period difference-in-means component, scaled by the factor $(1-\beta)^2$ reflecting the proportion of focal periods.
The second term plays a similar role for the burn-in correction, except that this part of the estimator is constructed using an inverse-propensity reweighting, which introduces additional cross-block covariance among consecutive burn-in averages.
Finally, the remaining term captures the covariance between the focal and burn-in components, arising from the fact that the estimator combines information from both parts of the block in a single weighted statistic.

As with the simpler case in \eqref{eq:clt_var}, the asymptotic variance $\Tilde{V}$ in \eqref{eq:clt_var_bc} generally does not admit an unbiased estimator. In particular, several of the cross-terms in \eqref{eq:clt_var_cross} involve covariances that are unobservable under the realized design, and therefore cannot be estimated empirically. As before, conservative confidence intervals can be constructed by applying the Cauchy–Schwarz inequality to bound these unidentified components. Motivated by the jackknife procedure we developed for the difference-in-means estimator and block-resampling methods for correlated data \citep{kunsch1989jackknife}, 
we outline a block jackknife procedure in Section \ref{sec:jackknife} of the supplemental material. In our simulation studies, this block jackknife estimator performs reasonably well in practice, although a full theoretical characterization remains an interesting direction for future investigation.

\section{Numerical Experiments}
\label{sec:sim}

In this section, we show that the asymptotic rates derived in the previous sections
approximate well the optimal rates that can be achieved by the 
estimators in estimating both the global average treatment effect and the filtered average treatment effect
in simulations.

\subsection{Simple Illustration}
\label{subsec:sim_toy}

We start with a very simple setting in which all carryovers are mediated by
a hidden state variable $H_t\in \{0,1,\dots,20 \}$ that evolves according to a random
walk with drift. The evolution of $H_t$ is governed by the assigned treatment, as well as an exogenous market condition
variable $M_t\in\{1,2,3\}$. 
We generate $M_t$ with the following dynamic process:
\begin{equation}
P(M_{t+1}=m)=
\begin{dcases}
2/3,\qquad\text{ if } M_t=m, \\
1/6,\qquad \text{ otherwise }, \\
\end{dcases}
\end{equation}
i.e., the market condition will switch with probability $0.5$, and once it switches, the new market condition is a uniform random draw from $\{1,2,3\}$. 
Conditionally on the sequence of market conditions, the state variable $H_t$ is generated as follows: If $W_t=1$,
\begin{equation}
H_{t+1} = \begin{cases}
\min\cb{H_t+M_t,20} &\qquad \text{ with probability } 0.7, \\
\max\cb{H_t-M_t,0} &\qquad \text{ with probability } 0.3;
\end{cases}
\end{equation}
if $W_t=0$,
\begin{equation}
H_{t+1} = \begin{cases}
\min\cb{H_t+M_t,20} &\qquad \text{ with probability } 0.3, \\
\max\cb{H_t-M_t,0} &\qquad \text{ with probability } 0.7.
\end{cases}
\end{equation}
Given $(H_t, W_t)$, $Y_t$ is then generated as
\begin{equation}
Y_t=H_t+0.5\cdot W_t\cdot H_t +\epsilon_t,
\end{equation}
where $\epsilon_t\sim \nn(0,\sigma^2)$. Throughout, we fix $\sigma=3$, and vary the horizon $T$ from $400$ to $25,600$.

\begin{figure}[t]
\centering
\includegraphics[width=\linewidth]{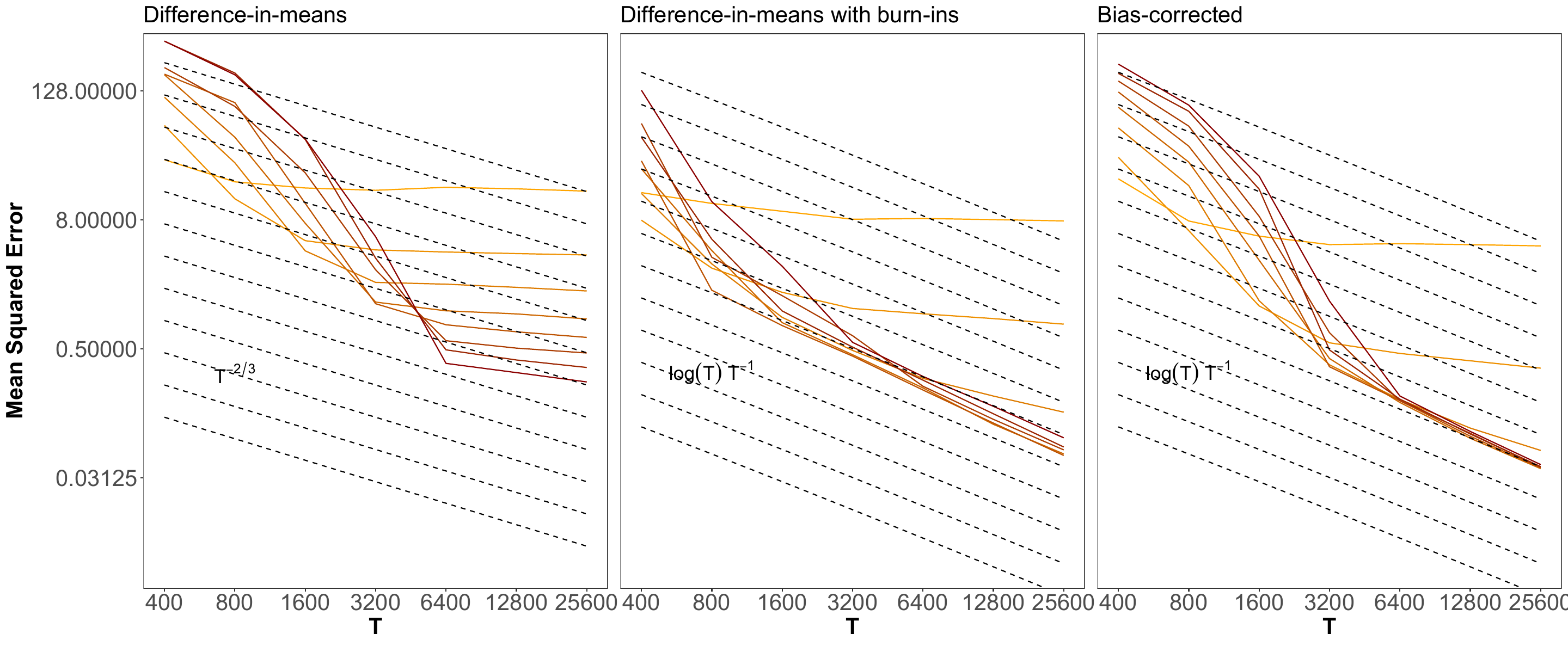}
\caption{MSE as a function of $T$ under different block lengths $l$ on a log scale across $5000$ iterations. The lightest orange corresponds to the case with the smallest $l$, with a gradient to dark red representing $l$ increases from the smallest choice to the largest choice gradually. The black dashed lines indicate the convergence rate, with an order of $T^{-2/3}$ on the left and $\log(T)/T$ in the middle and on the right.}
\label{sim:all}
\end{figure}

First, we evaluate the performance of $\htau_{\text{DM}}^{(l,b)}$ in estimating $\tau_\GATE$. Motivated by Corollary \ref{coro:order_combined1}, we fix $b$ to be always zero, and vary $l$ from $60$ to $480$. The left panel of Figure \ref{sim:all} displays how MSE changes with different choices of $l$ and $T$ on a log scale across $5000$ iterations. In general, $\htau_{\text{DM}}^{(l,0)}$ with a larger $l$ achieves a quicker rate of convergence in estimating $\tau_\GATE$, while $\htau_{\text{DM}}^{(l,0)}$ with a smaller $l$ starts off with a smaller error.
Consequently, as the horizon $T$ gets larger, the optimal $l$ also grows gradually from $l=60$ to $l=480$. Nevertheless, the mean-squared error achieved even with the best choices of $l$ will never surpass a $T^{-2/3}$ rate of convergence in large samples.

In the middle panel of Figure \ref{sim:all}, we plot the performance of \smash{$\htau_{\text{DM}}^{(l,b)}$} in estimating \smash{$\tau_{\FATE}^{(l, b)}$}.
This is equivalent to the setting when the burn-in bias is negligible.
Motivated by Corollary \ref{coro:order_combined2}, we vary $b$ from $10$ to $80$, and consider the set of $l$ such that $l=b+30$. We observe a similar pattern in the relationship between the horizon and the length of the burn-in periods, with the optimal $b$ growing gradually from $10$ to $20$ as $T$ increases. However, we can notice that, when there is no burn-in bias, 
we are able to achieve a much better error rate using the estimator with burn-in periods. This is especially important if the practitioners are planning for an experiment over a relatively long period of time. Furthermore, comparing the plot to the left panel, we notice that the performance is relatively robust to the choice of $b$. We also considered the set of $l$ such that $l=b+50$ and obtained similar results.
Our results thus suggest that, in addition to improving asymptotic behavior relative to standard switchbacks, using burn-in periods
(when there is no burn-in bias) also makes the performance of the experiment more stable across
different choices of tuning parameters.

Next, we investigate the performance of \smash{$\htau_{\text{BC}}^{(l,b)}$} in estimating \smash{$\tau_{\GATE}$}, as represented in the right panel of Figure \ref{sim:all}. 
Once again, we vary $b$ from $10$ to $80$, and consider the set of $l$ such that $l=b+30$. 
The overall performance of \smash{$\htau_{\text{BC}}^{(l,b)}$} in estimating \smash{$\tau_{\GATE}$} closely mirrors that of \smash{$\htau_{\text{DM}}^{(l,b)}$} in estimating \smash{$\tau_{\FATE}^{(l, b)}$}, both achieving an error rate of \smash{$\log(T)\cdot T^{-1}$} in mean-squared error for estimating their targets. 
Overall, we find that the empirically optimal choice of $l$ grows very differently across the estimators. For the difference-in-means estimator without burn-in, the optimal $l$ increases rapidly with $T$: it grows from about $l^*=60$ when $T=400$ to $l^*=480$ when $T=6400$. In contrast, for the estimators with burn-ins, the optimal $l$ grows much more slowly: it increases from roughly $l^*=40$ when $T=400$ to only around $l^*=70$ when $T=6400$. This pattern aligns with our rate results in the preceding sections, where the estimator without burn-ins requires $l^*$ to grow on the order of $T^{1/3}$, while with burn-ins $l^*$ only needs to grow at a logarithmic rate.

{\renewcommand{\arraystretch}{1.25}
\begin{table}[t]
\centering
\begin{tabular}{|c|c|c|c|}
  \hline
 Design & $\htau_{\text{DM}}^{(l,0)}\, (\tau_{\GATE})$  & $\htau_{\text{DM}}^{(l,b)}\, (\tau_{\FATE}^{(l,b)})$ & $\htau_{\text{BC}}^{(l,b)}\, (\tau_{\GATE})$  \\ 
  \hline
 $T = 20000,\,b = 50,\,l = 200$ & 0\% & 94.4\% & 93.3\%  \\ 
  \hline
 $T = 20000,\,b = 100,\,l = 200$  & 0\% & 95.3\% & 92.6\%  \\ 
   \hline
 $T = 20000,\,b = 150,\,l = 200$  & 0\% & 93.1\% & 92.5\%  \\ 
   \hline
\end{tabular}
\caption{Coverage of 95\% confidence intervals across $1,000$ iterations under various designs.}
\label{tab:sim_toy_coverage1}
\end{table}}

Finally, we assess the coverage provided by confidence intervals associated with the three estimators: $\htau_{\text{DM}}^{(l,0)}$, $\htau_{\text{DM}}^{(l,b)}$, and $\htau_{\text{BC}}^{(l,b)}$. Throughout this analysis, we set $T=20,000$ and $l=200$, while varying $b$ from $50$ to $150$. To estimate the variance of the difference-in-means estimators, we employ a jackknife resampling procedure \citep{miller1974jackknife}, which iteratively excludes each block and computes the variance of the estimator based on the remaining observations. For the bias-corrected estimators, we utilize a block-jackknife resampling procedure \citep{kunsch1989jackknife}, which excludes two blocks simultaneously to account for the significant correlation between blocks introduced by weighting with the treatment assigned to the preceding
block.\footnote{We experimented with excluding different numbers of blocks in the block-jackknife estimator, but did not find this to meaningfully affect performance here.} We present details of the jackknife variance estimators in Section \ref{sec:jackknife} of the supplemental material.

The coverage of 95\% confidence intervals around the three estimators can be found in Table \ref{tab:sim_toy_coverage1}.
We immediately see that the difference-in-means estimator without burn-in periods has zero coverage for $\tau_{\GATE}$. 
This anomaly arises because of the slow mixing of the Markov chain under study, leading to substantial bias;
see Table \ref{tab:sim_toy_decomp1} in the supplemental material.
In contrast, the two estimators employing a burn-in period exhibit significantly improved performance and provide
reasonable coverage despite the presence of a substantial carryover effect.
Moreover, these estimators also demonstrate robustness to variations in the length of the burn-in period.
The bias-corrected estimator shows slight undercoverage; however, as seen in Table \ref{tab:sim_toy_decomp1},
this does not appear to be driven by bias.
The block-jackknife variance estimate can be quite variable, which may contribute to the observed undercoverage (see Figure~\ref{fig:toy_var_ratio}).
In Section \ref{subsec:toy_res} of the supplemental material we also consider an analogous-but-easier
setting with faster mixing, and find that the coverage of all three estimators stays closer to the nominal level.

\subsection{A Ride-Sharing Simulator}
\label{subsec:sim_taxi}

We also evaluated our estimators using a large-scale ride-sharing simulator adapted from \citet{farias2022markovian}. The simulator generates drivers and riders based on data from the NYC taxi trip records dataset \citep{tlcrecord}. In this simulator, drivers enter the system continuously, each with a fixed capacity of 3 riders. Their initial positions are randomly selected from the trip records dataset, and the duration of their shifts follows an exponential distribution. Once a driver completes their shift, they go offline.

At any given time, a rider may initiate a ride request, with pick-up and drop-off locations randomly drawn from the trip records dataset and their value-of-time parameter (which informs offer acceptance) drawn from a lognormal distribution. When a request is initiated, the system dispatches a driver according to the current dispatching policy. Upon dispatch, an offer is made to the rider based on the expected cost of the dispatched driver's service. If the dispatched driver is a pool driver, an additional discount of 50\% is applied to the offer. The rider then compares the offer with their outside option before deciding whether to accept. If accepted, the dispatched driver's route is updated accordingly. 

As in \citet{farias2022markovian}, we experiment on dispatching policies. Specifically, we examine a class of policies that determine whether to dispatch a pool driver based on the cost comparison with an idle driver.
In these policies, a pool driver is dispatched only if the cost of the candidate pool driver is less than $\theta_d$ times the cost of the candidate idle driver, where $\theta_d$ is a threshold parameter determining the dispatching decision.
We investigate two policies corresponding to different values of $\theta_d$: $\theta_d=0.5$ (treatment arm, $W_t=1$), indicating a more stringent dispatching criterion, and $\theta_d=1$ (control arm, $W_t=0$), representing a less stringent criterion.
Upon completion of the ride, the rider makes a payment, and the ride-sharing company incurs a cost.
The time horizon is discretized into 400-second intervals.
The outcome of interest we study is the aggregated profit, which is defined as the total difference between the price charged to riders and the cost incurred by the ride-sharing company for completed requests.
We refer readers to Section \ref{subsec:ridesharing_details} of the supplemental material for more detailed information about the ride-sharing simulator.

The ride-sharing system can be naturally modeled as an MDP, with a large latent state space $S_t$ encompassing
drivers' positions and routes, potential rider's locations, traffic conditions, riders' willingness-to-pay, outside options, and more.
The outcome $Y_t$, which is the aggregated profit, is a function of both the current treatment $W_t$ and the unobserved state variable $S_t$.
The dispatch policy can have a long-term effect by altering the state variable: If a change in the dispatching policy results
in a different driver being dispatched to a request, it will impact this driver's route, subsequently affecting their
future position and the requests they will be dispatched to. This carryover effect may persist indefinitely and may
also influence nearby drivers.

\begin{figure}[t]
\centering
\includegraphics[width=\linewidth]{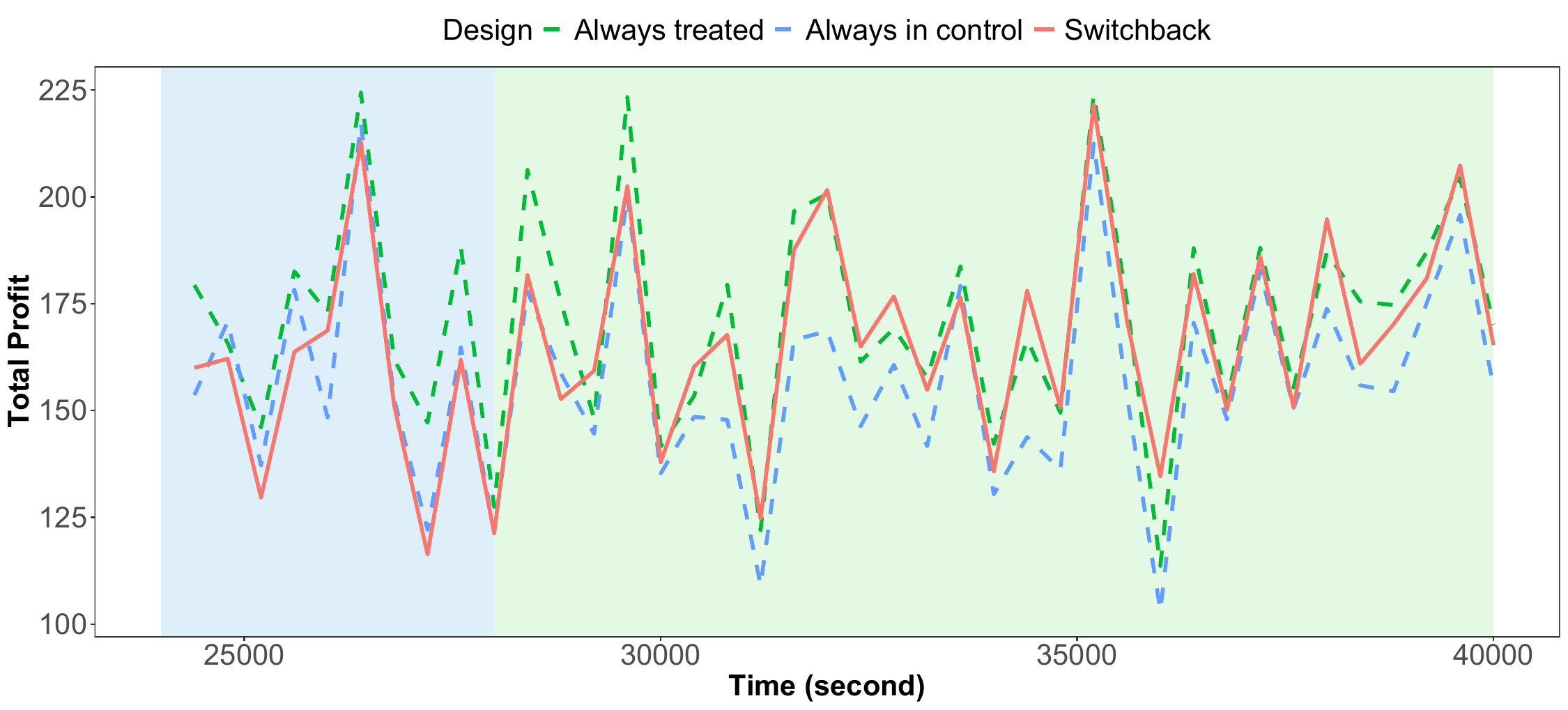}
\caption{A sample trajectory of aggregated profit from the ride-sharing simulator under a fixed environment. The red solid curve represents the profit observed in a switchback experiment, while the green and blue curves represent the profit observed in counterfactual scenarios where the treatment assignment is fixed to always treated and always in control, respectively. The green-shaded region indicates the treatment block in the switchback experiment, while the blue-shaded region indicates the control block in the switchback experiment.
}
\label{figure:ridesharing_illustration}
\end{figure}

We illustrate these carryovers in Figure \ref{figure:ridesharing_illustration}. To this end, we run three coupled simulators,
one under the never-treat policy, one under an always-treat policy, and one with a switchback experiment. Each of the coupled
simulators has the same random realizations of driver arrivals and rider requests, but then uses different dispatch policies. 
The figure displays aggregated profit over time from these 3 coupled simulators across a single switch from control to treatment.
We see that the red solid curve representing the observed-in-switchback outcomes starts by approximating the never-treat trajectory
before the switch; then, after the switch, it slowly diverges from it and eventually starts to approximate the always-treat trajectory.
The coupling doesn't become perfect, though, since the effect of past dispatch decisions on the current system state
never fully vanishes on any reasonable timescale.
To further demonstrate the mixing behavior illustrated in Figure~\ref{figure:ridesharing_illustration}, we compare the outcome distributions of two systems initialized from different treatment histories when exposed to the same treatment policy. The two distributions become progressively closer over time; see Section~\ref{sec:outcome_mixing_sim} of the supplemental material.



\begin{figure}[t]
\centering
\includegraphics[width=\linewidth]{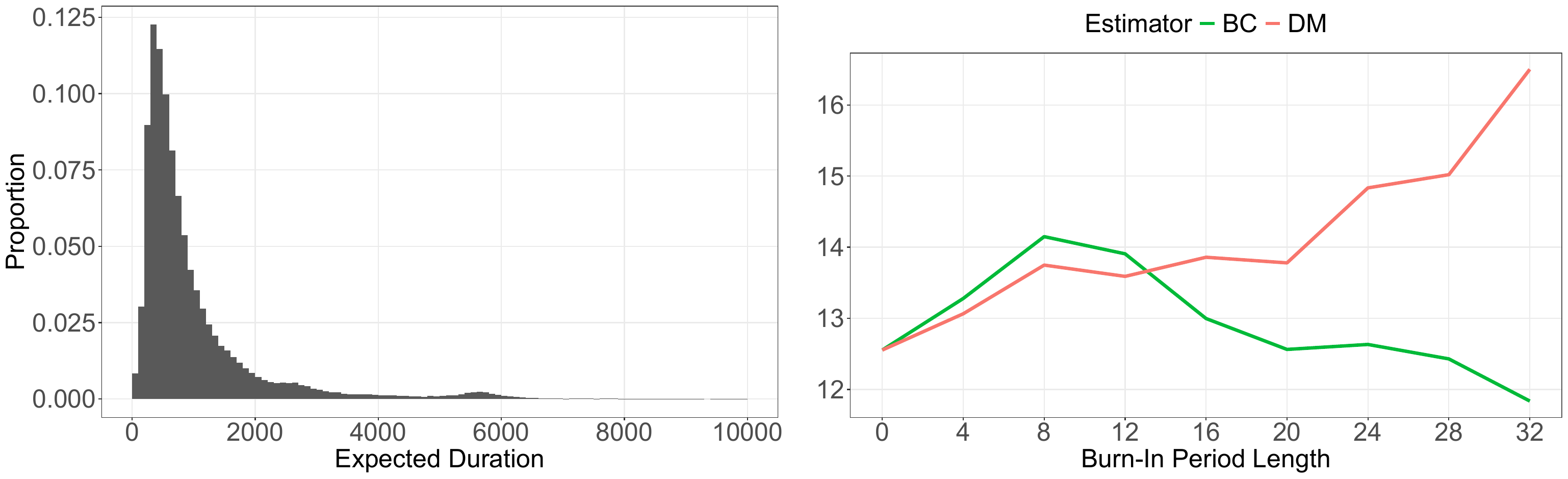}
\caption{Left: Distribution of ride durations in a historical always-control dataset (in seconds). Right: Behavior of the estimators $\hat{\tau}_{\text{DM}}^{(40,b)}$ and $\hat{\tau}_{\text{BC}}^{(40,b)}$ as the burn-in period $b$ varies. Superscripts are given in 100s of seconds.
}
\label{figure:ride_diagnosis}
\end{figure}

To choose the block length for the switchback experiment, we examine the distribution of the rides’ expected durations using a historical always-control dataset containing approximately 400,000 seconds of operations. As shown in the left panel of Figure \ref{figure:ride_diagnosis}, the majority of rides are completed within 4,000–6,000 seconds. This suggests that a block length of roughly 4,000 seconds should be sufficient to ensure that the carryover effect from the previous treatment assignment has largely dissipated. In Section \ref{subsec:ridesharing_res} of the supplemental material, sensitivity analyses show that experiments using a shorter block length of 2,000 seconds perform substantially worse, while a longer block length of 6,000 seconds performs only slightly better and at a comparable overall level. In practice, we recommend selecting the block length conservatively based on domain knowledge (e.g., the typical ride duration) so that any cross-block carryover effects are essentially negligible.

With a block length of 4,000 seconds ($l=40$), we run the experiments for a total duration of 400,000 seconds ($k = 100$). We examine three estimators: the difference-in-means estimator $\hat{\tau}_{\text{DM}}^{(l,0)}$, the difference-in-means estimator with burn-in periods $\hat{\tau}_{\text{DM}}^{(l,b)}$, and the bias-corrected estimator $\hat{\tau}_{\text{BC}}^{(l,b)}$. 
To guide the choice of the burn-in period length $b$, we inspect how $\hat{\tau}_{\text{DM}}^{(l,b)}$ and $\hat{\tau}_{\text{BC}}^{(l,b)}$ change as $b$ increases using a single experimental trajectory. 
From the right panel of Figure \ref{figure:ride_diagnosis}, both estimators increase together for small $b$, which is likely due to the common reduction in carryover bias. 
The estimators stabilize around $b\approx 8$.
Beyond this point, however, they begin to behave very differently as $b$ continues to increase, which is likely due to the dominance of variance. This is unsurprising, as the outcomes themselves are highly noisy (Figure~\ref{figure:ridesharing_illustration}). These observations motivate choosing a burn-in period of roughly $b=8$. 
This choice is also consistent with the mixing behavior discussed in Section~\ref{sec:outcome_mixing_sim} of the supplemental material: substantial history dependence remains at 800 seconds, but decreases considerably thereafter.
In practice, this inspection can be performed using a previous experiment with treatments that induce similar mixing behavior, while applying it to the same trajectory used for analysis could invalidate inference and should therefore be avoided when confidence statements are desired.

To provide an additional point of comparison, we also evaluate a Horvitz–Thompson-type estimator in the spirit of \cite{bojinov2020design} (Equation (4) of \cite{bojinov2020design} when $p=m=l$ in our setup). To stabilize the weights and enable a reasonable comparison, we consider a self-normalized (H\'ajek) version of this estimator, which we denote as $\htau_{\text{H}}^{(l)}$.
This self-normalized Horvitz–Thompson estimator corresponds exactly to our bias-corrected estimator with $b=l$, since both rely exclusively on observations occurring when two consecutive blocks are assigned the same treatment.
It is also closely related to the classical carryover adjustment of \cite{cochran1941double}, which regresses outcomes on the current and previous treatments. 

{\renewcommand{\arraystretch}{1.25}
\begin{table}[t]
\centering
\begin{tabular}{|c|c|c|c|c|}
  \hline
  Estimator & Bias  & Standard Error & Mean Squared Error & Coverage (Monte Carlo SE)  \\ 
  \hline
  $\htau_{\text{DM}}^{(40,0)}$ & 2.64 & 2.60 & 13.65 & 78\% (4.1\%) \\ 
  \hline
    \hline
  $\htau_{\text{DM}}^{(40,8)}$ & 0.56 & 2.69 & 7.46 & 92\% (2.7\%) \\
  \hline
  $\htau_{\text{DM}}^{(40,16)}$ & 0.16 & 2.88 & 8.26 & 90\% (3.0\%) \\
  \hline
    \hline
  $\htau_{\text{BC}}^{(40,8)}$ & 0.40 & 2.75 & 7.66 & 92\% (2.7\%) \\
  \hline
  $\htau_{\text{BC}}^{(40,16)}$ & 0.09 & 3.09 & 9.47 & 92\% (2.7\%) \\
  \hline
    \hline
  $\htau_{\text{H}}^{(20)}$ & 0.37 & 3.46 & 11.99 & ---  \\
    \hline
  $\htau_{\text{H}}^{(40)}$ & 0.48 & 3.75 & 14.11 &  --- \\
  \hline
  $\htau_{\text{H}}^{(60)}$ & 0.13 & 3.08 & 9.38 &  --- \\
    \hline
\end{tabular}
\caption{Bias, standard error, mean squared error, and coverage of 95\% confidence intervals across $100$ iterations.
For readability, the superscripts for all estimators are given in 100s of seconds.}
\label{tab:sim_taxi_res}
\end{table}
}

Table \ref{tab:sim_taxi_res} summarizes the performance of these estimators for the two burn-in choices $b=8$ and $b=16$, together with the Hájek benchmark at $l=20,40,60$. As expected, we notice that the use of burn-in periods considerably reduces the bias of the treatment effect estimation. Although the variance of the estimators with burn-in periods increases, the decrease in bias still results in a large decrease in the overall mean squared error. Furthermore, including burn-in periods leads to more accurate inferential results, with confidence intervals having coverage substantially closer to the nominal level. Again, we observe that the performance of the estimators with burn-in periods remains relatively stable regardless of the chosen length of the burn-in period. 
The Hájek estimator $\htau_{\text{H}}^{(l)}$ performs reasonably well once $l$ is sufficiently large. However, in our setting where carryover effects never fully vanish, using burn-in periods with $b<l$ typically yields additional improvements in performance.
In Section \ref{subsec:ridesharing_res} of the supplemental material, we report additional results examining the performance of the estimators under a wider range of $(l,b)$ choices. Overall, the choice suggested by the heuristics in Figure \ref{figure:ride_diagnosis} is never much worse than any of the alternatives we tried, and the results are fairly stable across all reasonable choices of $(l,b)$. If the practitioner has access to a simulator similar to ours, another option for selecting $(l,b)$ in practice is to replicate this type of simulation study and simply choose the design that performs best in the simulated environment.

\section{Discussion}

We studied switchback experiments under a generic, time-heterogeneous Markovian model for carryover effects.
We found that, under this model, regular switchback estimators as they are often implemented in practice---that
is, without any burn-in periods---suffer from a severe bias problem that limits the best attainable rate of
convergence. We also showed that there exist a number of practical solutions to this bias problem.
If researchers are willing to change their statistical target to a filtered average treatment effect, then
the bias problem can be side-stepped by using burn-in periods with a regular switchback estimator. Meanwhile,
researchers who want to target the global average treatment effect can achieve good performance by using
a modified, bias-corrected estimator to process the data generated from a Bernoulli switchback design.

More broadly, our results suggest promise in using dynamic stochastic modeling techniques to understand and
improve popular tools for causal inference. The switchback is a simple and intuitive experimental design
that can be used without explicitly writing down a stochastic model. However, we found that modeling the
underlying system as a Markov decision process enabled us to get new insights on how and why switchback
experiments enable accurate causal inference---and these insights can then be put into practice without
requiring researchers to make large changes to their analytic approach (and, in particular, without requiring
researchers to fit a Markov decision process). It is likely that we could also improve our understanding
of other (at face value model-free) causal estimators via a stochastic modeling approach.

\ifms

\bibliographystyle{informs2014}
\bibliography{references}

\newpage

\begin{center}
{\Large\bfseries Switchback Experiments under Geometric Mixing}

\vspace{0.5em}

{\large Yuchen Hu \qquad Stefan Wager}

\end{center}

\vspace{1em}

\begin{APPENDIX}{Supplemental Material}

\else

\section*{Funding}

This research was supported by NSF grant SES-2242876.

\bibliographystyle{plainnat}
\bibliography{references}

\newpage

\appendix
\begin{center}
\textbf{\Large Supplemental Material} \\ 
\end{center}

\fi

\setcounter{equation}{0}
\setcounter{figure}{0}
\setcounter{table}{0}
\setcounter{page}{1}
\makeatletter
\renewcommand{\theequation}{S\arabic{equation}}
\renewcommand{\thefigure}{S\arabic{figure}}
\renewcommand{\thetable}{S\arabic{table}}
\renewcommand{\bibnumfmt}[1]{[S#1]}
\renewcommand{\citenumfont}[1]{S#1}

\section{Additional Details on the Filtered Average Treatment Effect}
\label{supp:fate}

We start by introducing a general definition of the filtered average treatment effect for any chosen set of time points.

\begin{defi}
Under Assumption \ref{assu:mdp},  for any set of time points $\mathcal{I} \subset \{1, \, \ldots, \, T\}$, the filtered average treatment effect (FATE) is
\begin{equation}
\tau_{\FATE}(\mathcal{I}) = \frac{1}{\abs{\mathcal{I}}}\sum_{t \in \mathcal{I}} \tau_t.
\label{eq:estimand_fate}
\end{equation}
\label{defi:estimand_fate}
\end{defi}

This family of estimands provides flexibility in specifying the causal target, and the quantity used in Theorem \ref{theo:GATE_mse}, $\tau_{\FATE}^{(l, b)}$, is a special case.
In particular, $\tau_{\FATE}^{(l, b)}$ allows us to disregard time periods in which the DM estimator is biased due to carryover effects. In other situations, however, a FATE might be of interest in its own right; for example, one might wish to estimate causal effects separately for daytime and nighttime periods.

Below, we provide full details of Corollaries \ref{coro:order_combined1} and \ref{coro:order_combined2} presented in Section \ref{sec:regular_estimation}. This gives guidance on how to choose the switchback block length $l$ and burn-in length $b$ such as to make the error
guarantees for regular switchback estimators obtained in Theorem \ref{theo:GATE_mse} as good as possible;
here, we focus on optimizing mean-squared error. In doing so, however, we
first need to specify our target---do we want to target $\tau_{\GATE}$, or is it acceptable to target \smash{$\tau_{\FATE}^{(l, b)}$} from
\eqref{eq:lbfate} instead? We first give results for both targets below, and then follow up with a discussion of trade-offs.

\begin{coro}
Under the conditions in Theorem \ref{theo:GATE_mse}, with the choice that
\begin{equation}
l= {\frac{(4/3)^{1/3}}{(1-e^{-1/t_\mix})^{2/3}}}T^{1/3}\qquad \text{and} \qquad b=0,
\end{equation}
$\htau_{\text{DM}}^{(l,b)}$ achieves the error bound
\begin{equation}
\EE{\p{\htau_{\text{DM}}^{(l,b)}-\tau_{\GATE}}^2}\le\frac{48^{2/3}\Lambda^2}{(1-e^{-1/t_\mix})^{2/3}}T^{-2/3}+\smallO\p{ T^{-2/3}}
\end{equation}
in estimating $\tau_\GATE$. 
Furthermore, no regular switchback estimator as given in Definition \ref{defi:switchback} can guarantee
a faster-than-$\oo\p{ T^{-2/3}}$ rate of convergence in this setting.
\label{coro:design}
\end{coro}

\begin{coro}
Under the conditions in Theorem \ref{theo:GATE_mse}, for any bounded constant $C_1>0$, with the choice that
\begin{equation}
l= b+C_1 \quad \text{and} \quad b=\frac{t_\mix}{2}\log T,
\end{equation}
$\htau_{\text{DM}}^{l,b}$ achieves the error bound
\begin{equation}
\EE{\p{\htau_{\text{DM}}^{(l,b)}-\tau^{(l,b)}_\FATE}^2}\le
\p{6\Lambda^2+\frac{2\sigma^2}{C_1}}\cdot t_\mix\cdot \log T\cdot T^{-1} +\smallO\p{ \log T\cdot T^{-1}}
\end{equation}
in estimating $\tau^{(l,b)}_\FATE$. 
\label{coro:mse_filtered}
\end{coro}

These two corollaries make precise how burn-in periods trade off against block length in shaping the error properties of regular Bernoulli switchbacks. If we insist on targeting $\tau_{\GATE}$,
then carryover bias matters, but we cannot use burn-ins to ameliorate the situation. The problem is, essentially, that
the burn-in bias from using a non-trivial $b$ grows faster than the carryover bias decays, thus making the choice of $b = 0$
rate optimal. The solution for estimating $\tau_{\GATE}$ using regular Bernoulli switchbacks is then to use a very long block length
of $T^{1/3}$, resulting in an $\oo(T^{-1/3})$ rate of convergence in root-mean squared error.

On the other hand, if we are willing to consider \smash{$\tau_{\FATE}^{(l, b)}$} as our target estimand, the situation becomes
much better. This choice of estimand eliminates burn-in bias, and we can then make aggressive use of burn-in periods to achieve
a substantially better rate of convergence. In Corollary \ref{coro:mse_filtered}, we are able to reach the strong benchmark rate
of $\oo(\sqrt{\log(T)/T})$ using shorter blocks (of length $\log(T)$). Recall that the best rate of convergence we could have
hoped for in the absence of carryovers is $\oo(T^{-1/2})$; and so we are paying a relatively small penalty for the existence
of carryovers here.

While the estimand \smash{$\tau_{\FATE}^{(l, b)}$} may seem surprising at first glance, we argue that it may be a reasonable target in
many application areas. In non-stationary settings as considered here,
both $\tau_{GATE}$ and \smash{$\tau_{FATE}^{(l, b)}$} share a certain arbitrary nature, in that they both
only assess treatment effects during some specific set of time periods specified by the experiment; and, with judicious choices of $l$
and $b$, it is likely that \smash{$\tau_{\FATE}^{(l, b)}$} could be just as relevant as $\tau_{\GATE}$ for downstream decision making.
Consider, for example, a hypothetical switchback run by an online marketplace in the first two weeks of November 2021 with a block length $l = 192$ minutes and burn-in time $b = 96$ minutes
(for a total of 105 periods). Then, $\tau_{\GATE}$ is an average over all time periods over the first two weeks of November 2021,
while \smash{$\tau_{\FATE}^{(l, b)}$} would be an average of half of those periods, where each weekday-hour-minute tuple enters into the target set
\smash{$\mathcal{I}^{(l, b)}$} exactly once (either in the first week or in the second). In this setting, for most purposes, $\tau_{\GATE}$ and
\smash{$\tau_{\FATE}^{(l, b)}$} would be roughly equally informative summaries of the effectiveness of the target intervention. The biggest
question in considering external validity of the study would be whether, in this online marketplace, the first two weeks of November 2021 are representative of future time
periods where the treatment may be deployed at scale.

\section{Jackknife variance estimation}
\label{sec:jackknife}

\subsection{Algorithm}
In the numerical examples, we calculate the confidence intervals of the three estimators of interest using a Jackknife variance estimation procedure. Here, we provide details on how the confidence intervals are constructed.

\begin{algorithm}[t]
\caption{Jackknife Variance Estimation}
\label{alg:jackknife}
\begin{algorithmic}[1]
\Require $\mathbf{Y}=\{Y_1, Y_2, \ldots, Y_T\}$, $\mathbf{Z}=\{Z_1, Z_2, \ldots, Z_k\}$, $\htau_{\text{DM}}^{(l,b)}$, $(l, b)$
\State Let $k=\lfloor T/l \rfloor$ be the number of blocks
\For{$i = 1$ to $k$}
    \State Remove $Y_{(i-1)l+1},\dots,Y_{il}$ and $Z_i$ to obtain $\mathbf{Y}_{-i}$ and $\mathbf{Z}_{-i}$
    \State Calculate $\htau_{\text{DM},-i}^{(l,b)}$ using $\mathbf{Y}_{-i}$ and $\mathbf{Z}_{-i}$ according to equation \eqref{eq:DM}
\EndFor
\State Compute the jackknife variance estimator: $\hV_{\text{DM}}^{(l,b)} = (k-1)k^{-1} \sum_{i=1}^{k} (\htau_{\text{DM},-i}^{(l,b)} - \htau_{\text{DM}}^{(l,b)})^2$\\
\Return $\hV_{\text{DM}}^{(l,b)}$
\end{algorithmic}
\end{algorithm}

For the two difference-in-means estimators $\htau_{\text{DM}}^{(l,0)}$ and $\htau_{\text{DM}}^{(l,b)}$, we note that this can be regarded as a special case of estimating variance with resampling methods using clustered data. The treatment assignments are perfectly correlated within blocks and independent across blocks, while the outcomes are weakly correlated with an exponentially decaying correlation. Thus, we follow the standard practice and resample at the cluster level, i.e., we iteratively exclude one block at a time. The details of the jackknife variance estimation are described in Algorithm \ref{alg:jackknife}.
We then construct the confidence intervals as $\htau_{\text{DM}}^{(l,b)}\pm z_{\alpha/2}\sqrt{\hV_{\text{DM}}^{(l,b)}}$, where $z_{\alpha/2}$ is the critical value from the standard normal distribution corresponding to the desired level of confidence.

\begin{figure}[t]
\centering
\includegraphics[width=0.7\linewidth]{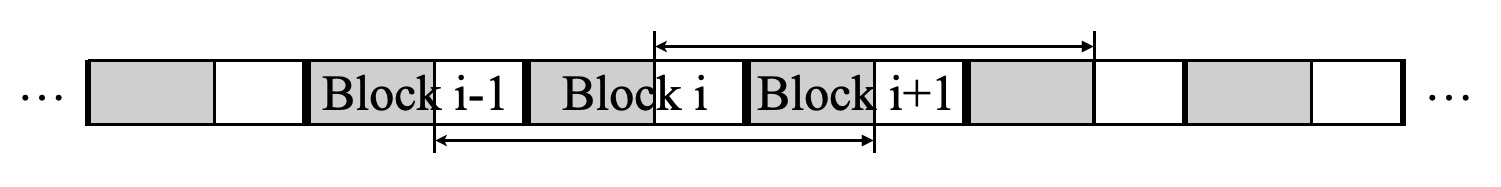}
\caption{An illustration of the block jackknife variance estimation. The
burn-in periods are in gray while the focal periods are in white. The arrows
above/below the display show different (overlapping) regions that will in
turn be removed with the block jackknife.}
\label{fig:block_jackknife}
\end{figure}

\begin{algorithm}[t]
\caption{Block Jackknife Variance Estimation}
\label{alg:jackknife_block}
\begin{algorithmic}[1]
\Require $\mathbf{Y}=\{Y_1, Y_2, \ldots, Y_T\}$, $\mathbf{Z}=\{Z_1, Z_2, \ldots, Z_k\}$, $\htau_{\text{BC}}^{(l,b)}$, $(l, b)$
\State Let $k=\lfloor T/l \rfloor$ be the number of blocks
\For{$i = 1$ to $k-2$}
    \State Remove $Y_{(i-1)l+b+1},\dots,Y_{(i+1)l+b}$ and $Z_{i},Z_{i+1}$ to obtain $\mathbf{Y}_{-i,2}$ and $\mathbf{Z}_{-i,2}$
    \State Calculate $\htau_{\text{BC},-i}^{(l,b)}$ using $\mathbf{Y}_{-i,2}$ and $\mathbf{Z}_{-i,2}$ according to equation \eqref{eq:IPW}
\EndFor
\State Compute the jackknife variance estimator: $\hV^{(l,b)}_{\text{BC}} = \frac{(k-3)^2}{2(k-2)^2} \sum_{i=1}^{k-2} (\htau_{\text{BC},-i}^{(l,b)} - \htau_{\text{BC}}^{(l,b)})^2$\\
\Return $\hV_{\text{BC}}^{(l,b)}$
\end{algorithmic}
\end{algorithm}

For the bias-corrected estimator $\htau_{\text{BC}}^{(l,b)}$, we utilize a block jackknife resampling procedure \citep{kunsch1989jackknife}, which excludes two blocks simultaneously to account for the significant correlation between blocks introduced by weighting with the treatment assigned to the preceding block. To ensure that data preceding and following the excluded blocks are (almost) independent of each other, we iteratively remove focal periods in blocks $i, i+1$ and burn-in periods in blocks $ i+1,i+2$. We provide an illustration of this strategy in Figure \ref{fig:block_jackknife}, and details of the block jackknife variance estimation in Algorithm \ref{alg:jackknife_block}.

\subsection{Conservativeness for the difference-in-means estimators}

In this section, we show that the jackknife variance estimator outlined in
Algorithm~\ref{alg:jackknife} for the difference-in-means estimators
$\hat\tau_{\text{DM}}^{(l,0)}$ and $\hat\tau_{\text{DM}}^{(l,b)}$ is consistent in a
conservative sense. 
Recall from Theorem~\ref{theo:clt} that
\[
V = V_0 + V_1 +2V_{01} + \Sigma_\Delta
\]
is the asymptotic variance appearing in the central limit theorem for
$\hat\tau_{\text{DM}}^{(l,b)}$, where $V_0$ and $V_1$ are the focal-period
variances under treatment and control, $V_{01}$ is their cross-term not identifiable from the observed data, and
$\Sigma_\Delta$ collects the additional variance from the within-block
noise terms. We first define the conservative variance that our jackknife estimator targets, mirroring the standard conservative variance used in design-based causal inference \citep{imbens2015causal}.

\begin{lemm}
Define the conservative variance
\begin{equation}
\bar V = 2V_0 + 2V_1 + \Sigma_\Delta
\end{equation}
Then $\bar V$ is an upper bound on the true asymptotic variance $V$.
\label{lemm:conservative_variance}
\end{lemm}

Below, we show that the jackknife variance estimator in Algorithm \ref{alg:jackknife} is (conservatively) consistent for the asymptotic variance $V$; specifically, $k\hV_{\text{DM}}^{(l,b)}\to_p \bar V$ as $T\to\infty$. Thus, Theorem~\ref{theo:clt} combined with Algorithm~\ref{alg:jackknife} yields asymptotically conservative confidence intervals.

\begin{theo}
Under the conditions of Theorem~\ref{theo:clt},  the block jackknife variance estimator $\hV_{\text{DM}}^{(l,b)}$ defined in Algorithm~\ref{alg:jackknife} satisfies
$k\hV_{\text{DM}}^{(l,b)} \to_p \bar V$ as $T\to\infty$,
where $\bar V$ is the conservative variance bound defined in Lemma~\ref{lemm:conservative_variance}.
\label{theo:variance_estimation}
\end{theo}

\section{Additional Details of Numerical Experiments}

\subsection{Additional results for the simple illustration}
\label{subsec:toy_res}

In Section \ref{subsec:sim_toy}, we demonstrate that the coverage of the two estimators employing burn-in periods, $\htau_{\text{DM}}^{(l,b)}$ and $\htau_{\text{BC}}^{(l,b)}$, stays much closer to the nominal level than that of the difference-in-means estimator, $\htau_{\text{DM}}^{(l,0)}$, which exhibits zero coverage.
To elucidate the rationale behind this observation, we analyze the bias and standard errors of the three estimators across various designs, as outlined in Table \ref{tab:sim_toy_decomp1}. We observe a disproportionately large bias in estimating $\tau$ with $\htau_{\text{DM}}^{(l,0)}$, relative to its standard error. This confirms our conjecture that the undercoverage of the difference-in-means estimator stems from its substantial bias.

{\renewcommand{\arraystretch}{1.25}
\begin{table}[t]
\centering
\begin{tabular}{|c|c|c|c|c|c|c|}
  \hline
  \multirow{2}{*}{Estimator (Estimand)} & \multicolumn{3}{c|}{Bias} & \multicolumn{3}{c|}{Standard Error}  \\ 
  \cline{2-7}
  & $b = 50$ & $b = 100$ & $b = 150$ & $b = 50$ & $b = 100$ & $b = 150$ \\
  \hline
  $\htau_{\text{DM}}^{(l,0)}\, (\tau_{\GATE})$ & 1.170 & 1.170 & 1.170 & 0.228 & 0.228 & 0.228  \\ 
  \hline
  $\htau_{\text{DM}}^{(l,b)}\, (\tau_{\FATE}^{(l,b)})$ & 0.007 & 0.005 & 0.004 & 0.199 & 0.239 & 0.332  \\
  \hline
  $\htau_{\text{BC}}^{(l,b)}\, (\tau_{\GATE})$ & 0.0001 & 0.004 & 0.003 & 0.202 & 0.226 & 0.248  \\
  \hline
\end{tabular}
\caption{Bias and standard error over $1,000$ iterations under the original setup with $T=20,000$ and $l=200$.}
\label{tab:sim_toy_decomp1}
\end{table}
}

In addition to the MDP discussed in Section \ref{subsec:sim_toy}, we explore a simpler MDP characterized by faster mixing dynamics. In this alternate setting, the state transition is governed by the following rules:
\begin{equation}
H_{t+1} = \begin{cases}
\min\cb{H_t+M_t,20} &\qquad \text{ with probability } 0.7, \\
0 &\qquad \text{ with probability } 0.3;
\end{cases}
\end{equation}
if $W_t=1$, and
\begin{equation}
H_{t+1} = \begin{cases}
\min\cb{H_t+M_t,20} &\qquad \text{ with probability } 0.3, \\
0 &\qquad \text{ with probability } 0.7.
\end{cases}
\end{equation}
if $W_t=0$.
Tables \ref{tab:sim_toy_coverage2} and \ref{tab:sim_toy_decomp2} present the coverage of 95\% confidence intervals, as well as biases and variances of the three estimators under this easier setup. In this fast-mixing environment, the bias associated with the difference-in-means estimator is significantly reduced, leading to a much higher coverage rate for its confidence interval.

{\renewcommand{\arraystretch}{1.25}
\begin{table}[t]
\centering
\begin{tabular}{|c|c|c|c|}
  \hline
 Design & $\htau_{\text{DM}}^{(l,0)} (\tau_{\GATE})$  & $\htau_{\text{DM}}^{(l,b)} (\tau_{\FATE}^{(l,b)})$ & $\htau_{\text{BC}}^{(l,b)} (\tau_{\GATE})$  \\ 
  \hline
 $T = 20000,\,b = 50,\,l = 200$ & 95\% & 95.6\% & 94.6\%  \\ 
  \hline
 $T = 20000,\,b = 100,\,l = 200$  & 95\% & 95.1\% & 94.1\%  \\ 
   \hline
 $T = 20000,\,b = 150,\,l = 200$  & 95\% & 95.5\% & 93.5\%  \\ 
   \hline
\end{tabular}
\caption{Coverage of 95\% confidence intervals over $1,000$ iterations under an easier setup.}
\label{tab:sim_toy_coverage2}
\end{table}
}

{\renewcommand{\arraystretch}{1.25}
\begin{table}[t]
\centering
\begin{tabular}{|c|c|c|c|c|c|c|}
  \hline
  \multirow{2}{*}{Estimator (Estimand)} & \multicolumn{3}{c|}{Bias} & \multicolumn{3}{c|}{Standard Error}  \\ 
  \cline{2-7}
  & $b = 50$ & $b = 100$ & $b = 150$ & $b = 50$ & $b = 100$ & $b = 150$ \\
  \hline
  $\htau_{\text{DM}}^{(l,0)}\, (\tau_{GATE})$ & 0.030 & 0.030 & 0.030 & 0.175 & 0.175 & 0.175  \\ 
  \hline
  $\htau_{\text{DM}}^{(l,b)}\, (\tau_{FATE}^{(l,b)})$ & 0.001 & 0.001 & 0.013 & 0.200 & 0.245 & 0.345  \\
  \hline
  $\htau_{\text{BC}}^{(l,b)}\, (\tau_{GATE})$ & 0.003 & 0.003 & 0.005 & 0.202 & 0.222 & 0.242  \\
  \hline
\end{tabular}
\caption{Bias and standard error over $1,000$ iterations under an easier setup with $T=20,000$ and $l=200$.}
\label{tab:sim_toy_decomp2}
\end{table}
}

\begin{figure}[t]
\centering
\includegraphics[width=\linewidth]{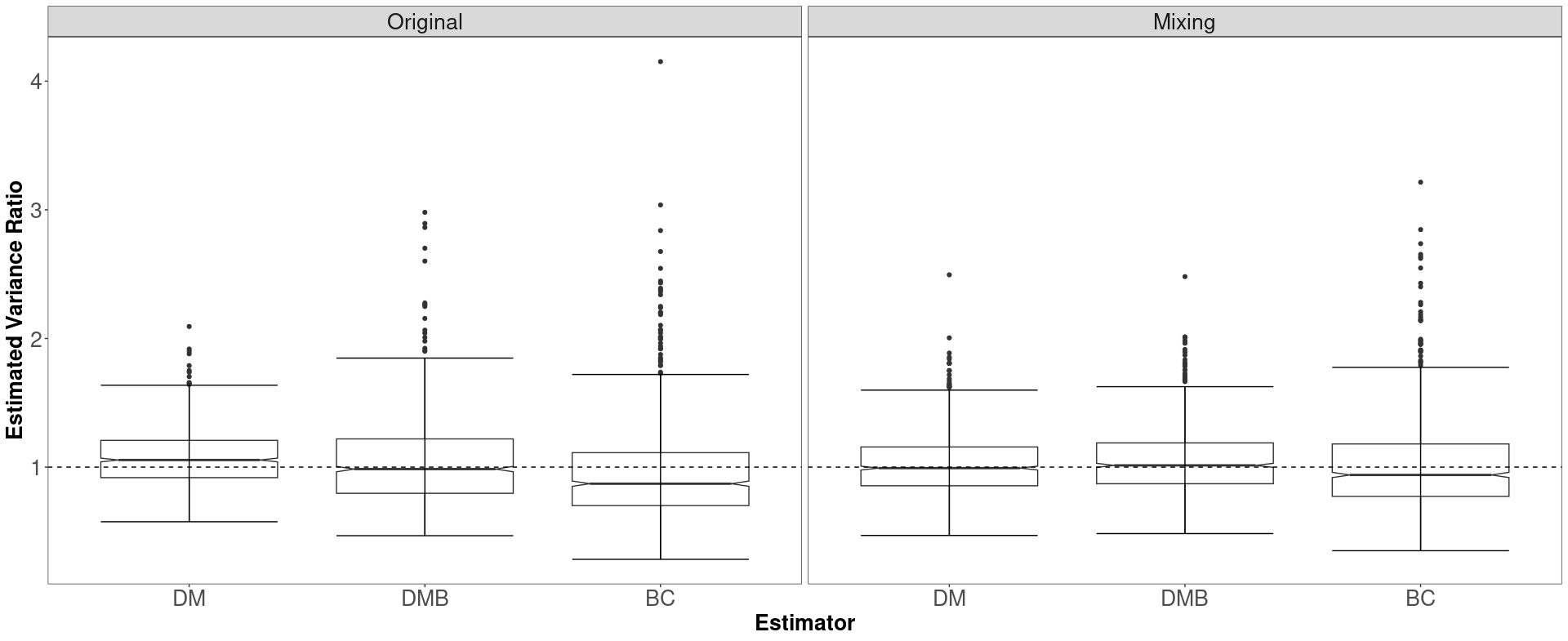}
\caption{Distributions of the ratios between the estimated variance and true variance for the three estimators $\htau_{\text{DM}}^{(l,0)}$ (DM), $\htau_{\text{DM}}^{(l,b)}$ (DMB), $\htau_{\text{BC}}^{(l,b)}$ (BC) under the original and the easier setups.}
\label{fig:toy_var_ratio}
\end{figure}

In Figure \ref{fig:toy_var_ratio}, we plot the distributions of the ratios between the estimated variance and true variance for the three estimators under the original and easier setups. The variance estimates can be quite variable, especially under the original setup, which may contribute to the slight undercoverage in Table \ref{tab:sim_toy_coverage1}. Under the faster-mixing setup, the ratios are more concentrated and the coverage stays closer to the nominal level, as shown in Table \ref{tab:sim_toy_coverage2}.

\subsection{Mixing behavior in the ride-sharing simulator}
\label{sec:outcome_mixing_sim}

To further illustrate the mixing behavior in the ride-sharing simulator, we first construct two different initial states. Using the same sequence of rider and driver arrivals, we run one system under the treatment policy and the other under the control policy for 30,000 seconds. This produces two states that differ only because of their past dispatching policies. Starting from these fixed states, we then apply the treatment policy to both systems for an additional 6,000 seconds. We repeat this experiment 100 times. Each repetition uses a new realization of future rider and driver arrivals, but the two systems within a repetition receive exactly the same arrivals and other random inputs. We aggregate completed platform profit over 400-second intervals and compare its distribution across repetitions. If the simulator gradually forgets its initial state, the profit distributions from the two treatment histories should become increasingly similar over time.

\begin{figure}[t]
\centering
\includegraphics[width=0.7\linewidth]{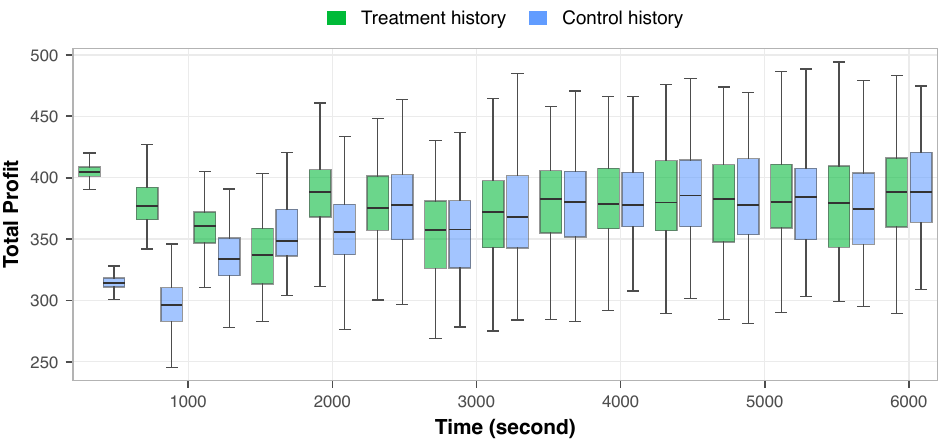}
\caption{Decay of carryover effects in the ride-sharing simulator. Distribution of aggregated profit across 100 runs for two different initial treatment histories. Both systems follow the treatment policy after time zero.}
\label{fig:outcome_mixing}
\end{figure}

Figure~\ref{fig:outcome_mixing} compares the distribution of aggregated profit across 100 simulations at successive times after both systems begin following the same treatment policy. The two systems are initialized from states generated under different treatment histories, so their outcome distributions are quite different at the beginning of the experiment. As they spend more time under the same policy, however, this difference becomes progressively smaller, and the two distributions start to overlap. The influence of the earlier treatment history remains visible for some time, but weakens steadily as the system evolves. This aligns with the qualitative behavior suggested by the geometric mixing condition, under which dependence on the initial state decays over time.

\subsection{Additional results for the  ride-sharing simulation}
\label{subsec:ridesharing_res}

In this section, we present additional results examining the performance of the estimators under a wider range of $(l,b)$ choices. The results can be found in Table \ref{tab:sim_taxi_res_additional}.
When the block length is short ($l=20$), increasing the burn-in period substantially reduces bias and mean squared error. With a longer block length ($l=60$), the estimators have much smaller bias and their performance is relatively stable across the burn-in periods considered. Overall, these results are consistent with the patterns observed for the main choice of $l=40$.

{\renewcommand{\arraystretch}{1.25}
\begin{table}[t]
\centering
\begin{tabular}{|c|c|c|c|c|}
  \hline
  Estimator & Bias  & Standard Error & Mean Squared Error & Coverage (Monte Carlo SE)  \\ 
  \hline
  $\htau_{\text{DM}}^{(20,0)}$ & 5.42 & 2.48  & 35.40 & 34\% (4.7\%) \\ 
  \hline
  $\htau_{\text{DM}}^{(20,4)}$ & 3.42 & 2.71 & 18.99 & 65\% (4.8\%)  \\
  \hline
  $\htau_{\text{DM}}^{(20,12)}$ & 0.97 & 3.08 & 10.32 & 93\% (2.6\%)  \\
  \hline
  $\htau_{\text{BC}}^{(20,4)}$ & 2.86 & 2.71 & 15.47 & 74\% (4.4\%) \\
  \hline
  $\htau_{\text{BC}}^{(20,12)}$ & 0.58 & 3.09 & 9.79 & 90\% (3.0\%)  \\
  \hline
  \hline
  $\htau_{\text{DM}}^{(60,0)}$ & 1.53 & 2.48  & 8.42 & 87\% (3.4\%)  \\ 
  \hline
  $\htau_{\text{DM}}^{(60,8)}$ & 0.05 & 2.52 & 6.31 & 95\% (2.2\%) \\
  \hline
  $\htau_{\text{DM}}^{(60,16)}$ & 0.34 & 2.55 & 6.53 & 96\% (2.0\%) \\
  \hline
  $\htau_{\text{BC}}^{(60,8)}$ & 0.12 & 2.55 & 6.43 & 90\% (3.0\%) \\
  \hline
  $\htau_{\text{BC}}^{(60,16)}$ & 0.12 & 2.55 & 6.44 & 93\% (2.6\%) \\
  \hline
\end{tabular}
\caption{Bias, standard error, mean squared error, and coverage of 95\% confidence intervals across $100$ iterations.
For readability, the superscripts for all estimators are given in 100s of seconds.}
\label{tab:sim_taxi_res_additional}
\end{table}
}

\subsection{Additional details on the ride-sharing simulator}
\label{subsec:ridesharing_details}

In this section, we provide further details on the ride-sharing simulator. To generate drivers and ride requests, we rely on the NYC Yellow Taxi Trip Records data from the year 2015 \citep{tlcrecord}. Our simulation design is adapted from \citet{farias2022markovian}. \\

\noindent \textbf{Ride Request Generation.} 
Request times are randomly generated, with interarrival times drawn from an exponential distribution with a rate of $0.1\gamma_\text{taxi}$, where $\gamma_\text{taxi}$ is the average rate at which requests are observed in the dataset.
The pickup and dropoff locations for each request are randomly drawn with replacement from the full trip records dataset.
We further generate, for each rider, a value-of-time parameter as a random draw from a lognormal distribution with a mean of $0.003$ and a variance of $1$. Each request is for a single rider.\\

\noindent \textbf{Driver Generation.} 
Drivers are created by drawing pickup locations with replacement from the same requests dataset as used above, with the interarrival time following an exponential distribution with a rate of $0.003\gamma_\text{taxi}$.
The driver then remains in the system until the end of their shift, with a fixed capacity of $3$ riders. The duration of the shift follows an exponential distribution with a mean of 30,000 seconds.\\

\noindent \textbf{Dispatch.} When a ride request is initiated, the dispatcher selects a driver according to the following procedure:
\begin{enumerate}
    \item Fetch the top 40 nearest drivers and sort them based on their estimated time of arrival.
    \item Categorize the drivers into two groups: idle drivers and pool drivers. Consider the top 10 pool drivers and the top 1 idle driver as candidates for dispatching.
    \item Calculate the cost of adding the new ride request to the route of each candidate driver.
    \item If the cost of a candidate pool driver is less than $\theta_d$ times the cost of the idle driver, dispatch the pool driver with the lowest cost. Otherwise, dispatch the idle driver.
\end{enumerate}
Whenever a driver is dispatched to a ride request, their current route is updated by inserting the pickup and dropoff locations based on Dijkstra's algorithm that finds the shortest paths.\\


\noindent \textbf{Offer.} After dispatching a driver, we extend an offer to the rider based on the anticipated cost of providing the ride. This cost is calculated as the sum of multiplying the total distance (in kilometers) by $0.6$ and the total time (in seconds) by $0.01$, covering driver earnings, maintenance, fuel, and all other expenses. The resulting cost is then multiplied by $1.5$ to determine the price offered to the rider. If the dispatched driver is a pool driver, the rider receives an additional $50\%$ discount on the offer.\\

\noindent \textbf{Response.} Riders have access to an outside option allowing them to travel directly from the pickup to the dropoff location after a 15-minute wait. The price of this outside option is calculated as the sum of multiplying the total distance (in kilometers) by $0.6$ and the total time (in seconds, including the waiting time) by $0.01$.
To determine whether to accept a ride offer, riders compare it with their outside option. This comparison is based on their disutilities associated with each option, calculated as the estimated time to their destination multiplied by their value-of-time parameter, plus the price of the option (before any discount).\\

\noindent \textbf{Completion.} If the offer is accepted, the dispatched driver's route is updated accordingly. Upon completion of the ride, both the rider and the ride-sharing company incur payments and costs, respectively. The net platform profit is calculated as the difference between the payment made by the customer and the cost of providing the ride.

\section{Proof of Theorems}

\subsection{Proof of Theorem \ref{theo:GATE_mse}}
\label{proof:GATE_mse}

We break the proof of Theorem \ref{theo:GATE_mse} into two lemmas on bias and variance, respectively.
First, we upper bound the bias of $\EE{\htau_{\text{DM}}^{(l,b)}}$ in estimating $\tau_{\FATE}^{(l, b)}$ and $\tau_{\GATE}$. Recall that $k=k_1+k_0=\lfloor T/l \rfloor$. For simplicity, we write
\begin{equation}
\htau_{\text{DM}}^{(l,b)}=\frac{1}{k}\sum_{i=1}^{k}\frac{1}{l-b}
\sum_{t=b+1}^l\htau_{(i-1)l+t},
\label{eq:estimator}
\end{equation}
where
\begin{equation}
\htau_{s}=\frac{Y_{s}W_s}{k_1/k}-\frac{Y_{s}(1-W_s)}{k_0/k}.
\end{equation}

\begin{lemm}
Under the assumptions in Theorem \ref{theo:GATE_mse},
the bias of $\htau_{\text{DM}}^{(l,b)}$ as an estimator of $\tau_{\FATE}^{(l, b)}$
can be upper bounded as
\[
\abs{\EE{\htau_{\text{DM}}^{(l,b)}}-\tau_{\FATE}^{(l, b)}}\le \frac{4\Lambda}{1-\exp\p{-1/t_\mix}}\cdot\frac{\exp\p{-b/t_{\text{mix}} }}{l-b}+\oo\p{2^{-k}}.\]
Furthermore,
\begin{equation*}
\abs{\tau_{\FATE}^{(l, b)} - \tau_{\GATE}} \leq \Psi \, b/l.
\end{equation*}
\label{lemm:mse_bias}
\end{lemm}

\begin{lemm}
Under the assumptions in Theorem \ref{theo:GATE_mse},
the variance of $\htau_{\text{DM}}^{(l,b)}$
can be upper bounded as
\begin{equation*}
\begin{split}
\Var{\htau_{\text{DM}}^{(l,b)}}&\le \frac{16\Lambda^2}{k}+\frac{16\Lambda^2}{k(l-b)}\frac{\exp\p{-b/t_{\rm mix}}}{1-\exp\p{-1/t_{\rm mix}}}\\
&\qquad\qquad+\frac{4\sigma^2}{k(l-b)}+\oo\p{\frac{1}{k^2(l-b)}}+\oo\p{2^{-k}}.
\end{split}
\end{equation*}

\label{lemm:mse_var}
\end{lemm}

Combining Lemma \ref{lemm:mse_bias} and Lemma \ref{lemm:mse_var} gives the result in Theorem \ref{theo:GATE_mse}.

\subsection{Proof of Theorem \ref{theo:clt}}
\label{sec:proof_clt}

Using notation from \eqref{eq:bY_def} and \eqref{eq:mu_M_def}, we start
by showing that our estimator can be decomposed as follows.

\begin{lemm}
Under the assumptions in Theorem \ref{theo:clt}, the difference-in-means estimator $\htau_{\text{DM}}^{(l,b)}$ can be decomposed as
\begin{equation}
\label{eq:decom}
\htau_{\text{DM}}^{(l,b)}
= \tau_{\FATE}^{(l, b)} + \Gamma^{(l,b)} + \Delta^{(l,b)} + \oop\p{\frac{e^{-b/t_{\text{mix}}}}{l - b}}, 
\end{equation}
where
\begin{equation}
\Gamma^{(l,b)} = \frac{1}{k}\sum_{i=1}^{k}\sqb{\frac{\p{\bmu_i(1) - \frac{1}{k} \sum_{j = 1}^k \bmu_j(1)}Z_i}{k_1/k}-\frac{\p{\bmu_i(0) - \frac{1}{k} \sum_{j = 1}^k \bmu_j(0)}(1-Z_i)}{1-k_1/k}},
\end{equation}
\begin{equation}
\Delta^{(l,b)} = \frac{1}{k}\sum_{i=1}^{k}\sqb{\frac{\cb{\bY_i(1)-\bM_i(1)}Z_i}{k_1/k}-\frac{\cb{\bY_i(0)-\bM_i(0)}(1-Z_i)}{1-k_1/k}}.
\end{equation}
\label{lemm:clt_decom}
\end{lemm}

It remains to verify that $\sqrt{k}(\Gamma^{(l,b)}+\Delta^{(l,b)})$ is asymptotically normal. 
To do this, we write $\Gamma^{(l,b)}+\Delta^{(l,b)}$ 
as the average of a martingale difference sequence (up to terms negligible asymptotically). For notation simplicity, we write $Z:=(Z_1,\ldots,Z_k)$ and
\begin{equation}
D_i:=2(2Z_i-1)\p{\bY_i-\EE{\bY_i\cond\mathcal F_{(i-1)l+1}, Z}},\qquad i=1,\ldots,k.
\end{equation}

\begin{lemm}
Under the assumptions in Theorem \ref{theo:clt},  \begin{equation}
\begin{split}
\sqrt{k}\Gamma^{(l,b)}+\sqrt{k}\Delta^{(l,b)}
&=\frac{1}{\sqrt{k}}\sum_{i=1}^k(2Z_i-1)\p{\bmu_i(1)+\bmu_i(0)-\frac{1}{k}\sum_{j=1}^k\{\bmu_j(1)+\bmu_j(0)\}}\\
&\qquad+\frac{1}{\sqrt{k}}\sum_{i=1}^kD_i+o_P(1),
\end{split}
\label{eq:theorem4_decomp}
\end{equation}
\label{lemm:decomposition}
\end{lemm}

Consider a filtration that first reveals $Z_1,\ldots,Z_k$ sequentially, and then reveals the state and outcome trajectory block by block conditionally on the full assignment vector $Z$. Then, 
the $2k$ leading terms in~\eqref{eq:theorem4_decomp} form a martingale difference sequence, as
\begin{equation}
\begin{split}
&\EE{\frac{1}{\sqrt{k}}(2Z_i-1)\p{\bmu_i(1)+\bmu_i(0)-\frac{1}{k}\sum_{j=1}^k\{\bmu_j(1)+\bmu_j(0)\}}\cond Z_1,\ldots,Z_{i-1}}\\
&\qquad=\frac{1}{\sqrt{k}}\p{\bmu_i(1)+\bmu_i(0)-\frac{1}{k}\sum_{j=1}^k\{\bmu_j(1)+\bmu_j(0)\}}\EE{2Z_i-1\cond Z_1,\ldots,Z_{i-1}}=0,
\end{split}
\end{equation}
and
\begin{equation}
\EE{D_i\cond \mathcal F_{(i-1)l+1}, Z}=2(2Z_i-1)\EE{\left(\bY_i-\EE{\bY_i\cond\mathcal F_{(i-1)l+1}, Z}\right)\cond \mathcal F_{(i-1)l+1}, Z} = 0.
\end{equation}

Next, we characterize the limiting variance contributions of the two martingale components in~\eqref{eq:theorem4_decomp}.

\begin{lemm}
Under the assumptions in Theorem \ref{theo:clt}, 
\begin{equation}
\frac{1}{{k}}\sum_{i=1}^k\p{\bmu_i(1)+\bmu_i(0)-\frac{1}{k}\sum_{j=1}^k\{\bmu_j(1)+\bmu_j(0)\}}^2
\to_p 
\Sigma_\Gamma ,
\label{eq:clt_moments_Gamma}
\end{equation}
and
\begin{equation}
\frac{1}{{k}}\sum_{i=1}^k\EE{D_i^2\cond \mathcal F_{(i-1)l+1}, Z}
\to_p 
\Sigma_\Delta.
\label{eq:clt_moments_Delta}
\end{equation}
\label{lemm:clt_moments}
\end{lemm}

Since $\bmu_i(1)+\bmu_i(0)$ and $D_i$ are all bounded under the condition that $\abs{\bY_i}\le\Gamma_0$, the conditional Lindeberg condition holds.
Applying the martingale central limit theorem~\citep{hall1980martingale} therefore gives
\begin{equation}
\sqrt{k}\Gamma^{(l,b)}+\sqrt{k}\Delta^{(l,b)}\to_d\mathcal{N}(0,\Sigma_\Gamma+\Sigma_\Delta).
\end{equation}
The claimed result then follows directly by applying Slutsky's theorem and Lemma~\ref{lemm:clt_decom}.

\subsection{Proof of Theorem \ref{theo:IPW_mse}}

As in Section \ref{proof:GATE_mse}, we decompose the proof into the proof of two lemmas on bias and variance. To ease the notation, assume that there is one additional block with treatment assignment $Z_0\sim \text{Bernoulli}(0.5)$ such that $W_{t}=Z_0$ for $t=-l,\dots,0$, and consider the following estimator:
\begin{equation}
\widetilde{\htau_{\text{BC}}^{(l,b)}}=\frac{1}{kl}\sum_{i=1}^{k}
\sum_{t=1}^l\htau_{(i-1)l+t},
\label{eq:estimator_bc}
\end{equation}
where we define with a slight abuse of notation
\begin{equation}
\htau_{s}=
\begin{cases}
\dfrac{Y_sW_s}{k_1/k}-\dfrac{Y_s(1-W_s)}{k_0/k},
& \text{if }1+((s-1)\bmod l)>b,\\[6pt]
\dfrac{Y_sW_sW_{s-l}}{k_{11}/k}
-\dfrac{Y_s(1-W_s)(1-W_{s-l})}{k_{00}/k},
& \text{if }1+((s-1)\bmod l)\le b.
\end{cases}
\end{equation}
For blocks $i=1,\dots,k$, 
define in addition that 
\begin{equation}
\label{eq:bYb_def}
\bY^b_i=\frac{1}{b}\sum_{t=1}^b Y_{(i-1)l+t}, \ \ \ \ \bY^b_i(w) = \bY_i^b\cond do(Z_i=Z_{i-1}=w)
\end{equation}
and
\begin{equation}
\label{eq:Mb_def}
\bM^b_i(w) = \EE{\bY^b_i(w)}.
\end{equation}

\begin{lemm}
Under the assumptions in Theorem \ref{theo:IPW_mse},
the bias of $\widetilde{\htau_{\text{BC}}^{(l,b)}}$ as an estimator of $\tau_{\GATE}$
can be upper bounded as
\begin{equation}
   \abs{\EE{{\htau_{\text{BC}}^{(l,b)}}}-\tau_{\GATE}}\le \frac{4\Lambda}{1-\exp\p{-1/t_\mix}}\cdot\frac{\exp\p{-b/t_{\text{mix}} }}{l}+\oo\p{\frac{1}{k}}.
\end{equation}
\label{lemm:mse_bias_bc}
\end{lemm}

\begin{lemm}
Under the assumptions in Theorem \ref{theo:IPW_mse},
the variance of ${\htau_{\text{BC}}^{(l,b)}}$ as an estimator of $\tau_{\GATE}$
can be upper bounded as
\begin{equation}
\begin{split}
\Var{\htau_{\text{BC}}^{(l,b)}}&\le\frac{88\Lambda^2}{k}+\frac{16\sqrt{2}\Lambda^2\exp\p{-l/t_{\rm mix}}}{kl\p{1-\exp\p{-1/t_{\rm mix}}}}+\frac{8\sigma^2}{kl}+\oo\p{\frac1{k^2}}.
\end{split}
\end{equation}
\label{lemm:mse_var_bc}
\end{lemm}

Combining Lemmas \ref{lemm:mse_bias_bc} and \ref{lemm:mse_var_bc} gives the result in Theorem \ref{theo:IPW_mse}.

\subsection{Proof of Theorem \ref{theo:clt_bc}}

Note that
\begin{equation}
\label{eq:GATE_recall}
\tau_{\GATE} = \frac{1}{k} \sum_{i = 1}^k \sqb{\frac{l-b}{l}\cb{\bmu_i(1) - \bmu_i(0)} +  \frac{b}{l}\cb{\bmu^b_i(1) - \bmu^b_i(0)}}.
\end{equation}
We first decompose $\htau_{\text{BC}}^{(l,b)}$ as follows.

\begin{lemm}
The bias-corrected difference-in-means estimator $\htau_{\text{BC}}^{(l,b)}$ with burn-in periods can be decomposed as
\begin{equation}
\label{eq:decom_bc}
\htau_{\text{BC}}^{(l,b)}
= \tau_{\GATE} + \Tilde{\Gamma}^{(l,b)} + \Tilde{\Delta}^{(l,b)} + \oop\p{\frac{1}{k}} + \oop\p{\frac{e^{-b/t_{\text{mix}}}}{l}}, 
\end{equation}
where
\begin{equation}
\label{eq:Gamma_bc}
\begin{split}
\Tilde{\Gamma}^{(l,b)} &= \frac{l-b}{l}\Gamma^{(l,b)}+
\frac{b}{kl}\sum_{i=1}^{k}\left[\frac{\p{\bmu^b_i(1) - \frac{1}{k} \sum_{j = 1}^k \bmu^b_j(1)}Z_iZ_{i-1}}{k_{11}/k}\right.\\
&\qquad\qquad\qquad\qquad\left.
-\frac{\p{\bmu^b_i(0) - \frac{1}{k} \sum_{j = 1}^k \bmu^b_j(0)}(1-Z_i)(1-Z_{i-1})}{k_{00}/k}\right],
\end{split}
\end{equation}
\begin{equation}
\label{eq:Delta_bc}
\begin{split}
\Tilde{\Delta}^{(l,b)} &= \frac{l-b}{l}\Delta^{(l,b)}+
\frac{b}{kl}\sum_{i=1}^{k}\left[\frac{\cb{\bY^b_i(1)-\bM^b_i(1)}Z_iZ_{i-1}}{k_{11}/k}\right.\\
&\qquad\qquad\qquad\qquad\left.
-\frac{\cb{\bY^b_i(0)-\bM^b_i(0)}(1-Z_i)(1-Z_{i-1})}{k_{00}/k}\right].
\end{split}
\end{equation}
\label{lemm:clt_decom_bc}
\end{lemm}

Again, as in Section \ref{sec:proof_clt}, we prove a central limit theorem for $\Tilde{\Gamma}^{(l,b)}$.

\begin{lemm}
Under the assumptions in Theorem \ref{theo:clt_bc}, as $k\to\infty$,
\begin{equation}
\label{eq:GammaCLT_bc}
\sqrt{k}\Tilde{\Gamma}^{(l,b)} \to_d \nn\p{0, \Tilde{\Sigma}_\Gamma},
\end{equation}
where 
\begin{equation}
\begin{split}
\Tilde{\Sigma}_{\Gamma} 
&= (1-\beta)^2\p{V^f_0+V^f_1+2V^f_{01}}\\
&\qquad\qquad +\beta^2 \p{3V^b_0+3V^b_1+2V^b_{01}+2V^{bs}_{01}+2V^{bs}_{10}+2V^{bs}_{00}+2V^{bs}_{11}}\\
&\qquad\qquad +2\beta(1-\beta)
\p{{V}^{bf}_{0}+{V}^{bf}_{1}+{V}^{bf}_{01}+{V}^{bf}_{10}+{V}^{fs}_{00}+{V}^{fs}_{11}+{V}^{fs}_{01}+{V}^{fs}_{10} }.
\end{split}
\end{equation}
\label{lemm:GammaCLT_bc}
\end{lemm}

Finally, we show that
$\Tilde{\Delta}^{(l,b)}$ is a smaller-order term that is negligible. 

\begin{lemm}
Under the assumptions in Theorem \ref{theo:clt_bc},  as $k\to\infty$,
\begin{equation}
\begin{split}
\sqrt{k}\Tilde{\Delta}^{(l,b)}\to_p 0.
\end{split}
\end{equation}
\label{lemm:Delta_bc}
\end{lemm}

\subsection{Proof of Theorem \ref{theo:variance_estimation}}

For what follows, we will focus on the event where $k_0,k_1\ge 2$, whose complement happens with exponentially small probability.
Define
\begin{equation}
\begin{split}
\delta_{1,-j}=\frac{1}{k_{1,-j}}\sum_{i\ne j}Z_i\bY_i - \frac{1}{k_1}\sum_{i}Z_i\bY_i,
\end{split}
\end{equation}
\begin{equation}
\begin{split}
\delta_{0,-j}=\frac{1}{k_{0,-j}}\sum_{i\ne j}(1-Z_i)\bY_i - \frac{1}{k_0}\sum_{i}(1-Z_i)\bY_i,
\end{split}
\end{equation}
and let $k_{1,-j}$ and $k_{0,-j}$ be the total number of treated and control blocks after excluding block $j$.
The jackknife variance estimator can be decomposed as
\begin{equation}
\begin{split}
\hV_{\text{DM}}^{(l,b)} 
&= \frac{k-1}{k} \sum_{j=1}^{k} (\htau_{\text{DM},-j}^{(l,b)} - \htau_{\text{DM}}^{(l,b)})^2\\
&= \frac{k-1}{k} \sum_{j=1}^{k} (\delta_{1,-j}^2 + \delta_{0,-j}^2 + 2\delta_{1,-j}\delta_{0,-j}),
\end{split}
\end{equation}
Define $\bY^{(1)} = \sum_{i}Z_i\bY_i/k_1$ and $\bY^{(0)} = \sum_{i}(1-Z_i)\bY_i/k_0$ to be the average outcome in the treated and control periods. $\delta_{1,-j}$ and $\delta_{0,-j}$ can be simplified as
\begin{equation}
\begin{split}
\delta_{1,-j}
&= \frac{1}{k_{1,-j}}\sum_{i\ne j}Z_i\bY_i - \frac{1}{k_1}\sum_{i}Z_i\bY_i\\
&= \frac{1}{k_{1,-j}}\p{\sum_{i}Z_i\bY_i-Z_j\bY_j }- \frac{1}{k_1}\sum_{i}Z_i\bY_i\\
&= \frac{1}{k_{1,-j}}\p{\frac{Z_j}{k1}\sum_{i}Z_i\bY_i-Z_j\bY_j }\\
&= \frac{1}{k_{1,-j}} Z_j \p{\bY^{(1)}-\bY_j }
\end{split}
\end{equation}
and
\begin{equation}
\begin{split}
\delta_{0,-j}
&= \frac{1}{k_{0,-j}} (1-Z_j) \p{\bY^{(0)}-\bY_j },
\end{split}
\end{equation}
respectively. Note that
\begin{equation}
\begin{split}
\delta_{1,-j}\delta_{0,-j}
&= \frac{1}{k_{1,-j}k_{0,-j}} Z_j(1-Z_j) \p{\bY^{(1)}-\bY_j }\p{\bY^{(0)}-\bY_j }=0,
\end{split}
\end{equation}
and thus
\begin{equation}
\begin{split}
\hV_{\text{DM}}^{(l,b)} 
&= \frac{k-1}{k} \sum_{j=1}^{k} (\delta_{1,-j}^2 + \delta_{0,-j}^2)\\
&= \frac{k-1}{k} \frac{k_1}{(k_{1}-1)^2}\cdot \p{\frac{1}{k_1}\sum_{j=1}^k Z_j(\bY_j-\bY^{(1)})^2 } \\
&\qquad\qquad + \frac{k-1}{k} \frac{k_0}{(k_{0}-1)^2}\cdot \p{\frac{1}{k_0}\sum_{j=1}^k (1-Z_j)(\bY_j-\bY^{(0)})^2 }.
\end{split}
\end{equation}

We start by looking at the term $\sum_{j=1}^k Z_j(\bY_j-\bY^{(1)})^2/k_1$ and everything applies the same to $\sum_{j=1}^k (1-Z_j)(\bY_j-\bY^{(0)})^2/k_0$. Recall that we use $\bmu_j(1)$ to denote the expectation of $\bY_j$ in a system that always receives treatment, and $\bM_i(1)$ to denote the expectation of $\bY_j$ in switchback given $Z_j=1$. Note that
\begin{align}
&\frac{1}{k_1}\sum_{j=1}^k Z_j(\bY_j-\bY^{(1)})^2 \nonumber\\
&\qquad= \frac{1}{k_1}\sum_{j=1}^k Z_j\cb{(\bmu_j(1) - \sum_{i} Z_i\bmu_i(1)/k_1)+(\bY_j-\bmu_j(1))-  \frac{1}{k_1}\sum_{i}Z_i(\bY_i - \bmu_i(1)) }^2\nonumber\\
&\qquad= \frac{1}{k_1}\sum_{j} Z_j(\bmu_j(1) - \sum_{i}Z_i\bmu_i(1)/k_1)^2
+\frac{1}{k_1}\sum_{j} Z_j(\bY_j(1)-\bmu_j(1))^2 \label{eq:varp_est}\\
&\qquad\qquad -\p{\frac{1}{k_1}\sum_{i}Z_i(\bY_i(1) - \bmu_i(1))}^2 \label{eq:varp_fourth}\\
&\qquad\qquad 
+\frac{2}{k_1}\sum_{j} Z_j(\bmu_j(1) - \sum_{i}Z_i\bmu_i(1)/k_1)(\bY_j(1)-\bmu_j(1))\label{eq:varp_corr_same} \\
&\qquad\qquad +\frac{2}{k_1^2} \cb{\sum_{i}Z_i(\bmu_i(1) - \sum_{h}Z_h\bmu_h(1)/k_1)} \cb{\sum_j Z_j(\bY_j-\bmu_j(1))} \label{eq:varp_corr_cross} 
\end{align}

We first show that \eqref{eq:varp_fourth}-\eqref{eq:varp_corr_cross} are all small-order terms that can be ignored.
First, by \eqref{eq:mean_diff}, 
\begin{equation}
\bY_j(1) - \bmu_j(1) = \bY_j(1) - \bM_j(1) + \oop\p{\frac{\exp\p{-b/t_{\text{mix}}}}{l-b}} .
\end{equation}
Furthermore, from~\eqref{eq:lemma15_approximation},
\begin{equation}
\frac{1}{k_1}\sum_{j}Z_j(\bY_j(1) -\bM_j(1)) 
= \frac{1}{k_1}\sum_{j}Z_j\p{\bY_j-\EE{\bY_j\cond\mathcal F_{(j-1)l+1},Z}}+\oo\p{\frac{e^{-b/t_\mix}}{l-b}},
\end{equation}
which is $\oop(k^{-1/2})$ because the first term is an average of bounded martingale differences conditional on $Z$ and the second term is $o(k^{-1/2})$.
Thus, $\eqref{eq:varp_fourth}=o_p(1)$.

Next, we show that
\begin{equation}
\eqref{eq:varp_corr_same} 
= \frac{2}{k_1}\sum_{j} Z_j(\bmu_j(1) - \sum_{i}Z_i\bmu_i(1)/k_1)(\bY_j(1)-\bM_j(1)) + \oo\p{\frac{\exp\p{-b/t_{\text{mix}}}}{l-b}}
\end{equation}
is negligible. Again following~\eqref{eq:lemma15_approximation},
\begin{equation*}
\begin{split}
\eqref{eq:varp_corr_same} 
&= \frac{2}{k_1}\sum_{j} (\bmu_j(1) - \sum_{i}Z_i\bmu_i(1)/k_1)\\
&\qquad \cdot Z_j\p{\bY_j-\EE{\bY_j\cond\mathcal F_{(j-1)l+1},Z}}+\oop\p{\frac{e^{-b/t_\mix}}{l-b}}
\end{split}
\end{equation*}
Again, because the first term is an average of bounded martingale differences conditional on $Z$ and the second term is $o(k^{-1/2})$,
$\eqref{eq:varp_corr_same}=o_p(1)$.

Finally, for \eqref{eq:varp_corr_cross}, note that
\begin{equation}
\sum_{i}Z_i(\bmu_i(1) - \sum_{h}Z_h\bmu_h(1)/k_1)
=\sum_{i}Z_i\bmu_i(1) - \sum_{h}Z_h\bmu_h(1) = 0,
\end{equation}
and thus $\eqref{eq:varp_corr_cross}=0$.
Putting everything together, we have shown that
\begin{equation}
\begin{split}
k\hV_{\text{DM}}^{(l,b)}
&= \frac{k}{k_{1}^2}\sum_{j:Z_j=1} \p{\bmu_j(1) - \frac{\sum_{i}Z_i\bmu_i(1)}{k_1}}^2+ \frac{k}{k_{0}^2}\sum_{j:Z_j=0} \p{\bmu_j(0) - \frac{\sum_{i}(1-Z_i)\bmu_i(0)}{k_0}}^2\\
&\qquad + \frac{k}{k_{1}^2}\sum_{j} Z_j(\bY_j(1)-\bmu_j(1))^2 + \frac{k}{k_{0}^2}\sum_{j} (1-Z_j)(\bY_j(0)-\bmu_j(0))^2+ o_P(1)\\
&= \frac{2}{k_{1}}\sum_{j:Z_j=1} \p{\bmu_j(1) - \frac{\sum_{i}Z_i\bmu_i(1)}{k_1}}^2+ \frac{2}{k_{0}}\sum_{j:Z_j=0} \p{\bmu_j(0) - \frac{\sum_{i}(1-Z_i)\bmu_i(0)}{k_0}}^2\\
&\qquad + \frac{2}{k_{1}}\sum_{j} Z_j(\bY_j(1)-\bM_j(1))^2 + \frac{2}{k_{0}}\sum_{j} (1-Z_j)(\bY_j(0)-\bM_j(0))^2+ o_P(1).\\
\end{split}
\label{eq:jackknife_final}
\end{equation}

Recall that $V = V_0 + V_1 +2V_{01} + \Sigma_\Delta \le 2V_0 + 2V_1 + \Sigma_\Delta$, where
\begin{equation*}
    V_0=\lim_{k\to\infty}\frac{1}{k} \sum_{i = 1}^k \p{\bmu_i(0) - \frac{1}{k} \sum_{j = 1}^k \bmu_j(0)}^2, \quad V_1=\lim_{k\to\infty} \frac{1}{k} \sum_{i = 1}^k \p{\bmu_i(1) - \frac{1}{k} \sum_{j = 1}^k \bmu_j(1)}^2,
\end{equation*}
and recall that
\begin{equation*}
\Sigma_{\Delta,k}:=\frac{2}{k}\sum_{i=1}^k\p{\Var{\bY_i(1)}+\Var{\bY_i(0)}}\to \Sigma_{\Delta}.
\end{equation*}
Note that the sum of the first two terms in \eqref{eq:jackknife_final} is equivalent to the variance estimator of difference-in-means estimator for finite-population causal inference and converges in probability to $2(V_0 + V_1)$ by standard law of large numbers \citep{imbens2015causal}. For the last two terms,
by the same martingale argument in Section~\ref{sec:proof_clt},
\begin{equation*}
\frac{2}{k_{1}}\sum_{j} Z_j(\bY_j(1)-\bM_j(1))^2 + \frac{2}{k_{0}}\sum_{j} (1-Z_j)(\bY_j(0)-\bM_j(0))^2 \to_p \Sigma_{\Delta}.
\end{equation*}
Thus, $k\hV_{\text{DM}}^{(l,b)}\to 2V_0 + 2V_1 + \Sigma_\Delta$.

\section{Proof of Lemmas}

\ifms
\proof{Proof of Lemma \ref{lemm:mse_bias}}
\else
\begin{proof}[Proof of Lemma \ref{lemm:mse_bias}]
\fi

For $i=1,\dots,k$,
\begin{equation}
\begin{split}
&\EE{\frac{1}{l-b}
\sum_{t=b+1}^l\htau_{(i-1)l+t}}\\
=&\EE{\frac{1}{l-b}
\sum_{t=b+1}^l\p{\frac{Y_{(i-1)l+t}Z_i}{k_1/k}-\frac{Y_{(i-1)l+t}(1-Z_i)}{k_0/k} } }\\
=&\frac{1}{l-b}\sum_{t=b+1}^l
\EE{\frac{Z_i\EE{Y_{(i-1)l+t}\cond Z_i=1,Z_{-i}}}{k_1/k}
-\frac{(1-Z_i)\EE{Y_{(i-1)l+t}\cond Z_i=0,Z_{-i}}}{k_0/k}}.
\end{split}
\end{equation}

By the assumption on mixing time, for $t=b+1,\dots,l$ and any $z_{-i}$,
\begin{equation}
\begin{split}
&\abs{\EE{Y_{(i-1)l+t}\cond Z_i=1,Z_{-i}=z_{-i}}
-\EE[\law_1]{Y_{(i-1)l+t}}}\\
=&\left|\int_s
\EE{Y_{(i-1)l+t}\cond Z_i=1,Z_{-i}=z_{-i},
S_{(i-1)l+t}=s}\right.\\
&\qquad\qquad\left.
\cdot\cb{ \PP{S_{(i-1)l+t}=s\cond Z_i=1,Z_{-i}=z_{-i}} -\PP[\law_1]{S_{(i-1)l+t}=s}}ds\right|\\
\le& \Lambda\int_s \abs{ \PP{S_{(i-1)l+t}=s\cond Z_i=1,Z_{-i}=z_{-i}}
-\PP[\law_1]{S_{(i-1)l+t}=s} }ds\\
\le&2\Lambda\exp\p{-(t-1)/t_{\text{mix}}},
\end{split}
\end{equation}
where the last inequality follows from the mixing assumption and the fact that the two distributions have transited under the same sequence of transition operators $\{P_{(i-1)l}^1,P_{(i-1)l+1}^1,\dots,P_{(i-1)l+t-1}^1\}$.
Similarly,
\begin{equation}
\abs{\EE{Y_{(i-1)l+t}\cond Z_i=0,Z_{-i}=z_{-i}}
-\EE[\law_0]{Y_{(i-1)l+t}}}
\le 2\Lambda\exp\p{-(t-1)/t_{\text{mix}}}.
\end{equation}
Thus,
\begin{equation}
\begin{split}
\abs{\EE{\frac{1}{l-b}
\sum_{t=b+1}^l \htau_{(i-1)l+t}}- \frac{1}{l-b}\sum_{t=b+1}^l \tau_{(i-1)l+t}}
\le
\frac{4\Lambda}{l-b}\sum_{t=b+1}^l\exp\p{-(t-1)/t_{\text{mix}}}
+\oo\p{2^{-k}}.
\end{split}
\end{equation}
Therefore,
\begin{equation}
\begin{split}
\abs{\EE{\htau_{\text{DM}}^{(l,b)}}-\tau^{(l,b)}_\FATE}
&\le \frac{1}{k}\sum_{i=1}^{k}\abs{\EE{\frac{1}{l-b}\sum_{t=b+1}^l\htau_{(i-1)l+t}}- \frac{1}{l-b}\sum_{t=b+1}^l\tau_{(i-1)l+t}}\\
&\le
\frac{4\Lambda}{l-b}\sum_{t=b+1}^l\exp\p{-(t-1)/t_{\text{mix}}}
+\oo\p{2^{-k}}\\
&\le \frac{4\Lambda}{1-\exp\p{-1/t_\mix}}\cdot\frac{\exp\p{-b/t_{\text{mix}} }}{l-b}
+\oo\p{2^{-k}},
\end{split}
\label{eq:mixing_bias}
\end{equation}
where the $\oo\p{2^{-k}}$ term comes from the events $k_1=0$ or $k_0=0$.
Meanwhile, since the maximum difference in treatment effects is bounded by a constant $\Psi$,
\begin{equation}
\begin{split}
\abs{\tau_{\FATE}^{(l, b)} - \tau_{\GATE}}
&\le \frac{1}{k}\sum_{i=1}^{k}\abs{\frac{1}{l-b}\sum_{t=b+1}^l\tau_{(i-1)l+t}- \frac{1}{l}\sum_{t=1}^b\tau_{(i-1)l+t}-\frac{1}{l}\sum_{t=b+1}^l\tau_{(i-1)l+t}}\\
&\le \frac{1}{kl}\sum_{i=1}^{k}\abs{\sum_{t=1}^b \tau_{(i-1)l+t}-\sum_{t=b+1}^l\frac{b}{l-b}\tau_{(i-1)l+t} }\\
&\le \Psi\cdot\frac{ b}{l}.
\end{split}
\label{eq:burnin_bias}
\end{equation}

\ifms
\endproof
\else
\end{proof}
\fi

\ifms
\proof{Proof of Lemma \ref{lemm:mse_var}}
\else
\begin{proof}[Proof of Lemma \ref{lemm:mse_var}]
\fi

To start with, we decompose the variance into two parts as
\begin{equation}
\begin{split}
\Var{\htau_{\text{DM}}^{(l,b)}}
&=\frac{1}{k^2(l-b)^2}\Var{
\sum_{i=1}^{k}\sum_{t=b+1}^l\htau_{(i-1)l+t}}\\
&=\frac{1}{k^2(l-b)^2}\EE{\Var{
\sum_{i=1}^{k}\sum_{t=b+1}^l\htau_{(i-1)l+t}\cond W_{1:T},S_{1:T}}}+\\
&\qquad\qquad\frac{1}{k^2(l-b)^2}\Var{\EE{
\sum_{i=1}^{k}\sum_{t=b+1}^l\htau_{(i-1)l+t}\cond W_{1:T},S_{1:T}}}.
\label{eq:var_htau}
\end{split}
\end{equation}
To bound the first term of (\ref{eq:var_htau}), note that after conditioning on $\{W_{1:T},S_{1:T}\}$, $\htau_t$ are uncorrelated of each other, and thus
\begin{equation}
\begin{split}
&\frac{1}{k^2(l-b)^2}\EE{\Var{
\sum_{i=1}^{k}\sum_{t=b+1}^l\htau_{(i-1)l+t}\cond W_{1:T},S_{1:T}}}\\
&\qquad\qquad=\frac{1}{k^2(l-b)^2}\EE{\sum_{i=1}^{k}\sum_{t=b+1}^l\Var{
\htau_{(i-1)l+t}\cond W_{1:T},S_{1:T}}}\\
&\qquad\qquad\le\frac{1}{k^2(l-b)^2}\EE{\sum_{i=1}^{k}\sum_{t=b+1}^l\sigma^2\p{\frac{kZ_i}{k_1}-\frac{k(1-Z_i)}{k_0}}^2}\\
&\qquad\qquad=\frac{\sigma^2}{k(l-b)}\EE{\frac{k^2Z_i^2}{k_1^2}+\frac{k^2(1-Z_i)^2}{k_0^2}-\frac{2k^2Z_i(1-Z_i)}{k_1k_0}}\\
&\qquad\qquad=\frac{\sigma^2}{k(l-b)}\EE{\frac{k^2\EE{Z_i\cond k_1}}{k_1^2}+\frac{k^2\EE{1-Z_i\cond k_1}}{k_0^2}}\\
&\qquad\qquad=\frac{\sigma^2}{k(l-b)}\EE{\frac{k}{\max\p{k_1,1}}+\frac{k}{\max\p{k_0,1}}}+\oo\p{2^{-k}}.
\end{split}
\end{equation}
By a Taylor expansion of $k/\max\p{k_1,1}$ at $\max\p{k_1,1}=0.5k$,
\begin{equation}
\EE{\frac{k}{\max\p{k_1,1}}}=\frac{k}{0.5k}+\frac{ 0.5(1-0.5)k^2}{(0.5)^3k^3}+\smallO\p{\frac{1}{k}}+\oo\p{2^{-k}}=2+\frac{2}{k}+\smallO\p{\frac{1}{k}}.
\label{eq:taylor}
\end{equation}
Similarly, $\EE{k/{\max\p{k_0,1}}}=2+2/k+\smallO\p{1/{k}}.$
Thus,
\begin{equation}
\begin{split}
&\frac{1}{k^2(l-b)^2}\EE{\Var{\sum_{i=1}^{k}\sum_{t=b+1}^l\htau_{(i-1)l+t}\cond W_{1:T},S_{1:T}}}\\
&\qquad\qquad\le\frac{\sigma^2}{k(l-b)}\EE{\frac{k}{\max\p{k_1,1}}+\frac{k}{\max\p{k_0,1}}}+\oo\p{2^{-k}}\\
&\qquad\qquad\le\frac{4\sigma^2}{k(l-b)}+\oo\p{\frac{1}{k^2(l-b)}}+\oo\p{2^{-k}}.
\label{eq:var_htau_noise}
\end{split}
\end{equation}

To bound the second term of (\ref{eq:var_htau}), denote for shorthand
\[A_i=\frac{1}{l-b}\sum_{t=b+1}^l\mu_{(i-1)l+t}\p{S_{(i-1)l+t},Z_i}.\]
Then the second term of (\ref{eq:var_htau}) is
\[\Var{\frac{1}{k_1}\sum_{i=1}^k Z_iA_i-\frac{1}{k_0}\sum_{i=1}^k(1-Z_i)A_i}.\]
We first compare this quantity with the variance of the Horvitz-Thompson estimator
\begin{equation*}
\Var{\frac{2}{k}\sum_{i=1}^k Z_iA_i-\frac{2}{k}\sum_{i=1}^k(1-Z_i)A_i}
= \Var{\frac{2}{k}\sum_{i=1}^k(2Z_i-1)A_i}.
\end{equation*}
Note that
\begin{equation}
\begin{split}
&\frac{1}{k_1}\sum_{i=1}^k Z_iA_i
-\frac{1}{k_0}\sum_{i=1}^k(1-Z_i)A_i
-\frac{2}{k}\sum_{i=1}^k(2Z_i-1)A_i\\
&\qquad=\p{1-\frac{2k_1}{k}}
\left\{ \frac{1}{k_1}\sum_{i=1}^k Z_iA_i+ \frac{1}{k_0}\sum_{i=1}^k(1-Z_i)A_i \right\}.
\end{split}
\end{equation}
Since $\abs{A_i}\le\Lambda$ a.s. and $k_1\sim\operatorname{Binomial}(k,1/2)$,
\begin{equation}
\begin{split}
\EE{\p{\frac{1}{k_1}\sum_{i=1}^k Z_iA_i -\frac{1}{k_0}\sum_{i=1}^k(1-Z_i)A_i
-\frac{2}{k}\sum_{i=1}^k(2Z_i-1)A_i}^2}
&\le  4\Lambda^2 \EE{\p{1-\frac{2k_1}{k}}^2}\\
&= \frac{4\Lambda^2}{k}.
\label{eq:focal_var_bound}
\end{split}
\end{equation}
Therefore,
\begin{equation}
\begin{split}
\Var{\frac{1}{k_1}\sum_{i=1}^k Z_iA_i
-\frac{1}{k_0}\sum_{i=1}^k(1-Z_i)A_i}\le 2\Var{\frac{2}{k}\sum_{i=1}^k(2Z_i-1)A_i} +\frac{8\Lambda^2}{k}.
\end{split}
\label{eq:var_upper_bound}
\end{equation}

It remains to bound the first term of~\eqref{eq:var_upper_bound}. Note that
\begin{equation}
\begin{split}
\Var{\frac{2}{k}\sum_{i=1}^k(2Z_i-1)A_i}
=&\frac{4}{k^2}\sum_{i=1}^k
\Var{(2Z_i-1)A_i}\\
&+\frac{8}{k^2}\sum_{i<j}
\Cov{(2Z_i-1)A_i,(2Z_j-1)A_j}\\
\le& \frac{4\Lambda^2}{k} +\frac{8}{k^2}\sum_{i<j} \abs{\Cov{(2Z_i-1)A_i,(2Z_j-1)A_j}},
\end{split}
\end{equation}
For $i<j$, by the definition of $A_i$,
\begin{equation}
\begin{split}
&\Cov{(2Z_i-1)A_i,(2Z_j-1)A_j}\\
=&\frac{1}{(l-b)^2}\sum_{r=b+1}^l\sum_{s=b+1}^l\Cov{(2Z_i-1)\mu_{(i-1)l+r}(S_{(i-1)l+r},Z_i),(2Z_j-1)\mu_{(j-1)l+s}(S_{(j-1)l+s},Z_j)}.
\end{split}
\end{equation}
where
\begin{equation}
\begin{split}
&\Cov{(2Z_i-1)\mu_{(i-1)l+r}(S_{(i-1)l+r},Z_i),(2Z_j-1)\mu_{(j-1)l+s}(S_{(j-1)l+s},Z_j)}\\
=&\Cov{(2Z_i-1)\mu_{(i-1)l+r}(S_{(i-1)l+r},Z_i),\EE{(2Z_j-1)\mu_{(j-1)l+s}(S_{(j-1)l+s},Z_j)\cond Z_i,S_{(i-1)l+r},S_{il+1}}}\\
=&\Cov{(2Z_i-1)\mu_{(i-1)l+r}(S_{(i-1)l+r},Z_i),\EE{(2Z_j-1)\mu_{(j-1)l+s}(S_{(j-1)l+s},Z_j)\cond S_{il+1}}},
\label{eq:cov_simplified}
\end{split}
\end{equation}
Since $Z_{i+1},\ldots,Z_j$ are independent of $S_{il+1}$, for any $x,x'\in\mathcal S$,
\begin{equation}
\begin{split}
&\abs{\EE{(2Z_j-1)\mu_{(j-1)l+s}(S_{(j-1)l+s},Z_j)\cond S_{il+1}=x}-\EE{(2Z_j-1)\mu_{(j-1)l+s}(S_{(j-1)l+s},Z_j)\cond S_{il+1}=x'}}\\
\le&\mathbb{E}\left[\left|\EE{(2Z_j-1)\mu_{(j-1)l+s}(S_{(j-1)l+s},Z_j)\cond S_{il+1}=x,Z_{i+1},\ldots,Z_j}\right.\right.\\
&\qquad \left.\left.-\EE{(2Z_j-1)\mu_{(j-1)l+s}(S_{(j-1)l+s},Z_j)\cond S_{il+1}=x',Z_{i+1},\ldots,Z_j}\right|\right]\\
\le&2\Lambda\exp\p{-\frac{(j-i-1)l+s-1}{t_{\rm mix}}},
\end{split}
\end{equation}
where the last inequality follows from Assumption~\ref{assu:mixing}.
Since $\abs{(2Z_i-1)\mu_{(i-1)l+r}(S_{(i-1)l+r},Z_i)}\le\Lambda$, by Popoviciu's inequality on variances,
\begin{equation}
\begin{split}
\eqref{eq:cov_simplified}&=\abs{\Cov{(2Z_i-1)\mu_{(i-1)l+r}(S_{(i-1)l+r},Z_i),\EE{(2Z_j-1)\mu_{(j-1)l+s}(S_{(j-1)l+s},Z_j)\cond S_{il+1}}}}\\
&\le \frac{1}{4}\p{2\Lambda}\sup_{x,x'\in\mathcal S}\left|\EE{(2Z_j-1)\mu_{(j-1)l+s}(S_{(j-1)l+s},Z_j)\cond S_{il+1}=x}\right.\\
&\qquad\qquad\left.-\EE{(2Z_j-1)\mu_{(j-1)l+s}(S_{(j-1)l+s},Z_j)\cond S_{il+1}=x'}\right|\\
&\le\Lambda^2\exp\p{-\frac{(j-i-1)l+s-1}{t_{\rm mix}}}.
\end{split}
\end{equation}
Therefore,
\begin{equation}
\begin{split}
&\abs{\Cov{(2Z_i-1)A_i,(2Z_j-1)A_j}}\\
\le&\frac{\Lambda^2}{(l-b)^2}\sum_{r=b+1}^l\sum_{s=b+1}^l\exp\p{-\frac{(j-i-1)l+s-1}{t_{\rm mix}}}\\
=&\frac{\Lambda^2}{l-b}\exp\p{-\frac{(j-i-1)l+b}{t_{\rm mix}}}\frac{1-\exp\p{-(l-b)/t_{\rm mix}}}{1-\exp\p{-1/t_{\rm mix}}}.
\end{split}
\end{equation}
It then follows that
\begin{equation}
\begin{split}
&\sum_{i<j}\abs{\Cov{(2Z_i-1)A_i,(2Z_j-1)A_j}}\\
\le&\frac{\Lambda^2}{l-b}\frac{\exp\p{-b/t_{\rm mix}}\p{1-\exp\p{-(l-b)/t_{\rm mix}}}}{1-\exp\p{-1/t_{\rm mix}}}\sum_{1\le i<j\le k}\exp\p{-\frac{(j-i-1)l}{t_{\rm mix}}}\\
\le&\frac{\Lambda^2}{l-b}\frac{\exp\p{-b/t_{\rm mix}}\p{1-\exp\p{-(l-b)/t_{\rm mix}}}}{1-\exp\p{-1/t_{\rm mix}}} \cdot \frac{k}{1-\exp\p{-l/t_{\rm mix}}}\\
\le&\frac{k\Lambda^2}{l-b}\frac{\exp\p{-b/t_{\rm mix}}}{1-\exp\p{-1/t_{\rm mix}}},
\end{split}
\end{equation}
where the last inequality follows from $l-b\le l$.
Thus, by~\eqref{eq:var_upper_bound},
\begin{equation}
\Var{\frac{1}{k_1}\sum_{i=1}^kZ_iA_i-\frac{1}{k_0}\sum_{i=1}^k(1-Z_i)A_i}
\le\frac{16\Lambda^2}{k}+\frac{16\Lambda^2}{k(l-b)}\frac{\exp\p{-b/t_{\rm mix}}}{1-\exp\p{-1/t_{\rm mix}}}.
\end{equation}

\ifms
\endproof
\else
\end{proof}
\fi

\ifms
\proof{Proof of Lemma \ref{lemm:clt_decom}}
\else
\begin{proof}[Proof of Lemma \ref{lemm:clt_decom}]
\fi

We work on the event that both treatment arms are nonempty, whose complement has probability $2^{-k+1}$.
We start by decomposing our estimator as follows:
\begin{equation}
\label{eq:Delta_def}
\begin{split}
\htau_{\text{DM}}^{(l,b)}&=\frac{1}{k}\sum_{i=1}^{k}\sqb{\frac{\bY_i(1)Z_i}{k_1/k}-\frac{\bY_i(0)(1-Z_i)}{1-k_1/k}}\\
&=\frac{1}{k}\sum_{i=1}^{k}\sqb{\frac{\bM_i(1)Z_i}{k_1/k}-\frac{\bM_i(0)(1-Z_i)}{1-k_1/k}} + \Delta^{(l,b)}.
\end{split}
\end{equation}
Now, by Assumption \ref{assu:mixing} on mixing time,
\begin{equation}
\begin{split}
\abs{\bM_i(w) - \bmu_i(w)}
&= \abs{\EE{\bY_i\cond Z_i=w}-\EE{\bY_i\cond Z_i=w, Z_{i-1}=w,\cdots}}\\
&\le \frac{1}{l-b}\sum_{t=(i-1)l+b+1}^{il}\abs {\EE{Y_t\cond  Z_i=w}-\EE{Y_t\cond Z_i=w, Z_{i-1}=w,\cdots}}\\
&\le \frac{2\Lambda}{l-b}\sum_{r=b}^{l-1}\exp\p{-r/t_\mix}\\
&\le \frac{2\Lambda}{1-\exp\p{-1/t_\mix}}\cdot\frac{\exp\p{-b/t_\mix}}{l-b}.
\label{eq:mean_diff}
\end{split}
\end{equation}
uniformly across $i = 1, \, \ldots, \, k$ and $w = 0, \, 1$.
Thus, continuing from \eqref{eq:Delta_def} and recalling \eqref{eq:FATE_recall}, we see that
\begin{equation}
\begin{split}
\htau_{\text{DM}}^{(l,b)}
&=\frac{1}{k}\sum_{i=1}^{k}\sqb{\frac{\bmu_i(1)Z_i}{k_1/k}-\frac{\bmu_i(0)(1-Z_i)}{1-k_1/k}} + \Delta^{(l,b)} + \oop\p{\frac{e^{-b/t_{\text{mix}}}}{l - b}}, \\
&= \tau_{\FATE}^{(l, b)} + \Gamma^{(l,b)} + \Delta^{(l,b)} + \oop\p{\frac{e^{-b/t_{\text{mix}}}}{l - b}}.
\end{split}
\end{equation}

\ifms
\endproof
\else
\end{proof}
\fi

\ifms
\proof{Proof of Lemma \ref{lemm:decomposition}}
\else
\begin{proof}[Proof of Lemma \ref{lemm:decomposition}]
\fi

We first derive a first-order expansion of $\Gamma^{(l,b)}$ up to terms negligible asymptotically.
Again, we work on the event $k_0,k_1>0$ whose complement has probability
$2^{-k+1}$. Since
$\sum_{i=1}^k\{\bmu_i(w)-k^{-1}\sum_{j=1}^k\bmu_j(w)\}=0$,
we have
\begin{equation}
\Gamma^{(l,b)}
=\sum_{w=0}^1 \frac{k}{k_w}\frac{1}{k}\sum_{i=1}^k\left(Z_i-\frac{1}{2}\right) \left(\bmu_i(w)-\frac{1}{k}\sum_{j=1}^k\bmu_j(w)\right).
\end{equation}
Since the conditional expectations are
bounded almost surely, for $w=0,1$, 
\begin{equation}
\frac{1}{k}\sum_{i=1}^k
\left(Z_i-\frac{1}{2}\right)
\left(\bmu_i(w)-\frac{1}{k}\sum_{j=1}^k\bmu_j(w)\right)
=\oop(k^{-1/2}),
\qquad \frac{k}{k_w}=2+\oop(k^{-1/2}).
\end{equation}
Thus, replacing each factor $k/k_w$ by $2$ only contributes an
$\oop(k^{-1})$ error. As a result, we can write
\begin{equation}
\sqrt{k}\,\Gamma^{(l,b)}
=\frac{1}{\sqrt{k}}\sum_{i=1}^k(2Z_i-1)\p{\bmu_i(1)+\bmu_i(0)-\frac{1}{k}\sum_{j=1}^k\{\bmu_j(1)+\bmu_j(0)\}}+o_P(1).
\label{eq:lemma14_decom}
\end{equation}

Next, we write $\Delta^{(l,b)}$ as the sum of martingale differences up to an asymptotically negligible error. 
Denote for shorthand $Z:=(Z_1,\ldots,Z_k)$.
It follows from Assumption~\ref{assu:mixing} that
\begin{equation}
\begin{split}
\max_{1\le i\le k}\abs{\EE{\bY_i\cond\mathcal F_{(i-1)l+1}, Z}-\bM_i(Z_i)}
&=\max_{1\le i\le k}\abs{\EE{\bY_i\cond\mathcal F_{(i-1)l+1}, Z}-\EE{\bY_i\cond Z_i}}\\
&\le \frac{2\Lambda}{l-b}\sum_{t=b}^{l-1}e^{-t/t_{\mix}}\\
&=\oo\p{\frac{e^{-b/t_{\mix}}}{l-b}}.
\end{split}
\label{eq:lemma15_centering}
\end{equation}
As a result, we can write $\Delta^{(l,b)}$ as 
\begin{equation}
\begin{split}
\Delta^{(l,b)}
&=\frac{1}{k}\sum_{i=1}^{k}\sqb{\frac{\cb{\bY_i(1)-\bM_i(1)}Z_i}{k_1/k}-\frac{\cb{\bY_i(0)-\bM_i(0)}(1-Z_i)}{1-k_1/k}}\\
&=\frac{1}{k}\sum_{i=1}^k\left\{\frac{Z_i}{k_1/k}-\frac{1-Z_i}{k_0/k}\right\}\left(\bY_i-\EE{\bY_i\cond\mathcal F_{(i-1)l+1},Z}\right)+O_P\left(\frac{e^{-b/t_{\mix}}}{l-b}\right)\\
&=\frac{2}{k}\sum_{i=1}^k(2Z_i-1)\left(\bY_i-\EE{\bY_i\cond\mathcal F_{(i-1)l+1},Z}\right)\\
&\qquad+\left(\frac{k}{k_1}-2\right)\frac{1}{k}\sum_{i=1}^kZ_i\left(\bY_i-\EE{\bY_i\cond\mathcal F_{(i-1)l+1},Z}\right)\\
&\qquad-\left(\frac{k}{k_0}-2\right)\frac{1}{k}\sum_{i=1}^k(1-Z_i)\left(\bY_i-\EE{\bY_i\cond\mathcal F_{(i-1)l+1},Z}\right)+O_P\left(\frac{e^{-b/t_{\mix}}}{l-b}\right)\\
&=\frac{1}{k}\sum_{i=1}^kD_i+\oo_P(k^{-1})+O_P\left(\frac{e^{-b/t_{\mix}}}{l-b}\right).
\end{split}
\label{eq:lemma15_approximation}
\end{equation}

\ifms
\endproof
\else
\end{proof}
\fi

\ifms
\proof{Proof of Lemma \ref{lemm:clt_moments}}
\else
\begin{proof}[Proof of Lemma \ref{lemm:clt_moments}]
\fi

Again, we condition on the event $k_0,k_1>0$ throughout.
First of all, \eqref{eq:clt_moments_Gamma} follows directly from Assumption~\ref{assu:mom_limits}:
\begin{equation}
\frac{1}{k}\sum_{i=1}^k\p{\bmu_i(1)+\bmu_i(0)-\frac{1}{k}\sum_{j=1}^k\{\bmu_j(1)+\bmu_j(0)\}}^2
\to V_0+V_1+2V_{01}
=\Sigma_\Gamma.
\end{equation}

Next, we show that~\eqref{eq:clt_moments_Delta} holds. 
We start by calculating the conditional variance of $D_i$. Since $(2Z_i-1)^2=1$,
\begin{equation}
\EE{D_i^2\cond\mathcal F_{(i-1)l+1},Z}=4\Var{\bY_i\cond\mathcal F_{(i-1)l+1},Z}.
\end{equation}
We then show that the noise terms are conditionally uncorrelated with each other. 
By Assumption~\ref{assu:mdp}, $\EE{\epsilon_t\cond\mathcal F_t,Z}=0$.
Then, for any two focal times $s<t$ in the same block $i$,
\begin{equation}
\EE{\epsilon_t\epsilon_s\cond\mathcal F_{(i-1)l+1},Z}=\EE{\epsilon_s\EE{\epsilon_t\cond\mathcal F_t,Z}\cond\mathcal F_{(i-1)l+1},Z}=0,
\label{eq:lemma15_noise1}
\end{equation}
Similarly,
\begin{equation}
\EE{\epsilon_t\mu_s(S_s,Z_i)\cond\mathcal F_{(i-1)l+1},Z}=0.
\end{equation}
On the other hand, for two focal times $s>t$, again bu Assumption~\ref{assu:mdp},
\begin{equation}
\begin{split}
\EE{\epsilon_t\mu_s(S_s,Z_i)\cond\mathcal F_{(i-1)l+1},Z}
&=\EE{\epsilon_t\EE{\mu_s(S_s,Z_i)\cond\mathcal F_t,Y_t,Z}\cond\mathcal F_{(i-1)l+1},Z}\\
&=\EE{\EE{\epsilon_t\cond\mathcal F_t,Z}\EE{\mu_s(S_s,Z_i)\cond\mathcal F_t,Z}\cond\mathcal F_{(i-1)l+1},Z}=0.
\end{split}
\label{eq:lemma15_noise2}
\end{equation}
Thus expanding $\bY_i$ gives
\begin{equation}
\begin{split}
\frac{1}{k}\sum_{i=1}^k\EE{D_i^2\cond\mathcal F_{(i-1)l+1},Z}
&=\frac{4}{k}\sum_{i=1}^k\Var{\frac{1}{l-b}\sum_{t=(i-1)l+b+1}^{il}\mu_t(S_t,Z_i)\cond\mathcal F_{(i-1)l+1},Z}\\
&\qquad+\frac{4}{k(l-b)^2}\sum_{i=1}^k\sum_{t=(i-1)l+b+1}^{il}\EE{\epsilon_t^2\cond\mathcal F_{(i-1)l+1},Z}.
\end{split}
\label{eq:lemma15_cond_var}
\end{equation}
We first consider the first term above in~\eqref{eq:lemma15_cond_var}. By Assumption~\ref{assu:mixing}, for all $i$,
\begin{equation}
\begin{split}
&\frac{4}{k}\sum_{i=1}^k\Var{\frac{1}{l-b}\sum_{t=(i-1)l+b+1}^{il}\mu_t(S_t,Z_i)\cond\mathcal F_{(i-1)l+1},Z}\\
&\qquad=\frac{4}{k}\sum_{i=1}^k\Var{\frac{1}{l-b}\sum_{t=(i-1)l+b+1}^{il}\mu_t(S_t,Z_i)\cond Z_i}
+\oo(e^{-b/t_{\mix}})\\
&\qquad=\frac{4}{k} \sum_{i=1}^k\left(Z_i \Var{\frac{1}{l-b}\sum_{t=(i-1)l+b+1}^{il}\mu_t(S_t,Z_i)\cond Z_i=1}\right.\\
&\qquad\qquad \left.  +(1-Z_i) \Var{\frac{1}{l-b}\sum_{t=(i-1)l+b+1}^{il}\mu_t(S_t,Z_i)\cond Z_i=0}\right)+\oop(e^{-b/t_\mix})\\
&\qquad=\frac{2}{k}\sum_{i=1}^k\sum_{w=0}^1\Var{\frac{1}{l-b}\sum_{t=(i-1)l+b+1}^{il}\mu_t(S_t,w)\cond Z_i=w}+\oop(k^{-1/2}+e^{-b/t_\mix}),
\end{split}
\label{eq:lemma15_cond_var1}
\end{equation}
We next consider the noise term in~\eqref{eq:lemma15_cond_var}. 
Note that
\begin{equation}
\EE{\frac{1}{(l-b)^2}\sum_{t=(i-1)l+b+1}^{il}\epsilon_t^2-\EE{\frac{1}{(l-b)^2}\sum_{t=(i-1)l+b+1}^{il}\epsilon_t^2\cond\mathcal F_{(i-1)l+1},Z}\cond\mathcal F_{(i-1)l+1},Z}=0
\end{equation}
by construction. Since the errors are bounded almost surely,
\begin{equation}
\frac{1}{k}\sum_{i=1}^k\left\{\frac{1}{(l-b)^2}\sum_{t=(i-1)l+b+1}^{il}\epsilon_t^2-\EE{\frac{1}{(l-b)^2}\sum_{t=(i-1)l+b+1}^{il}\epsilon_t^2\cond\mathcal F_{(i-1)l+1},Z}\right\}=\oop(k^{-1/2}),
\end{equation}
Furthermore, since the noises $\epsilon_t$ are mutually independent, 
\begin{equation}
\frac{1}{k}\sum_{i=1}^k\left\{\frac{1}{(l-b)^2}\sum_{t=(i-1)l+b+1}^{il}\epsilon_t^2-\EE{\frac{1}{(l-b)^2}\sum_{t=(i-1)l+b+1}^{il}\epsilon_t^2}\right\}=\oop(k^{-1/2}).
\end{equation}
Subtracting them gives
\begin{equation}
\frac{1}{k}\sum_{i=1}^k\EE{\frac{1}{(l-b)^2}\sum_{t=(i-1)l+b+1}^{il}\epsilon_t^2\cond\mathcal F_{(i-1)l+1},Z}
=
\frac{1}{k}\sum_{i=1}^k\EE{\frac{1}{(l-b)^2}\sum_{t=(i-1)l+b+1}^{il}\epsilon_t^2}+\oop(k^{-1/2}).
\end{equation}
It then follows directly by the law of total expectation (as what we did for~\eqref{eq:lemma15_cond_var1}) that
\begin{equation}
\begin{split}
&\frac{4}{k(l-b)^2}\sum_{i=1}^k\sum_{t=(i-1)l+b+1}^{il}\EE{\epsilon_t^2\cond\mathcal F_{(i-1)l+1},Z}\\
&\qquad=\frac{2}{k(l-b)^2}\sum_{i=1}^k\sum_{t=(i-1)l+b+1}^{il}\p{\EE{\epsilon_t^2\cond Z_i=1}+\EE{\epsilon_t^2\cond Z_i=0}}+\oop(k^{-1/2}).
\end{split}
\label{eq:lemma15_cond_var2}
\end{equation}
Combining~\eqref{eq:lemma15_cond_var1} and~\eqref{eq:lemma15_cond_var2} and recalling that $\EE{\epsilon_t\mu_s(S_s,Z_i)\cond\mathcal F_{(i-1)l+1},Z}=0$,
\begin{equation}
\frac{1}{k}\sum_{i=1}^k\EE{D_i^2\cond\mathcal F_{(i-1)l+1},Z}
=
\frac{2}{k}\sum_{i=1}^k\p{\Var{\bY_i(1)}+\Var{\bY_i(0)}}
+\oop(k^{-1/2}+e^{-b/t_\mix})
\to_p\Sigma_\Delta.
\end{equation}

\ifms
\endproof
\else
\end{proof}
\fi

\ifms
\proof{Proof of Lemma \ref{lemm:mse_bias_bc}}
\else
\begin{proof}[Proof of Lemma \ref{lemm:mse_bias_bc}]
\fi

Since $\abs{\EE{Y_t|S_t,W_t}}= \abs{\mu_t(S_t, W_t)}\le_{a.s.}\Lambda$, 
\begin{equation}
\begin{split}
\abs{\EE{\widetilde{\htau_{\text{BC}}^{(l,b)}}}
-\EE{\htau_{\text{BC}}^{(l,b)}}}
\le&\abs{\frac{1}{l}\sum_{t=1}^b
\EE{\frac{Z_1Z_0Y_t}{k_{11}}
-\frac{(1-Z_1)(1-Z_0)Y_t}{k_{00}}}}+\oo\p{\frac{1}{k}}\\
\le&\frac{\Lambda b}{l}
\left\{
\EE{\frac{I\{k_{11}>0\}}{k_{11}}}
+
\EE{\frac{I\{k_{00}>0\}}{k_{00}}}
\right\}+\oo\p{\frac{1}{k}}\\
=&\oo\p{\frac{\Lambda b}{kl}}+\oo\p{\frac{1}{k}}
=\oo\p{\frac{1}{k}},
\end{split}
\end{equation}
where the last equality follows from that
$\EE{I\{k_{11}>0\}/k_{11}}=\oo(1/k)$ and
$\EE{I\{k_{00}>0\}/k_{00}}=\oo(1/k)$, and the $\oo\p{1/k}$ remainder accounts for the events $k_{11}=0$ or $k_{00}=0$.

Note that, conditioning on the full treatment assignment vector $Z$,
\begin{equation}
\begin{split}
&\EE{\sum_{i=1}^k
\frac{1}{k_{11}}\sum_{t = 1}^b
\frac{Z_iZ_{i-1}Y_{(i - 1)l + t}}{l}}\\
=&\frac{1}{l}\sum_{i=1}^k\sum_{t=1}^b
\EE{
\frac{Z_iZ_{i-1}}{k_{11}}
\EE{Y_{(i-1)l+t}\cond Z}
}\\
=&\frac{1}{l}\sum_{i=1}^k\sum_{t=1}^b
\EE{
\frac{
Z_iZ_{i-1}
\EE{Y_{(i-1)l+t}\cond Z_i=Z_{i-1}=1,Z_{-(i-1,i)}}
}{k_{11}}
},
\end{split}
\end{equation}
and by the assumption on mixing time, for $t=1,...,b$ and any $z_{-(i-1,i)}$,
\begin{equation*}
\begin{split}
&\abs{\EE{Y_{(i-1)l+t}\cond Z_i=Z_{i-1}=1,Z_{-(i-1,i)}=z_{-(i-1,i)}}-\EE[\law_1]{Y_{(i-1)l+t}}}\\
=&\left|\int_s
\EE{Y_{(i-1)l+t}\cond Z_i=Z_{i-1}=1,Z_{-(i-1,i)}=z_{-(i-1,i)},S_{(i-1)l+t}=s}
\right.\\
&\qquad\qquad\left.
\cdot\left\{
\PP{S_{(i-1)l+t}=s\cond Z_i=Z_{i-1}=1,Z_{-(i-1,i)}=z_{-(i-1,i)}}
-\PP[\law_1]{S_{(i-1)l+t}=s}
\right\}ds\right|\\
\le&\Lambda\int_s
\abs{
\PP{S_{(i-1)l+t}=s\cond Z_i=Z_{i-1}=1,Z_{-(i-1,i)}=z_{-(i-1,i)}}
-\PP[\law_1]{S_{(i-1)l+t}=s}
}ds\\
\le&2\Lambda\exp\p{-\frac{l+t-1}{t_{\text{mix}}}}.
\end{split}
\end{equation*}
Moreover, uniformly over $i=1,\ldots,k$,
\begin{equation*}
\begin{split}
\EE{\frac{Z_iZ_{i-1}}{k_{11}}}
&=
\frac{\EE{Z_iZ_{i-1}}}{\EE{k_{11}}}
-\frac{\operatorname{Cov}(Z_iZ_{i-1},k_{11})}{\EE{k_{11}}^2}
+\oo\p{\frac{1}{k^2}}\\
&=\frac{1}{k}+\oo\p{\frac{1}{k^2}}.
\end{split}
\end{equation*}
Therefore,
\begin{equation*}
\begin{split}
&\abs{
\EE{\sum_{i=1}^k
\frac{Z_iZ_{i-1}}{k_{11}}
\sum_{t=1}^b\frac{Y_{(i-1)l+t}}{l}}
-\frac{1}{kl}\sum_{i=1}^k\sum_{t=1}^b
\EE[\law_1]{Y_{(i-1)l+t}}
}\\
\le&
\frac{2\Lambda}{l}\sum_{t=1}^b
\exp\p{-\frac{l+t-1}{t_{\text{mix}}}}
\sum_{i=1}^k
\EE{\frac{Z_iZ_{i-1}}{k_{11}}}+
\frac{\Lambda}{l}\sum_{t=1}^b\sum_{i=1}^k
\abs{
\EE{\frac{Z_iZ_{i-1}}{k_{11}}}
-\frac{1}{k}
}\\
\le&
\frac{2\Lambda}{l}\sum_{t=1}^b
\exp\p{-\frac{l+t-1}{t_{\text{mix}}}}
+\oo\p{\frac{1}{k}}.
\end{split}
\end{equation*}
Similarly,
\begin{equation*}
\begin{split}
\abs{\EE{Y_{(i-1)l+t}\cond Z_i=Z_{i-1}=0,Z_{-(i-1,i)}=z_{-(i-1,i)}}-\EE[\law_0]{Y_{(i-1)l+t}}}
 \le2\Lambda\exp\p{-\frac{l+t-1}{t_{\text{mix}}}}+\oo\p{\frac{1}{k}}.
\end{split}
\end{equation*}

Since 
\begin{equation}
\begin{split}
\tau_{\GATE}
&= \frac{l-b}{l}\tau_{\FATE}^{(l, b)}+ \frac{1}{kl}\sum_{i=1}^k\sum_{t=1}^b\tau_{(i-1)l+t},
\end{split}
\end{equation}
combining the results above with Lemma \ref{lemm:mse_bias} yields
\begin{equation}
\begin{split}
&\abs{\EE{\widetilde{\htau_{\text{BC}}^{(l,b)}}}-\tau_{\GATE}}\\
=&\left|\frac{l-b}{l}\p{\EE{\htau_{\text{DM}}^{(l,b)}}-\tau_{\FATE}^{(l, b)}}+
\EE{\sum_{i=1}^k
\frac{Z_iZ_{i-1}}{k_{11}}\sum_{t = 1}^b \frac{Y_{(i - 1)l + t}}{l}}-\right.\\
&\qquad\qquad \left.\EE{\sum_{i=1}^k
\frac{(1-Z_i)(1-Z_{i-1})}{k_{00}}\sum_{t = 1}^b \frac{Y_{(i - 1)l + t}}{l}}-\frac{1}{kl}\sum_{i=1}^k\sum_{t=1}^b\tau_{(i-1)l+t}\right|\\
\le & 4\Lambda \p{
\frac{1}{l}\sum_{t=b+1}^l\exp\p{-\frac{t-1}{t_{\text{mix}}}}
+ \frac{1}{l}\sum_{t=1}^b\exp\p{-\frac{l+t-1}{t_{\text{mix}}}}}
+\oo\p{\frac{1}{k}}\\
\le & \frac{4\Lambda}{1-\exp\p{-1/t_\mix}}\cdot
\frac{\exp\p{-b/t_{\text{mix}} }}{l}
+\oo\p{\frac{1}{k}}.
\end{split}
\end{equation}

\ifms
\endproof
\else
\end{proof}
\fi

\ifms
\proof{Proof of Lemma \ref{lemm:mse_var_bc}}
\else
\begin{proof}[Proof of Lemma \ref{lemm:mse_var_bc}]
\fi

Similar to \eqref{eq:var_htau},
\begin{equation}
\begin{split}
\Var{\htau_{\text{BC}}^{(l,b)}}&=\frac{1}{k^2l^2}\EE{\Var{
\sum_{i=1}^k\sum_{t=b+1}^l \htau_{(i-1)l+t}+ \sum_{i=2}^k \sum_{t=1}^b \htau_{(i-1)l+t}\cond W_{1:T},S_{1:T}}}+\\
&\qquad\qquad\frac{1}{k^2l^2}\Var{\EE{
\sum_{i=1}^k\sum_{t=b+1}^l \htau_{(i-1)l+t}+ \sum_{i=2}^k \sum_{t=1}^b \htau_{(i-1)l+t}\cond W_{1:T},S_{1:T}}}.
\label{eq:var_htau_bc}
\end{split}
\end{equation}

To bound the first term of \eqref{eq:var_htau_bc}, notice that
\begin{equation*}
\begin{split}
&\frac{1}{k^2l^2}\EE{\Var{
\sum_{i=1}^k\sum_{t=b+1}^l \htau_{(i-1)l+t}+ \sum_{i=2}^k \sum_{t=1}^b \htau_{(i-1)l+t}\cond W_{1:T},S_{1:T}}}\\
&\qquad\qquad =\frac{1}{k^2l^2}\EE{\sum_{i=1}^k\sum_{t=b+1}^l \Var{
\htau_{(i-1)l+t}\cond W_{1:T},S_{1:T}}+ \sum_{i=2}^k \sum_{t=1}^b\Var{ \htau_{(i-1)l+t}\cond W_{1:T},S_{1:T}}}.
\end{split}
\end{equation*}
From \eqref{eq:var_htau_noise},
\begin{equation}
\begin{split}
&\frac{1}{k^2l^2}\EE{\sum_{i=1}^{k}\sum_{t=b+1}^l\Var{
\htau_{(i-1)l+t}\cond W_{1:T},S_{1:T}}}\\
&\qquad\qquad\le\frac{4\sigma^2(l-b)}{kl^2}+\oo\p{\frac{1}{k^2l}}+\oo\p{2^{-k}}.
\label{eq:var_htau_noise_bc}
\end{split}
\end{equation}
Moreover,
\begin{equation*}
\begin{split}
&\frac{1}{k^2l^2}\EE{\sum_{i=2}^k\sum_{t=1}^b\Var{\htau_{(i-1)l+t}|W_{1:T},S_{1:T}}}\\
&\qquad\qquad \le\frac{1}{k^2l^2}\EE{\sum_{i=2}^k\sum_{t=1}^b\sigma^2\p{\frac{kZ_iZ_{i-1}}{k_{11}}-\frac{k(1-Z_i)(1-Z_{i-1})}{k_{00}}}^2}\\
&\qquad\qquad =\frac{\sigma^2b}{l^2}\EE{\frac{I\p{k_{11}>0}}{k_{11}}+\frac{I\p{k_{00}>0}}{k_{00}}}\\
&\qquad\qquad =\frac{\sigma^2b}{l^2}\p{\frac{8}{k-1}+\oo\p{\frac{1}{k^2}}}
=\frac{8\sigma^2b}{kl^2}+\oo\p{\frac{1}{k^2l}}.
\end{split}
\end{equation*}
Combining this with~\eqref{eq:var_htau_noise_bc},
\begin{equation}
\begin{split}
&\frac{1}{k^2l^2}\EE{\Var{
\sum_{i=1}^k\sum_{t=b+1}^l \htau_{(i-1)l+t}+ \sum_{i=2}^k \sum_{t=1}^b \htau_{(i-1)l+t}\cond W_{1:T},S_{1:T}}}\\
&\qquad\qquad \le\frac{8\sigma^2}{kl}+\oo\p{\frac{1}{k^2l}}+\oo\p{2^{-k}}.
\end{split}
\label{eq:noise_bound}
\end{equation}

To bound the second term of~\eqref{eq:var_htau_bc}, with a slight abuse of notation, we use $\mu_t:=\EE{Y_t|W_t,S_t}$ to denote the conditional expectation of the outcome at time $t$. The second term of~\eqref{eq:var_htau_bc} can be written as
\begin{equation}
\begin{split}
&\operatorname{Var}\left[\frac{1}{k_1l}\sum_{i=1}^k\sum_{t=b+1}^l Z_i\mu_{(i-1)l+t}-\frac{1}{k_0l}\sum_{i=1}^k\sum_{t=b+1}^l(1-Z_i)\mu_{(i-1)l+t}\right.\\
&\left.\qquad\qquad+\frac{1}{k_{11}l}\sum_{i=2}^k\sum_{t=1}^b Z_iZ_{i-1}\mu_{(i-1)l+t}-\frac{1}{k_{00}l}\sum_{i=2}^k\sum_{t=1}^b(1-Z_i)(1-Z_{i-1})\mu_{(i-1)l+t}\right].
\end{split}
\end{equation}
We compare this quantity with the variance of the Horvitz-Thompson estimator
\begin{equation*}
\begin{split}
&\operatorname{Var}\left[\frac{2}{kl}\sum_{i=1}^k\sum_{t=b+1}^l(2Z_i-1)\mu_{(i-1)l+t}+\frac{4}{(k-1)l}\sum_{i=2}^k\sum_{t=1}^b\p{Z_iZ_{i-1}-(1-Z_i)(1-Z_{i-1})}\mu_{(i-1)l+t}\right],
\end{split}
\end{equation*}
as we did in proof of Lemma~\ref{lemm:mse_var}.

For the focal-period component, \eqref{eq:focal_var_bound} implies that
\begin{equation}
\begin{split}
&\mathbb{E}\left[\left(\frac{1}{k_1(l-b)}\sum_{i=1}^k\sum_{t=b+1}^lZ_i\mu_{(i-1)l+t}-\frac{1}{k_0(l-b)}\sum_{i=1}^k\sum_{t=b+1}^l(1-Z_i)\mu_{(i-1)l+t}\right.\right.\\
&\left.\left.\qquad\qquad-\frac{2}{k(l-b)}\sum_{i=1}^k\sum_{t=b+1}^l(2Z_i-1)\mu_{(i-1)l+t}\right)^2\right]
\le \frac{4\Lambda^2}{k}.
\end{split}
\label{eq:var_focal_bound}
\end{equation}
For the burn-in component, note that
\begin{equation}
\begin{split}
&\frac{1}{bk_{11}}\sum_{i=2}^k\sum_{t=1}^bZ_iZ_{i-1}\mu_{(i-1)l+t}-\frac{1}{bk_{00}}\sum_{i=2}^k\sum_{t=1}^b(1-Z_i)(1-Z_{i-1})\mu_{(i-1)l+t}\\
&\qquad\qquad-\frac{4}{b(k-1)}\sum_{i=2}^k\sum_{t=1}^b\p{Z_iZ_{i-1}-(1-Z_i)(1-Z_{i-1})}\mu_{(i-1)l+t}\\
&\qquad=\p{1-\frac{4k_{11}}{k-1}}\frac{1}{bk_{11}}\sum_{i=2}^k\sum_{t=1}^bZ_iZ_{i-1}\mu_{(i-1)l+t}\\
&\qquad\qquad-\p{1-\frac{4k_{00}}{k-1}}\frac{1}{bk_{00}}\sum_{i=2}^k\sum_{t=1}^b(1-Z_i)(1-Z_{i-1})\mu_{(i-1)l+t}.
\end{split}
\label{eq:var_burn}
\end{equation}
Since $\EE{k_{11}}=\EE{k_{00}}=\frac{k-1}{4}$, 
$\Var{k_{11}}=\Var{k_{00}}=\frac{5(k-1)-2}{16}$,
\begin{equation}
\begin{split}
\EE{\p{1-\frac{4k_{11}}{k-1}}^2}
&=\Var{1-\frac{4k_{11}}{k-1}}+\p{\EE{1-\frac{4k_{11}}{k-1}}}^2\\
&=\frac{16}{(k-1)^2}\Var{k_{11}}+\p{1-\frac{4}{k-1}\EE{k_{11}}}^2\\
&=\frac{5(k-1)-2}{(k-1)^2},
\end{split}
\end{equation}
Similarly,
\begin{equation}
\begin{split}
\EE{\p{1-\frac{4k_{00}}{k-1}}^2}
=\frac{5(k-1)-2}{(k-1)^2}.
\end{split}
\end{equation}
As a result,
\begin{equation}
\begin{split}
&\mathbb{E}\left[\left(\frac{1}{bk_{11}}\sum_{i=2}^k\sum_{t=1}^bZ_iZ_{i-1}\mu_{(i-1)l+t}-\frac{1}{bk_{00}}\sum_{i=2}^k\sum_{t=1}^b(1-Z_i)(1-Z_{i-1})\mu_{(i-1)l+t}\right.\right.\\
&\left.\left.\qquad\qquad-\frac{4}{b(k-1)}\sum_{i=2}^k\sum_{t=1}^b\p{Z_iZ_{i-1}-(1-Z_i)(1-Z_{i-1})}\mu_{(i-1)l+t}\right)^2\right]\\
&\qquad\le2\Lambda^2\left\{\mathbb{E}\left[\p{1-\frac{4k_{11}}{k-1}}^2\right]+\mathbb{E}\left[\p{1-\frac{4k_{00}}{k-1}}^2\right]\right\}\\
&\qquad=\frac{\p{20(k-1)-8}\Lambda^2}{(k-1)^2}\le\frac{20\Lambda^2}{k-1}.
\end{split}
\label{eq:var_burn_bound}
\end{equation}
Since $(l-b)/l+b/l=1$,
combining \eqref{eq:var_focal_bound} and \eqref{eq:var_burn_bound} gives
\begin{equation}
\begin{split}
&\operatorname{Var}\left[\frac{1}{k_1l}\sum_{i=1}^k\sum_{t=b+1}^l Z_i\mu_{(i-1)l+t}-\frac{1}{k_0l}\sum_{i=1}^k\sum_{t=b+1}^l(1-Z_i)\mu_{(i-1)l+t}\right.\\
&\left.\qquad\qquad+\frac{1}{k_{11}l}\sum_{i=2}^k\sum_{t=1}^b Z_iZ_{i-1}\mu_{(i-1)l+t}-\frac{1}{k_{00}l}\sum_{i=2}^k\sum_{t=1}^b(1-Z_i)(1-Z_{i-1})\mu_{(i-1)l+t}\right]\\
\le&2\operatorname{Var}\left[\frac{2}{kl}\sum_{i=1}^k\sum_{t=b+1}^l(2Z_i-1)\mu_{(i-1)l+t}
+\frac{4}{(k-1)l}\sum_{i=2}^k\sum_{t=1}^b\p{Z_iZ_{i-1}-(1-Z_i)(1-Z_{i-1})}\mu_{(i-1)l+t}\right]\\
&\qquad\qquad+2\frac{l-b}{l}\frac{4\Lambda^2}{k}
+2\frac{b}{l}\frac{20\Lambda^2}{k-1}\\
\le&2\operatorname{Var}\left[\frac{2}{kl}\sum_{i=1}^k\sum_{t=b+1}^l(2Z_i-1)\mu_{(i-1)l+t}
+\frac{4}{(k-1)l}\sum_{i=2}^k\sum_{t=1}^b\p{Z_iZ_{i-1}-(1-Z_i)(1-Z_{i-1})}\mu_{(i-1)l+t}\right]
+\frac{40\Lambda^2}{k-1}.
\end{split}
\end{equation}

It remains to bound the variance of the Horvitz-Thompson estimator. To do so, we consider the covariance between the contributions from two blocks $i$ and $j$. 
Define for shorthand, for $i=2,\ldots,k$,
\begin{equation}
\begin{split}
Q_i&=\frac{2}{kl}\sum_{t=b+1}^l\p{2Z_i-1}\mu_{(i-1)l+t}+\frac{4}{(k-1)l}\sum_{t=1}^b\p{Z_iZ_{i-1}-(1-Z_i)(1-Z_{i-1})}\mu_{(i-1)l+t},
\end{split}
\end{equation}
and set
\begin{equation}
Q_1=\frac{2}{kl}\sum_{t=b+1}^l\p{2Z_1-1}\mu_t.
\end{equation}
Then, $\sum_{i=1}^k Q_i$ is the Horvitz-Thompson estimator defined above, and 
\begin{equation}
\label{eq:bc_signal_comparison_compact}
\Var{\EE{\htau_{\text{BC}}^{(l,b)}\cond W_{1:T},S_{1:T}}}\le 2\Var{\sum_{i=1}^k Q_i}+\frac{40\Lambda^2}{k-1}.
\end{equation}
Consider the following two cases:
\begin{enumerate}
\item When $i=j$ or $j=i+1$, i.e., $i$ and $j$ are the same block or two adjacent blocks. For each block $i$, the focal-period contribution has second moment bounded by
\begin{equation}
\EE{\left(\frac{2}{k(l-b)}\sum_{t=b+1}^l(2Z_i-1)\mu_{(i-1)l+t}\right)^2}
\le \frac{4\Lambda^2}{k^2}.
\end{equation}
Similarly, for the burn-in periods,
\begin{equation}
\EE{\left(\frac{4}{(k-1)b}\sum_{t=1}^b\p{Z_iZ_{i-1}-(1-Z_i)(1-Z_{i-1})}\mu_{(i-1)l+t}\right)^2}
\le \frac{8\Lambda^2}{(k-1)^2},
\end{equation}
as
\begin{equation}
\EE{\abs{Z_jZ_{j-1}-(1-Z_j)(1-Z_{j-1})}}=\PP{Z_j=Z_{j-1}}=\frac12.
\end{equation}
Thus, for each $i\ge 2$, using $(l-b)/l+b/l=1$ and by convexity of the square,
\begin{equation}
\begin{split}
\Var{Q_i}
\le&\EE{\p{\frac{l-b}{l}\frac{2}{k(l-b)}\sum_{t=b+1}^l(2Z_i-1)\mu_{(i-1)l+t}+\frac{b}{l}\frac{4}{(k-1)b}\sum_{t=1}^b\p{Z_iZ_{i-1}-(1-Z_i)(1-Z_{i-1})}\mu_{(i-1)l+t}}^2}\\
\le&\frac{l-b}{l}\EE{\p{\frac{2}{k(l-b)}\sum_{t=b+1}^l(2Z_i-1)\mu_{(i-1)l+t}}^2}\\
&\qquad\qquad+\frac{b}{l}\EE{\p{\frac{4}{(k-1)b}\sum_{t=1}^b\p{Z_iZ_{i-1}-(1-Z_i)(1-Z_{i-1})}\mu_{(i-1)l+t}}^2}\\
\le&\frac{4\Lambda^2(l-b)}{k^2l}+\frac{8\Lambda^2b}{(k-1)^2l}.
\end{split}
\end{equation}
For $i=1$, we only have the focal-period component. Therefore, 
\begin{equation}
\begin{split}
&\sum_{i=1}^k\Var{Q_i}\le\frac{4\Lambda^2(l-b)}{kl}+\frac{8\Lambda^2b}{(k-1)l}\le\frac{8\Lambda^2}{k-1}.
\end{split}
\end{equation}
Furthermore, since $2\abs{\Cov{X,Y}}\le \Var{X}+\Var{Y}$, the total contribution from the diagonal and adjacent-block terms is bounded by
\begin{equation}
\begin{split}
&3\sum_{i=1}^k
\Var{Q_i}\le\frac{24\Lambda^2}{k-1}.
\end{split}
\end{equation}
\item When $j-i\ge 2$, i.e., $i$ and $j$ are non-adjacent blocks. 
By the first case,
\begin{equation}
\begin{split}
\Var{\sum_{i=1}^k Q_i}
&=\sum_{i=1}^k\Var{Q_i}+2\sum_{i=1}^{k-1}\Cov{Q_i,Q_{i+1}}+2\sum_{i=1}^{k-2}\sum_{j=i+2}^k\Cov{Q_i,Q_j}\\
&\le\frac{24\Lambda^2}{k-1}+2\sum_{i=1}^{k-2}\sum_{j=i+2}^k\abs{\Cov{Q_i,Q_j}}.
\end{split}
\end{equation}
It remains to bound the last double sum.

For a fixed pair with $j-i\ge 2$, the assignments $Z_{i+1},\ldots,Z_j$ are independent of $S_{il+1}$. Conditional on these assignments, two state distributions starting from $S_{il+1}=x$ and $S_{il+1}=x'$ undergo $(j-i-1)l+t-1$ common transitions before time $(j-1)l+t$. 
B Assumption~\ref{assu:mixing}, for any $x,x'\in\mathcal S$,
\begin{equation}\label{eq:bc_block_oscillation_compact}
\begin{split}
&\abs{\EE{Q_j\cond S_{il+1}=x}-\EE{Q_j\cond S_{il+1}=x'}}\\
&\qquad\le\EE{\abs{\EE{Q_j\cond S_{il+1}=x,Z_{i+1},\ldots,Z_j}-\EE{Q_j\cond S_{il+1}=x',Z_{i+1},\ldots,Z_j}}}\\
&\qquad\le\frac{4\Lambda}{kl}\sum_{t=b+1}^l\exp\p{-\frac{(j-i-1)l+t-1}{t_{\rm mix}}}\\
&\qquad\qquad+\frac{4\Lambda}{(k-1)l}\sum_{t=1}^b\exp\p{-\frac{(j-i-1)l+t-1}{t_{\rm mix}}}\\
&\qquad\le\frac{4\Lambda}{(k-1)l}\sum_{t=1}^l\exp\p{-\frac{(j-i-1)l+t-1}{t_{\rm mix}}}\\
&\qquad=\frac{4\Lambda}{k-1}\frac{\exp\p{-\frac{(j-i-1)l}{t_{\rm mix}}}\p{1-\exp\p{-l/t_{\rm mix}}}}{l\p{1-\exp\p{-1/t_{\rm mix}}}}.
\end{split}
\end{equation}

Since $Q_i$ depends only on assignments and states through block $i$, the Markov property and independent block randomization give $\EE{Q_j\cond Q_i,S_{il+1}}=\EE{Q_j\cond S_{il+1}}$.
Thus, by Cauchy-Schwarz and Popoviciu's inequality,
\begin{equation}
\begin{split}
2\abs{\Cov{Q_i,Q_j}}
&=2\abs{\Cov{Q_i,\EE{Q_j\cond S_{il+1}}}}\\
&\le 2\sqrt{\Var{Q_i}}\sqrt{\Var{\EE{Q_j\cond S_{il+1}}}}\\
&\le\sqrt{\Var{Q_i}}\sup_{x,x'\in\mathcal S}\abs{\EE{Q_j\cond S_{il+1}=x}-\EE{Q_j\cond S_{il+1}=x'}}\\
&\le\frac{4\Lambda}{k-1}\frac{\exp\p{-\frac{(j-i-1)l}{t_{\rm mix}}}\p{1-\exp\p{-l/t_{\rm mix}}}}{l\p{1-\exp\p{-1/t_{\rm mix}}}}\sqrt{\Var{Q_i}}.
\end{split}
\end{equation}

Summing over non-adjacent pairs and using $\sum_{i=1}^k\Var{Q_i}\le 8\Lambda^2/(k-1)$ from the first case,
\begin{equation}
\begin{split}
2\sum_{i=1}^{k-2}\sum_{j=i+2}^k\abs{\Cov{Q_i,Q_j}}
\le&\frac{4\Lambda\p{1-\exp\p{-l/t_{\rm mix}}}}{(k-1)l\p{1-\exp\p{-1/t_{\rm mix}}}}\sum_{i=1}^{k-2}\sqrt{\Var{Q_i}}\sum_{j=i+2}^k\exp\p{-\frac{(j-i-1)l}{t_{\rm mix}}}\\
\le&\frac{4\Lambda\exp\p{-l/t_{\rm mix}}}{(k-1)l\p{1-\exp\p{-1/t_{\rm mix}}}}\sum_{i=1}^{k-2}\sqrt{\Var{Q_i}}\\
\le&\frac{4\Lambda\exp\p{-l/t_{\rm mix}}}{(k-1)l\p{1-\exp\p{-1/t_{\rm mix}}}}\sqrt{k\sum_{i=1}^k\Var{Q_i}}\\
\le&\frac{8\sqrt{2}\Lambda^2}{k-1}\sqrt{\frac{k}{k-1}}\frac{\exp\p{-l/t_{\rm mix}}}{l\p{1-\exp\p{-1/t_{\rm mix}}}},
\end{split}
\end{equation}

\end{enumerate}
Combining the two cases,
\begin{equation}
\Var{\sum_{i=1}^k Q_i}\le\frac{24\Lambda^2}{k-1}+\frac{8\sqrt{2}\Lambda^2}{k-1}\sqrt{\frac{k}{k-1}}\frac{\exp\p{-l/t_{\rm mix}}}{l\p{1-\exp\p{-1/t_{\rm mix}}}}.
\end{equation}
Together with \eqref{eq:bc_signal_comparison_compact} and the noise bound \eqref{eq:noise_bound}, this yields
\begin{equation}
\begin{split}
\Var{\htau_{\text{BC}}^{(l,b)}}
&\le 2\Var{\sum_{i=1}^k Q_i}+\frac{40\Lambda^2}{k-1}+\frac{8\sigma^2}{kl}+\oo\p{\frac{1}{k^2l}}+\oo\p{2^{-k}}\\
&\le\frac{88\Lambda^2}{k-1}+\frac{16\sqrt{2}\Lambda^2}{k-1}\sqrt{\frac{k}{k-1}}\frac{\exp\p{-l/t_{\rm mix}}}{l\p{1-\exp\p{-1/t_{\rm mix}}}}\\
&\qquad\qquad+\frac{8\sigma^2}{kl}+\oo\p{\frac{1}{k^2l}}+\oo\p{2^{-k}}.
\end{split}
\end{equation}
Finally, since $t_{\rm mix}$ is fixed, $l\ge1$, and
\begin{equation}
\frac{1}{k-1}=\frac1k+\oo\p{\frac1{k^2}},\qquad \frac{1}{k-1}\sqrt{\frac{k}{k-1}}=\frac1k+\oo\p{\frac1{k^2}},
\end{equation}
we obtain
\begin{equation}\label{eq:bc_variance_repaired_compact}
\Var{\htau_{\text{BC}}^{(l,b)}}\le\frac{88\Lambda^2}{k}+\frac{16\sqrt{2}\Lambda^2\exp\p{-l/t_{\rm mix}}}{kl\p{1-\exp\p{-1/t_{\rm mix}}}}+\frac{8\sigma^2}{kl}+\oo\p{\frac1{k^2}}.
\end{equation}

\ifms
\endproof
\else
\end{proof}
\fi

\ifms
\proof{Proof of Lemma \ref{lemm:clt_decom_bc}}
\else
\begin{proof}[Proof of Lemma \ref{lemm:clt_decom_bc}]
\fi

We start by decomposing our estimator as follows:
\begin{equation}
\begin{split}
\htau_{\text{BC}}^{(l,b)}=&
\widetilde{\htau_{\text{BC}}^{(l,b)}}+\oop\p{\frac{1}{k}}\\
=&\frac{1}{k}\sum_{i=1}^{k}\left[\frac{l-b}{l}\cdot\frac{\bM_i(1)Z_i}{k_1/k}-\frac{l-b}{l}\cdot\frac{\bM_i(0)(1-Z_i)}{1-k_1/k} + \right.\\
&\qquad\qquad \left.  \frac{b}{l}\cdot\frac{\bM_i^b(1)Z_iZ_{i-1}}{k_{11}/k}-\frac{b}{l}\cdot\frac{\bM_i^b(0)(1-Z_i)(1-Z_{i-1})}{k_{00}/k}  \right]+ \Tilde{\Delta}^{(l,b)}+\oop\p{\frac{1}{k}},
\end{split}
\end{equation}
Now, by Assumption \ref{assu:mixing} on mixing time,
\begin{equation}
\begin{split}
\frac{b}{l}\abs{\bM^b_i(w) - \bmu^b_i(w)}
&= \frac{b}{l}\abs{\EE{\bY^b_i\cond Z_i=Z_{i-1}=w}-\EE{\bY^b_i\cond Z_i=w, Z_{i-1}=w,\cdots}}\\
&\le \frac{1}{l}\sum_{t=(i-1)l+1}^{(i-1)l+b}\abs{\EE{Y_t\cond  Z_i=Z_{i-1}=w}-\EE{Y_t\cond Z_i=w, Z_{i-1}=w,\cdots}}\\
&= \oo\p{\frac{1}{l} \exp\p{-l/t_{\text{mix}} }},
\label{eq:mean_diff_b}
\end{split}
\end{equation}
uniformly across $i = 1, \, \ldots, \, k$ and $w = 0, \, 1$.
Recalling \eqref{eq:GATE_recall} and \eqref{eq:mean_diff}, we see that
\begin{equation}
\begin{split}
\htau_{\text{BC}}^{(l,b)}
=&\Tilde{\Delta}^{(l,b)}+\oop\p{\frac{1}{k}}+\frac{l-b}{kl}\sum_{i=1}^{k}\sqb{\frac{\bmu_i(1)Z_i}{k_1/k}-\frac{\bmu_i(0)(1-Z_i)}{1-k_1/k}} + \oop\p{\frac{e^{-b/t_{\text{mix}}}}{l}} \\
&\qquad\qquad +\frac{b}{kl}\sum_{i=1}^{k}\sqb{\frac{\bmu^b_i(1)Z_iZ_{i-1}}{k_{11}/k}-\frac{\bmu^b_i(0)(1-Z_i)(1-Z_{i-1})}{k_{00}/k}} + \oop\p{\frac{e^{-l/t_{\text{mix}}}}{l}}\\
=& \tau_{\GATE} + \Tilde{\Gamma}^{(l,b)} + \Tilde{\Delta}^{(l,b)} +\oop\p{\frac{1}{k}} + \oop\p{\frac{e^{-b/t_{\text{mix}}}}{l}}.
\end{split}
\end{equation}

\ifms
\endproof
\else
\end{proof}
\fi

\ifms
\proof{Proof of Lemma \ref{lemm:GammaCLT_bc}}
\else
\begin{proof}[Proof of Lemma \ref{lemm:GammaCLT_bc}]
\fi

Define
\begin{equation}
\begin{split}
    \Tilde{X_{i}} &= \left(Z_{i}, Z_{i}Z_{i-1}, (1-Z_{i})(1-Z_{i-1}), Z_{i}\bmu_i(1), (1-Z_{i})\bmu_i(0),\right.\\
    &\qquad\qquad \left.Z_{i}Z_{i-1}\bmu^b_i(1), (1-Z_{i})(1-Z_{i-1})\bmu^b_i(0)\right)^T
\end{split}
\end{equation}
and $\Tilde{\Sigma}_{X,k} = \Cov{\frac{1}{\sqrt{k}} \sum \Tilde X_i}$.
Note that 
\begin{equation}
    \EE{\sum_{i=1}^k \Tilde X_{i} } = \p{\frac{k}{2},\frac{k}{4},\frac{k}{4},\frac{1}{2}\sum_{i = 1}^k \bmu_i(1), \frac{1}{2}\sum_{i = 1}^k \bmu_i(0), \frac{1}{4}\sum_{i = 1}^k \bmu^b_i(1), \frac{1}{4}\sum_{i = 1}^k \bmu^b_i(0)}^T,
\end{equation}
and $\Cov{\sum_{i=1}^k  \Tilde X_{i} }$ can be calculated by noticing that
\begin{equation*}
\begin{split}
\Var{\sum_{i=1}^kZ_{i}} &= \frac{k}{4}\\
\Var{\sum_{i=1}^kZ_{i}Z_{i-1}} &= \frac{5k}{16}+\oo\p{1}\\
\Var{\sum_{i=1}^k(1-Z_{i})(1-Z_{i-1})} &= \frac{5k}{16}+\oo\p{1}\\
\Var{\sum_{i=1}^kZ_{i}\bmu_i(1)} &= \frac{1}{4}\sum_{i=1}^k\bmu_i(1)^2\\
\Var{\sum_{i=1}^k(1-Z_{i})\bmu_i(0)} &= \frac{1}{4}\sum_{i=1}^k\bmu_i(0)^2\\
\Var{\sum_{i=1}^kZ_{i}Z_{i-1}\bmu^b_i(1)} &= \frac{3}{16}\sum_{i=1}^k\bmu^b_i(1)^2+\frac{1}{8}\sum_{i=1}^{k-1}\bmu^b_i(1)\bmu^b_{i+1}(1)\\
\Var{\sum_{i=1}^k(1-Z_{i})(1-Z_{i-1})\bmu^b_i(0)} &= \frac{3}{16}\sum_{i=1}^k\bmu^b_i(0)^2+\frac{1}{8}\sum_{i=1}^{k-1}\bmu^b_i(0)\bmu^b_{i+1}(0),
\end{split}
\end{equation*}
and that
\begin{align*}
\Cov{\sum_{i=1}^kZ_{i},\sum_{i=1}^kZ_{i}Z_{i-1}} &= \frac{k}{4}+\oo\p{1}\\
\Cov{\sum_{i=1}^kZ_{i},\sum_{i=1}^k(1-Z_{i})(1-Z_{i-1})} &= -\frac{k}{4}+\oo\p{1}\\
\Cov{\sum_{i=1}^kZ_{i},\sum_{i=1}^kZ_{i}\bmu_i(1)} &= \frac{1}{4}\sum_{i=1}^k\bmu_i(1)\\
\Cov{\sum_{i=1}^kZ_{i},\sum_{i=1}^k(1-Z_{i})\bmu_i(0)} &= -\frac{1}{4}\sum_{i=1}^k\bmu_i(0)\\
\Cov{\sum_{i=1}^kZ_{i},\sum_{i=1}^kZ_{i}Z_{i-1}\bmu^b_i(1)} &= \frac{1}{4}\sum_{i=1}^k\bmu^b_i(1)+\oo\p{1}\\
\Cov{\sum_{i=1}^kZ_{i},\sum_{i=1}^k(1-Z_{i})(1-Z_{i-1})\bmu^b_i(0)} &= -\frac{1}{4}\sum_{i=1}^k\bmu^b_i(0)+\oo\p{1}\\
\Cov{\sum_{i=1}^kZ_{i}Z_{i-1},\sum_{i=1}^k(1-Z_{i})(1-Z_{i-1})} &= -\frac{3k}{16}+\oo\p{1}\\
\Cov{\sum_{i=1}^kZ_{i}Z_{i-1},\sum_{i=1}^kZ_{i}\bmu_i(1)} &= \frac{1}{4}\sum_{i=1}^k\bmu_i(1)+\oo\p{1}\\
\Cov{\sum_{i=1}^kZ_{i}Z_{i-1},\sum_{i=1}^k(1-Z_{i})\bmu_i(0)} &= -\frac{1}{4}\sum_{i=1}^k\bmu_i(0)+\oo\p{1}\\
\Cov{\sum_{i=1}^kZ_{i}Z_{i-1},\sum_{i=1}^kZ_{i}Z_{i-1}\bmu^b_i(1)} &= \frac{5}{16}\sum_{i=1}^k\bmu^b_i(1)+\oo\p{1}\\
\Cov{\sum_{i=1}^kZ_{i}Z_{i-1},\sum_{i=1}^k(1-Z_{i})(1-Z_{i-1})\bmu^b_i(0)} &= -\frac{3}{16}\sum_{i=1}^k\bmu^b_i(0)+\oo\p{1}\\
\Cov{\sum_{i=1}^k(1-Z_{i})(1-Z_{i-1}),\sum_{i=1}^kZ_{i}\bmu_i(1)} &= -\frac{1}{4}\sum_{i=1}^k\bmu_i(1)+\oo\p{1}\\
\Cov{\sum_{i=1}^k(1-Z_{i})(1-Z_{i-1}),\sum_{i=1}^k(1-Z_{i})\bmu_i(0)} &= \frac{1}{4}\sum_{i=1}^k\bmu_i(0)+\oo\p{1}\\
\Cov{\sum_{i=1}^k(1-Z_{i})(1-Z_{i-1}),\sum_{i=1}^kZ_{i}Z_{i-1}\bmu^b_i(1)} &= -\frac{3}{16}\sum_{i=1}^k\bmu^b_i(1)+\oo\p{1}\\
\Cov{\sum_{i=1}^k(1-Z_{i})(1-Z_{i-1}),\sum_{i=1}^k(1-Z_{i})(1-Z_{i-1})\bmu^b_i(0)} &= \frac{5}{16}\sum_{i=1}^k\bmu^b_i(0)+\oo\p{1}\\
\Cov{\sum_{i=1}^kZ_{i}\bmu_i(1),\sum_{i=1}^k(1-Z_{i})\bmu_i(0)} &= -\frac{1}{4}\sum_{i=1}^k\bmu_i(1)\bmu_i(0)\\
\Cov{\sum_{i=1}^kZ_{i}\bmu_i(1),\sum_{i=1}^kZ_{i}Z_{i-1}\bmu^b_i(1)} &= \frac{1}{8}\sum_{i=1}^k\bmu_i(1)\bmu^b_i(1)+\frac{1}{8}\sum_{i=1}^{k-1}\bmu_i(1)\bmu^b_{i+1}(1)\\
\Cov{\sum_{i=1}^kZ_{i}\bmu_i(1),\sum_{i=1}^k(1-Z_{i})(1-Z_{i-1})\bmu^b_i(0)} &= -\frac{1}{8}\sum_{i=1}^k\bmu_i(1)\bmu^b_i(0)-\frac{1}{8}\sum_{i=1}^{k-1}\bmu_i(1)\bmu^b_{i+1}(0)\\
\Cov{\sum_{i=1}^k(1-Z_{i})\bmu_i(0),\sum_{i=1}^kZ_{i}Z_{i-1}\bmu^b_i(1)} &= -\frac{1}{8}\sum_{i=1}^k\bmu_i(0)\bmu^b_i(1)-\frac{1}{8}\sum_{i=1}^{k-1}\bmu_i(0)\bmu^b_{i+1}(1)\\
\Cov{\sum_{i=1}^k(1-Z_{i})\bmu_i(0),\sum_{i=1}^k(1-Z_{i})(1-Z_{i-1})\bmu^b_i(0)} &= \frac{1}{8}\sum_{i=1}^k\bmu_i(0)\bmu^b_i(0)+\frac{1}{8}\sum_{i=1}^{k-1}\bmu_i(0)\bmu^b_{i+1}(0)\\
\Cov{\sum_{i=1}^kZ_{i}Z_{i-1}\bmu^b_i(1),\sum_{i=1}^k(1-Z_{i})(1-Z_{i-1})\bmu^b_i(0)} &= -\frac{1}{16}\sum_{i=1}^k\bmu^b_i(1)\bmu^b_i(0)-\\
&\qquad\qquad\frac{1}{16}\sum_{i=2}^k\bmu^b_i(1)\bmu^b_{i-1}(0)-\frac{1}{16}\sum_{i=1}^{k-1}\bmu^b_i(1)\bmu^b_{i+1}(0).
\end{align*}

Notice that
\begin{equation}
\begin{split}
\Tilde{\Gamma}^{(l,b)} =& \frac{\frac{l-b}{l}\frac{1}{k}\sum_{i}Z_{i}\bmu_i(1)}{\frac{1}{k}\sum_{i}Z_{i}}-\frac{l-b}{l}\frac{\sum_j\bmu_j(1)}{k} \\
&\qquad\qquad- \frac{\frac{l-b}{l}\frac{1}{k}\sum_{i}(1-Z_{i})\bmu_i(0)}{\frac{1}{k}\sum_{i}(1-Z_{i})}+\frac{l-b}{l}\frac{\sum_j\bmu_j(0)}{k}  \\
&\qquad\qquad +\frac{\frac{b}{l}\frac{1}{k}\sum_{i}Z_{i}Z_{i-1}\bmu^b_i(1)}{\frac{1}{k}\sum_{i}Z_{i}Z_{i-1}}-\frac{b}{l}\frac{\sum_j\bmu^b_j(1)}{k} \\
&\qquad\qquad - \frac{\frac{b}{l}\frac{1}{k}\sum_{i}(1-Z_{i})(1-Z_{i-1})\bmu^b_i(0)}{\frac{1}{k}\sum_{i}(1-Z_{i})(1-Z_{i-1})}+\frac{b}{l}\frac{\sum_j\bmu^b_j(0)}{k}.
\end{split}
\label{eq:lemma19_decomp}
\end{equation}
Since each $\tilde{X}_i$ depends only on $(Z_{i-1},Z_i)$ and that $\{\tilde{X}_i\}$ is uniformly bounded,
\begin{equation}
\frac{1}{k}\sum_{i=1}^k\p{\tilde{X}_i-\EE{\tilde{X}_i}}=\oop(k^{-1/2}).
\end{equation}
A first-order Taylor expansion of~\eqref{eq:lemma19_decomp} around the expectations then gives
\begin{equation}
\sqrt{k}\tilde{\Gamma}^{(l,b)}
= \p{\nabla\tilde{\Gamma}^{(l,b)}}^\top
\frac{1}{\sqrt{k}}\sum_{i=1}^k\left(\widetilde{X}_i-\EE{\tilde{X}_i}\right)
+o_P(1).
\label{eq:lemma19_expansion}
\end{equation}
Notice that the variance of the first term is
\begin{equation}
\label{eq:V_tau}
\Tilde{\Sigma}_{\Gamma,k}=\p{\nabla \Tilde{\Gamma}^{(l,b)}}^\top \Tilde{\Sigma}_{X,k} \p{\nabla \Tilde{\Gamma}^{(l,b)}},
\end{equation}
where
\begin{equation}
\begin{split}
\nabla \Tilde{\Gamma}^{(l,b)}&= \left(-2\frac{l-b}{kl}\sum_i\p{\bmu_i(1)+\bmu_i(0)},
-4\frac{b}{kl}\sum_i\bmu^b_i(1),\right.\\
&\qquad\qquad\left. 4\frac{b}{kl}\sum_i\bmu^b_i(0),
2\frac{l-b}{l},
-2\frac{l-b}{l},
4\frac{b}{l},
-4\frac{b}{l}
\right)  ^\top,
\end{split}
\label{eq:V_tau_derivative}
\end{equation}
and thus
\begin{align*}
\Tilde{\Sigma}_{\Gamma,k}&=\frac{(l-b)^2}{k^2l^2}\cb{\sum_i\p{\bmu_i(1)+\bmu_i(0)}}^2+
\frac{5b^2}{k^2l^2}\cb{\sum_i\bmu^b_i(1)}^2+
\frac{5b^2}{k^2l^2}\cb{\sum_i\bmu^b_i(0)}^2\\
&\qquad\qquad+\frac{(l-b)^2}{kl^2} \sum_i \bmu_i(1)^2+
\frac{(l-b)^2}{kl^2} \sum_i \bmu_i(0)^2\\
&\qquad\qquad +\frac{3b^2}{kl^2} \sum_i \bmu^b_i(1)^2 
+ \frac{2b^2}{kl^2} \sum_i \bmu^b_i(1)\bmu^b_{i+1}(1)
+ \frac{3b^2}{kl^2} \sum_i \bmu^b_i(0)^2 
+ \frac{2b^2}{kl^2} \sum_i \bmu^b_i(0)\bmu^b_{i+1}(0)\\
&\qquad\qquad +\frac{4(l-b)b}{k^2l^2}\cb{\sum_i\p{\bmu_i(1)+\bmu_i(0)}}\cb{\sum_i\bmu^b_i(1)}
+\frac{4(l-b)b}{k^2l^2}\cb{\sum_i\p{\bmu_i(1)+\bmu_i(0)}}\cb{\sum_i\bmu^b_i(0)}\\
&\qquad\qquad -\frac{2(l-b)^2}{k^2l^2}\cb{\sum_i\p{\bmu_i(1)+\bmu_i(0)}}\cb{\sum_i\bmu_i(1)}
-\frac{2(l-b)^2}{k^2l^2}\cb{\sum_i\p{\bmu_i(1)+\bmu_i(0)}}\cb{\sum_i\bmu_i(0)}\\
&\qquad\qquad -\frac{4(l-b)b}{k^2l^2}\cb{\sum_i\p{\bmu_i(1)+\bmu_i(0)}}\cb{\sum_i\bmu^b_i(1)}
-\frac{4(l-b)b}{k^2l^2}\cb{\sum_i\p{\bmu_i(1)+\bmu_i(0)}}\cb{\sum_i\bmu^b_i(0)}\\
&\qquad\qquad +\frac{6b^2}{k^2l^2}\p{\sum_i\bmu^b_i(1)}\p{\sum_i\bmu^b_i(0)}\\
&\qquad\qquad -\frac{4(l-b)b}{k^2l^2}\p{\sum_i\bmu_i(1)}\p{\sum_i\bmu^b_i(1)}
-\frac{4(l-b)b}{k^2l^2}\p{\sum_i\bmu_i(0)}\p{\sum_i\bmu^b_i(1)}\\
&\qquad\qquad -\frac{10b^2}{k^2l^2}\p{\sum_i\bmu^b_i(1)}^2
-\frac{6b^2}{k^2l^2}\p{\sum_i\bmu^b_i(1)}\p{\sum_i\bmu^b_i(0)}\\
&\qquad\qquad -\frac{4(l-b)b}{k^2l^2}\p{\sum_i\bmu_i(1)}\p{\sum_i\bmu^b_i(0)}
-\frac{4(l-b)b}{k^2l^2}\p{\sum_i\bmu_i(0)}\p{\sum_i\bmu^b_i(0)}\\
&\qquad\qquad -\frac{6b^2}{k^2l^2}\p{\sum_i\bmu^b_i(1)}\p{\sum_i\bmu^b_i(0)}
-\frac{10b^2}{k^2l^2}\p{\sum_i\bmu^b_i(0)}^2\\
&\qquad\qquad +\frac{2(l-b)^2}{kl^2}\sum_{i}\bmu_i(1)\bmu_i(0)\\
&\qquad\qquad +\frac{2(l-b)b}{kl^2}\sum_{i}\bmu_i(1)\bmu^b_i(1) + \frac{2(l-b)b}{kl^2}\sum_{i}\bmu_i(1)\bmu^b_{i+1}(1)\\
&\qquad\qquad +\frac{2(l-b)b}{kl^2}\sum_{i}\bmu_i(1)\bmu^b_i(0) + \frac{2(l-b)b}{kl^2}\sum_{i}\bmu_i(1)\bmu^b_{i+1}(0)\\
&\qquad\qquad +\frac{2(l-b)b}{kl^2}\sum_{i}\bmu_i(0)\bmu^b_i(1) + \frac{2(l-b)b}{kl^2}\sum_{i}\bmu_i(0)\bmu^b_{i+1}(1)\\
&\qquad\qquad +\frac{2(l-b)b}{kl^2}\sum_{i}\bmu_i(0)\bmu^b_i(0) + \frac{2(l-b)b}{kl^2}\sum_{i}\bmu_i(0)\bmu^b_{i+1}(0)\\
&\qquad\qquad +\frac{2b^2}{kl^2}\sum_{i}\bmu^b_i(1)\bmu^b_i(0)
+\frac{2b^2}{kl^2}\sum_{i}\bmu^b_i(1)\bmu^b_{i-1}(0)
+\frac{2b^2}{kl^2}\sum_{i}\bmu^b_i(1)\bmu^b_{i+1}(0) + \oo(k^{-1})\\
&\to  (1-\beta)^2\p{V^f_0+V^f_1+2V^f_{01}}\\
&\qquad\qquad +\beta^2 \p{3V^b_0+3V^b_1+2V^b_{01}+2V^{bs}_{01}+2V^{bs}_{10}+2V^{bs}_{00}+2V^{bs}_{11}}\\
&\qquad\qquad +2\beta(1-\beta)
\p{{V}^{bf}_{0}+{V}^{bf}_{1}+{V}^{bf}_{01}+{V}^{bf}_{10}+{V}^{fs}_{00}+{V}^{fs}_{11}+{V}^{fs}_{01}+{V}^{fs}_{10} }.
\end{align*}

We now apply the central limit theorem to the first term in~\eqref{eq:lemma19_expansion} if $\tilde{\Sigma}_\Gamma>0$.
Since each $\tilde{X}_i$ depends only on $(Z_{i-1},Z_i)$, and
\[
\left(\nabla\tilde{\Gamma}^{(l,b)}\right)^\top
\left(\tilde{X}_i-\EE{\tilde{X}_i}\right)
\]
are mean zero and uniformly bounded,
by the central limit theorem for dependency neighborhoods \citep{ross2011fundamentals},
\begin{equation}
\frac{1}{\sqrt{k}}\sum_{i=1}^k
\p{\nabla\tilde{\Gamma}^{(l,b)}}^\top \p{\tilde{X}_i-\EE{\tilde{X}_i}}
\to_d\nn(0,\tilde{\Sigma}_\Gamma).
\end{equation}
Together with~\eqref{eq:lemma19_expansion}, this gives
\begin{equation}
\sqrt{k}\tilde{\Gamma}^{(l,b)}\to_d\mathcal{N}(0,\tilde{\Sigma}_\Gamma).
\end{equation}
If $\tilde{\Sigma}_\Gamma=0$, the variance of the leading term in~\eqref{eq:lemma19_expansion} converges to zero, so again the same conclusion follows by applying Chebyshev's inequality.

\ifms
\endproof
\else
\end{proof}
\fi

\ifms
\proof{Proof of Lemma \ref{lemm:Delta_bc}}
\else
\begin{proof}[Proof of Lemma \ref{lemm:Delta_bc}]
\fi

Define $Y_t(w)=Y_t\cond do(W_t=w)$, $M_t(w)=\EE{Y_t(w)}$, $Y^b_t(w)=Y_t\cond do(Z_i=Z_{i-1}=w)$, and $M^b_t(w)=\EE{Y^b_t(w)}$.
We start by writing 
\begin{equation}
\Tilde\Delta^{(l,b)}=\frac{1}{k}\sum_{i=1}^k\Tilde D_i,
\end{equation}
where
\begin{equation}
\begin{split}
\Tilde D_i
&:=\frac{l-b}{l}\left[
\frac{\{\bY_i(1)-\bM_i(1)\}Z_i}{k_1/k}
-\frac{\{\bY_i(0)-\bM_i(0)\}(1-Z_i)}{k_0/k}
\right]\\
&\quad+\frac{b}{l}\left[
\frac{\{\bY_i^b(1)-\bM_i^b(1)\}Z_iZ_{i-1}}{k_{11}/k}
-\frac{\{\bY_i^b(0)-\bM_i^b(0)\}(1-Z_i)(1-Z_{i-1})}{k_{00}/k}
\right].
\end{split}
\end{equation}
By the same mixing argument as in~\eqref{eq:lemma15_centering},
\begin{equation}
\abs{\EE{\bY_i\cond Z}-\bM_i(Z_i)}=O\left(\frac{e^{-b/t_{\mix}}}{l-b}\right),
\qquad
\abs{\EE{\bY_i^b\cond Z}-\bM_i^b(Z_i)}=O\left(\frac{e^{-l/t_{\mix}}}{b}\right)
\end{equation}
whenever the corresponding terms in $\Tilde D_i$ are nonzero. Therefore,
\begin{equation}
\abs{\EE{\Tilde D_i\cond Z}}
=O\left(\frac{l-b}{l}\frac{e^{-b/t_{\mix}}}{l-b}+\frac{b}{l}\frac{e^{-l/t_{\mix}}}{b}\right)
=\oo\p{\frac{e^{-b/t_{\mix}}}{l}}.
\end{equation}
Furthermore, by the same argument as in~\eqref{eq:lemma15_noise1}-\eqref{eq:lemma15_noise2} and Assumption~\ref{assu:mixing}, for $t<s$,
\begin{equation}
\abs{\Cov{Y_t,Y_s\cond Z}}
=\abs{\Cov{\mu_t(S_t,W_t),\mu_s(S_s,W_s)\cond Z}}
\le C\exp\p{-\frac{s-t}{t_{\mix}}}
\end{equation}
for some constant $C<\infty$. 
Conditional on $Z$, the $\bM$-terms are constants and thus don't affect the variance. Expanding $\bY_i$ and $\bY_i^b$ gives
\begin{equation*}
\begin{split}
\Var{\Tilde D_i\cond Z}
=\Var{\frac{1}{l}\p{\frac{Z_i}{k_1/k}-\frac{1-Z_i}{k_0/k}}\sum_{t=(i-1)l+b+1}^{il}Y_t
+\frac{1}{l}\p{\frac{Z_iZ_{i-1}}{k_{11}/k}-\frac{(1-Z_i)(1-Z_{i-1})}{k_{00}/k}}\sum_{t=(i-1)l+1}^{(i-1)l+b}Y_t\cond Z}.
\end{split}
\end{equation*}
Since the $k_1/k, k_0/k, k_{11}/k, k_{00}/k$ are bounded away from zero with probability approaching one, 
\begin{equation}
\begin{split}
\Var{\Tilde D_i\cond Z}
&\le \frac{C}{l^2}\sum_{t=(i-1)l+1}^{il}\sum_{s=(i-1)l+1}^{il}\abs{\Cov{Y_t,Y_s\cond Z}}\\
&\le \frac{C}{l^2}\sum_{t=(i-1)l+1}^{il}\sum_{s=(i-1)l+1}^{il}\exp\p{-\frac{\abs{t-s}}{t_{\mix}}}
=\oo\p{\frac{1}{l}}.
\end{split}
\end{equation}
Similarly, for $i<j$,
\begin{equation}
\abs{\Cov{\Tilde D_i,\Tilde D_j\cond Z}}
\le\frac{C}{l^2}\sum_{t=(i-1)l+1}^{il}\sum_{s=(j-1)l+1}^{jl}\exp\left(-\frac{s-t}{t_{\mix}}\right)
=\oo\p{\frac{e^{-(j-i-1)l/t_{\mix}}}{l^2}}.
\end{equation}

Putting everything together,
\begin{equation}
\Var{\Tilde\Delta^{(l,b)}\cond Z}
=\frac{1}{k^2}\sum_{i=1}^k\Var{\Tilde D_i\cond Z}
+\frac{2}{k^2}\sum_{i<j}\Cov{\Tilde D_i,\Tilde D_j\cond Z}
=\oo\p{\frac{1}{kl}},
\end{equation}
while
\begin{equation}
\abs{\EE{\Tilde\Delta^{(l,b)}\cond Z}}
\le\frac{1}{k}\sum_{i=1}^k\abs{\EE{\Tilde D_i\cond Z}}
=\oo\p{\frac{e^{-b/t_{\mix}}}{l}}.
\end{equation}
It then follows from Chebyshev's inequality that
\begin{equation}
\begin{split}
\sqrt{k}\Tilde\Delta^{(l,b)}
&=\sqrt{k}\p{\Tilde\Delta^{(l,b)}-\EE{\Tilde\Delta^{(l,b)}\cond Z}}+\sqrt{k}\EE{\Tilde\Delta^{(l,b)}\cond Z}\\
&=\oop\p{\frac{1}{\sqrt{l}}+\frac{\sqrt{k}e^{-b/t_{\mix}}}{l}}\to_p0.
\end{split}
\end{equation}

\ifms
\endproof
\else
\end{proof}
\fi

\ifms
\end{APPENDIX}

\fi

\end{document}

\section{A Lower Bound}

\begin{figure}[t]
\centering
\begin{tikzpicture}
\node[black, draw, circle, inner sep=2.4mm] (s1) at (0,0) {$h_1$};
\node[black, draw, circle, inner sep=2.4mm] (s2) at (3,0) {$h_2$};
\node (sd) at (6,0) {$\cdots$};
\node[black, draw, circle, inner sep=1mm] (sn1) at (9,0) {$h_{Q-1}$};
\node[black, draw, circle, inner sep=2mm, fill=gray!50!white] (sn) at (12,0) {$h_Q$};
\node (w12) at (5.5,0.3) {$W_t=W_{t-1}$, w.p.~$1-\delta$};
\node (w11) at (4,-1) {$W_t=W_{t-1}$, w.p.~$\delta$};
\node (w0) at (6,2.2) {$W_t\ne W_{t-1}$};
\graph {
(s1)->(s2)->(sd)->(sn1)->(sn);
{(s2),(sn1),(sn)}->[bend right, dashed](s1);
{(s2),(sn1),(sn)}->[bend left](s1)
};
\end{tikzpicture}
\caption{Transition model shared by $I_1$ and $I_2$. The chain would never approach $h_Q$ at time $t$ if the treatment switches more frequently than every $Q-1$ time periods.}
\label{instance}
\end{figure}

We consider two instances that share the same underlying transition rule as follows (see Figure \ref{instance} for a graphical representation):
\begin{itemize}
    \item If $W_t\ne W_{t-1}$, $\PP{H_{t+1}=h_1|H_t=h_j}=1$, $\forall j$;
    \item If $W_t= W_{t-1}$, $\PP{H_{t+1}=h_1|H_t=h_j}=\delta$, and $\PP{H_{t+1}=h_{(j+1) \land Q}|H_t=h_j}=1-\delta$, $\forall j$.
\end{itemize}
We further set the data generating distributions of $Y_t(h,w)$ under $I_1$ and $I_2$ as
\begin{itemize}
\item Under $I_1$, $Y_t(H_t,1)\sim N(1/2,\sigma^2)$, $Y_t(H_t,0)\sim N(-1/2,\sigma^2)$ if $H_t=h_Q$ and $H_{t-1}\ne h_Q$, otherwise $Y_t(H_t,w)\sim N(0,\sigma^2)$ for all $w\in\{0,1\}$;
\item Under $I_2$, $Y_t(H_t,1)\sim N(-1/2,\sigma^2)$, $Y_t(H_t,0)\sim N(1/2,\sigma^2)$ if $H_t=h_Q$ and $H_{t-1}\ne h_Q$, otherwise $Y_t(H_t,w)\sim N(0,\sigma^2)$ for all $w\in\{0,1\}$.
\end{itemize}

Now that we have constructed a family of hard instances, our next task is to study
the difficulty of detecting whether data is generated from $I_1$ or $I_2$.
We first apply Le Cam's lemma \citep{yu1997assouad} to control the power with which we can distinguish
$I_1$ from $I_2$ in terms of their Kullback–Leibler (KL) divergence.
The specific version of the result we use is what is given in Theorem 5 of \cite{wang2017optimal}. For an i.i.d. sample of size $n$ and any measurable function $\phi : \RR \rightarrow \cb{I_1, \, I_2}$,
\begin{align}
\label{eq:lecam}
\inf_{\hat V} \max_{I\in\{I_1,I_2\}}\mathbb{P}^I[\phi(\hat V)\ne I]\ge \frac{1}{8}e^{-n\gamma^{KL}_{I_1,I_2}(Y_{1:T},W_{1:T})},
\end{align}
where $\gamma^{KL}_{I_1,I_2}(Y_{1:T},W_{1:T})$ denotes the KL divergence between the distribution of $(Y_{1:T},W_{1:T})$ under $I_1$ and $I_2$, respectively, and $n=1$ in the setting we consider. 
It remains to upper-bound the KL divergence between the observed data generated from the two instances.

\begin{lemm}
\label{lemma:KLbound}
There exists some $T_l< \infty$ such that, for all $T> T_l$, the KL divergence between the observed data generated by $I_1$ and $I_2$ is upper bounded as
\begin{equation}
\label{eq:KLbound}
\gamma^{KL}_{I_1,I_2}(Y_{1:T},W_{1:T})
\le \frac{T(1-\delta)^{Q-1}}{2(Q-1)\sigma^2}.
\end{equation}
\label{lemm:kl}
\end{lemm}

\begin{proof}[Proof of Lemma \ref{lemm:kl}]
To start with, we notice that the chain $Y_{1:T}$ resets (i.e., wipes out past information and goes back to $h_1$) every time $W_t\ne W_{t-1}$. Thus, conditionally on $W_{1:T}$, it is possible to divide the length-$T$ chain $Y_{1:T}$ into $b$ independent blocks $Y_{1:B_1}, Y_{(B_1+1):B_2},\dots, Y_{(B_{b-1}+1):B_b}$, where $B_{1:b}=\cb{t:W_{t+1}\ne W_{t-1},t> 1}$. Since the data generating distribution at $t$ is the same under $I_1$ and $I_2$ for states $h_1,\dots,h_{Q-1}$ and for state $h_Q$ if we have already visited it at $t-1$, only $Y_{B_{j-1}+Q}$ will contain information that can be used to distinguish $I_1$ from $I_2$, and thus all the other time periods can be ignored. Moreover, since $W_{1:T}$ are always drawn from the same design, we can write the KL divergence between the observed data generated from the two instances as
\begin{equation}
\begin{split}
\gamma^{KL}_{I_1,I_2}(Y_{1:T},W_{1:T})
&\le \gamma^{KL}_{I_1,I_2}(Y_{1:T},W_{1:T},S_{1:T})\\
&=\gamma^{KL}_{I_1,I_2}(Y_{1:T}\cond W_{1:T},S_{1:T})+  \gamma^{KL}_{I_1,I_2}(W_{1:T},S_{1:T})\\
&=\gamma^{KL}_{I_1,I_2}(Y_{1:T}\cond W_{1:T},S_{1:T})\\
&= 
\sum_{j: B_j-B_{j-1}>Q-1}I(S_{B_{j-1}+Q}=h_Q) \gamma^{KL}_{I_1,I_2}(Y_{B_{j-1}+Q}).
\end{split}
\end{equation}
It is easy to show that the optimal design choice is to set $l= Q-1$, in which case
\begin{equation}
\begin{split}
\gamma^{KL}_{I_1,I_2}(Y_{1:T},W_{1:T})&=
\sum_{i=1}^k I(Z_i\ne Z_{i-1},S_{B_{j-1}+Q}=h_Q) \gamma^{KL}_{I_1,I_2}(Y_{(k-1)l+Q})\\
&= \frac{T(1-\delta)^{Q-1}}{2(Q-1)\sigma^2}.
\end{split}
\end{equation}
\end{proof}

Motivated by the form of \eqref{eq:KLbound}, we consider the instance with $Q$ chosen to satisfy the equation
\begin{equation}
\frac{(1-\delta)^{Q-1}}{Q-1}=\frac{2\sigma^2}{T}.
\end{equation}
It now remains to turn this result into a bound on the mean-squared error of $\htau$.
To this end, we note that, in our model, $+\tau(I_1) = -\tau(I_2) = (1-\delta)^{Q-1}/(2k)$.
Using the function
$\phi(\htau) = I_1$ if $\htau \ge 0$ and $\phi(\htau) = I_2$  if $\htau < 0$ implies:
\begin{equation}
\begin{split}
&\inf_{l, \htau} \max_{I\in\{I_1,I_2\}} \mathbb{E}_e^{I}\sqb{\p{\htau - \tau(I)}^2}\\
&\qquad\qquad\geq \PP[l]{H_t = h_Q}^2 \inf_{\htau} \max_{I\in\{I_1,I_2\}} \mathbb{P}_e^{I}\sqb{\phi(\htau)\ne I} \\
&\qquad\qquad\geq \frac{1}{8e} \cdot \frac{1}{2(1-\delta)^{Q-1}}\\
&\qquad\qquad= \frac{1}{8e} \cdot \frac{1}{2^{Q}}
\label{eq:lower_Q}
\end{split}
\end{equation}
(set $\delta=1/2$)
With $Q=\log T$,
\begin{equation}
\begin{split}
\inf_{l, \htau} \max_{I\in\{I_1,I_2\}} \mathbb{E}_e^{I}\sqb{\p{\htau - \tau(I)}^2}\\
\geq  \frac{1}{8e} \cdot \frac{1}{2^{Q}}
\end{split}
\end{equation}

\section{OLD Proof of Theorem \ref{theo:clt}}

For generality, we consider the Bernoulli switchback design with 
\begin{equation}
W_t \cond W_{t-1} = \begin{cases}
\text{Bernoulli}(\pi) & \text{if $t = (k-1)l + 1$ for some $k = 1, \, \ldots, \, \lfloor T/l \rfloor$,} \\
W_{t-1} & \text{else},
\end{cases}
\end{equation}
i.e., under Definition \ref{defi:switchback}, $\pi=0.5$. To simplify the notation, define $\hpi=k_1/k$, and let $\ff_{t}$ be all the information observed up to time $t$. 
For $i=1,\dots,k$, we denote
\begin{equation}
\bY_i=\frac{1}{l-b}\sum_{t=b+1}^l Y_{(i-1)l+t},
\end{equation}
and use the shorthand $\bY_i(w)$ for $\bY_i\cond \ff_{(i-1)l}, Z_i=w$, $w\in\{0,1\}$.
With a slight abuse of notation, we denote
\begin{equation}
\bmu_i(w)=\EE{\bY_i(w)}=\EE{\bY_i\cond \ff_{(i-1)l}, Z_i=w},
\end{equation}
\begin{equation}
\bmu(w)=\frac{1}{k}\sum_{i=1}^k\bmu_i(w),
\end{equation}
and
\begin{equation}
\btau_i(w)=\EE[\law_w^t]{\bY_i(w)}=\EE{\bY_i\cond Z_i=w, Z_{i-1}=w,\cdots},
\end{equation}
\begin{equation}
\btau(w)=\frac{1}{k}\sum_{i=1}^k\btau_i(w).
\end{equation}
Note that our estimand of interest, GATE, can be written as
\begin{equation}
\tau_{\GATE} = \frac{1}{k}\sum_{t=1}^k \sqb{\btau_i(1)-\btau_i(0)} = \btau(1)-\btau(0).
\end{equation}

We start by proving a central limit theorem for 
\begin{equation}
\begin{split}
\Delta^{(l,b)}
&=\frac{1}{k}\sum_{i=1}^{k}\sqb{\frac{\cb{\bY_i(1)-\bmu_i(1)}Z_i}{\pi}-\frac{\cb{\bY_i(0)-\bmu_i(0)}(1-Z_i)}{1-\pi}}.
\end{split}
\end{equation}
The following lemma is an application of the Rosenblatt central limit theorem for strong mixing sequences \citep{rosenblatt1956central,davis2011rosenblatt}.

\begin{lemm}
Under the assumptions in Theorem \ref{theo:clt},
as $k(l-b)\to\infty$, \[\Delta^{(l,b)}/\sqrt{\Var{\Delta^{(l,b)}}}\to_{d} N (0,1).\]
\label{lemm:clt_oracle}
\end{lemm}

It now remains to connect the residual $\Delta^{(l,b)}$ to the target estimator $\htau_{\text{DM}}^{(l,b)}$. Note that
\begin{equation*}
\begin{split}
\htau_{\text{DM}}^{(l,b)}&=\frac{1}{k}\sum_{i=1}^{k}\sqb{\frac{\bY_i(1)Z_i}{\hpi}-\frac{\bY_i(0)(1-Z_i)}{1-\hpi}}\\
&=\frac{1}{k}\sum_{i=1}^{k}\sqb{\frac{\cb{\bY_i(1)-\bmu_i(1)}Z_i}{\hpi}-\frac{\cb{\bY_i(0)-\bmu_i(0)}(1-Z_i)}{1-\hpi}}+\cb{\bmu(1)-\bmu(0) }+\\
&\qquad\qquad \frac{1}{k}\sum_{i=1}^{k}\sqb{\frac{(\bmu_i(1)-\bmu(1))Z_i}{\hpi} - \frac{(\bmu_i(0)-\bmu(0))(1-Z_i)}{1-\hpi} }  \\
&=\cb{\bmu(1)-\bmu(0) }+\frac{1}{k}\sum_{i=1}^{k}\sqb{\frac{(\bmu_i(1)-\bmu(1))Z_i}{\hpi} - \frac{(\bmu_i(0)-\bmu(0))(1-Z_i)}{1-\hpi} }+\Delta^{(l,b)}+\\
&\qquad\qquad \frac{1}{k}\sum_{i=1}^{k}\sqb{\frac{\cb{\bY_i(1)-\bmu_i(1)}Z_i}{\hpi}-\frac{\cb{\bY_i(1)-\bmu_i(1)}Z_i}{\pi} }-\\
&\qquad\qquad\frac{1}{k}\sum_{i=1}^{k}\sqb{\frac{\cb{\bY_i(0)-\bmu_i(0)}(1-Z_i)}{1-\hpi}-\frac{\cb{\bY_i(0)-\bmu_i(0)}(1-Z_i)}{1-\pi} }.
\end{split}
\end{equation*}
We analyze the right hand side of the equation above term by term. First, by Assumption \ref{assu:mixing} on the mixing time,
\begin{equation}
\begin{split}
\abs{\bmu_i(w)-\btau_i(w)}
&= \abs{\EE{\bY_i\cond \ff_{(i-1)l}, Z_i=w}-\EE{\bY_i\cond Z_i=w, Z_{i-1}=w,\cdots}}\\
&\le \frac{1}{l-b}\sum_{t=(i-1)l+b}^{il}\abs{\EE{Y_t\cond \ff_{(i-1)l}, Z_i=w}-\EE{Y_t\cond Z_i=w, Z_{i-1}=w,\cdots}}\\
&= \oop\p{\frac{1}{l-b} \exp\p{-b/t_{\text{mix}} }}
\end{split}
\end{equation}
Thus,
\begin{equation}
\bmu(1)-\bmu(0)=\tau_{\GATE}+ \oop\p{\frac{1}{l-b} \exp\p{-b/t_{\text{mix}} }},
\end{equation}
and
\begin{equation}
\begin{split}
&\frac{1}{k}\sum_{i=1}^{k}\sqb{\frac{(\bmu_i(1)-\bmu(1))Z_i}{\hpi} - \frac{(\bmu_i(0)-\bmu(0))(1-Z_i)}{1-\hpi}}\\
=& \frac{1}{k}\sum_{i=1}^{k}\sqb{\frac{(\btau_i(1)-\btau(1))Z_i}{\hpi} - \frac{(\btau_i(0)-\btau(0))(1-Z_i)}{1-\hpi}}+ \oop\p{\frac{1}{l-b} \exp\p{-b/t_{\text{mix}} }}.
\end{split}
\end{equation}
We can then apply a standard finite-population central limit theorem \citep{neyman1923applications} and verify that
\begin{equation}
\begin{split}
\frac{1}{\sqrt{k}}\sum_{i=1}^{k}\sqb{\frac{(\btau_i(1)-\btau(1))Z_i}{\hpi} - \frac{(\btau_i(0)-\btau(0))(1-Z_i)}{1-\hpi}}\to_{d} N(0,V_r),
\end{split}
\end{equation}
where $V_r=\Var{\btau_i(1)}/\pi + \Var{\btau_i(0)}/(1-\pi)$.

Next, note that by (\ref{eq:var_delta}),
$\Var{\Delta^{(l,b)}}=\oo(1/{k(l-b)})$. Since $\Delta^{(l,b)}/\sqrt{\Var{\Delta^{(l,b)}}}\to_{d} N (0,1)$ as $k(l-b)\to \infty$, $\Delta^{(l,b)}=\oop(1/\sqrt{k(l-b)})$. Also, for all $i,
j=1,\dots,k$,
\begin{equation*}
\begin{split}
&\EE{\p{\frac{\cb{\bY_i(1)-\bmu_i(1)}Z_i}{\pi}-\frac{\cb{\bY_i(0)-\bmu_i(0)}(1-Z_i)}{1-\pi}} 
\p{\frac{\cb{\btau_j(1)-\btau(1)}Z_j}{\hpi} - \frac{\cb{\btau_j(0)-\btau(0))(1-Z_j}}{1-\hpi}} }\\
=&\EE{\frac{Z_i Z_j\cb{\btau_j(1)-\btau(1)}}{\pi}\EE{\frac{\bY_i(1)-\bmu_i(1)}{\hpi}\cond Z_i,Z_j }}+\\
&\qquad\qquad \EE{\frac{(1-Z_i) (1-Z_j)\cb{\btau_j(0)-\btau(0)}}{1-\pi}\EE{\frac{\bY_i(0)-\bmu_i(0)}{1-\hpi}\cond Z_i,Z_j }}-\\
&\qquad\qquad \EE{\frac{Z_i (1-Z_j)\cb{\btau_j(0)-\btau(0)}}{\pi}\EE{\frac{\bY_i(1)-\bmu_i(1)}{1-\hpi}\cond Z_i,Z_j }}-\\
&\qquad\qquad \EE{\frac{(1-Z_i) Z_j\cb{\btau_j(1)-\btau(1)}}{1-\pi}\EE{\frac{\bY_i(0)-\bmu_i(0)}{\hpi}\cond Z_i,Z_j }}\\
=&0.
\end{split}
\end{equation*}
Since both of the random variables are sum of mean zero terms, 
\begin{equation*}
\Cov{k\Delta^{(l,b)},\sum_{i=1}^{k}\p{\frac{\cb{\btau_i(1)-\btau(1)}Z_i}{\hpi} - \frac{\cb{\btau_i(0)-\btau(0))(1-Z_i}}{1-\hpi}}}=0 .
\end{equation*}

Finally, we show that the fourth and the fifth terms are small. 
By a Taylor expansion of $\hpi$ around $\pi$,
\begin{equation*}
\begin{split}
&\frac{1}{k}\sum_{i=1}^{k}\sqb{\frac{\cb{\bY_i(1)-\bmu_i(1)}Z_i}{\hpi}-\frac{\cb{\bY_i(1)-\bmu_i(1)}Z_i}{\pi} }\\
=&-\frac{1}{k}\sum_{i=1}^{k}\frac{\cb{\bY_i(1)-\bmu_i(1)}Z_i}{\pi^2}\p{\hpi-\pi} + R\\
=&-\frac{\p{\hpi-\pi}}{\pi}\cdot\frac{1}{k}\sum_{i=1}^{k}\frac{\cb{\bY_i(1)-\bmu(1)}Z_i}{\pi} + R,
\end{split}
\end{equation*}
where $R$ is a remainder term of smaller order than the terms in the equation. Applying the same central limit theorem as in Lemma \ref{lemm:clt_oracle}, it is easy to show that
\begin{equation*}
\frac{1}{k}\sum_{i=1}^{k}\frac{\cb{\bY_i(1)-\bmu(1)}Z_i}{\pi}=\oop\p{\frac{1}{\sqrt{k(l-b)}}}.
\end{equation*}
Since $\hpi-\pi=\oop\p{1/\sqrt{k}}$,
\begin{equation*}
\frac{1}{k}\sum_{i=1}^{k}\sqb{\frac{\cb{\bY_i(1)-\bmu_i(1)}Z_i}{\hpi}-\frac{\cb{\bY_i(1)-\bmu_i(1)}Z_i}{\pi} }=\oop\p{\frac{1}{k\sqrt{l-b}}}.
\end{equation*}
Similarly,
\begin{equation*}
\frac{1}{k}\sum_{i=1}^{k}\sqb{\frac{\cb{\bY_i(0)-\bmu_i(0)}(1-Z_i)}{1-\hpi}-\frac{\cb{\bY_i(0)-\bmu_i(0)}(1-Z_i)}{1-\pi} }=\oop\p{\frac{1}{k\sqrt{l-b}}}.
\end{equation*}

To summarize, if $l-b=\oo(1)$,
\begin{equation*}
\begin{split}
\htau_{\text{DM}}^{(l,b)}
&=\tau_{\GATE}+ \Phi_1 +\Phi_2+\oop\p{\exp\p{-b/t_{\text{mix}} }}+ \oop\p{\frac{1}{{k}}},
\end{split}
\end{equation*}
where $\Phi_1/\sqrt{k}\to_{d} N(0,V_r)$ and $\Phi_2/\sqrt{k}\to_{d} N(0,k\Var{\Delta^{(l,b)}})$ as $k\to\infty$;
Otherwise, if $l-b\to\infty$ as $k\to\infty$,
\begin{equation*}
\begin{split}
\htau_{\text{DM}}^{(l,b)}
&=\tau_{\GATE}+\Phi_1+ \oop\p{\frac{1}{l-b} \exp\p{-b/t_{\text{mix}} }} + \oop\p{\frac{1}{\sqrt{k(l-b)}}}.
\end{split}
\end{equation*}
Thus, as $k\to\infty$,
\begin{equation*}
\begin{split}
\sqrt{k} \p{\htau_{\text{DM}}^{(l,b)}-\tau_\GATE} \to_{d} N(0,V_r+k\Var{\Delta^{(l,b)}}),
\end{split}
\end{equation*}
if $l-b=\oo(1)$ and $k=\smallO\p{\exp\p{2b/t_{\text{mix}}}}$,
and
\begin{equation*}
\begin{split}
\sqrt{k} \p{\htau_{\text{DM}}^{(l,b)}-\tau_\GATE} \to_{d} N(0,V_r),
\end{split}
\end{equation*}
if $l-b\to\infty$ as $k\to\infty$ and $k=\smallO\p{(l-b)^2 \exp\p{2b/t_{\text{mix}}}}$.

\begin{proof}[Proof of Lemma \ref{lemm:clt_oracle}]
For $i=1,\dots,k$, $t=(i-1)l+1,\dots,il$, define
\begin{equation}
D_{t}=\frac{\cb{Y_t(1)-\mu_t(1)}W_t}{\pi}-\frac{\cb{Y_t(0)-\mu_t(0)}(1-W_t)}{1-\pi},
\end{equation}
where $\mu_t(w)=\EE{Y_t\cond \ff_{(i-1)l}, W_t=w}$, $w=0,1$. Note that $\Delta^{(l,b)}$ can also be written as
\begin{equation*}
\begin{split}
\Delta^{(l,b)}
&=\frac{1}{k(l-b)}\sum_{t\in \mathcal{I}^{(l, b)}}\sqb{\frac{\cb{Y_t(1)-\mu_t(1)}W_t}{\pi}-\frac{\cb{Y_t(0)-\mu_t(0)}(1-W_t)}{1-\pi}}.
\end{split}
\end{equation*}

To start with, we show that the mean-zero sequence
\begin{equation}
\cb{\frac{\cb{Y_t(1)-\mu_t(1)}W_t}{\pi}-\frac{\cb{Y_t(0)-\mu_t(0)}(1-W_t)}{1-\pi}}_{t=1,\dots,kl}
\label{eq:score_seq}
\end{equation}
satisfies the strong mixing condition, in the sense that
\begin{equation}
\sup_{1\le t\le kl}\sup_{A\in\mathcal{A}_t,B\in\mathcal{B}_{t+h}}\abs{\PP{A\cap B}-\PP{A}\PP{B}} \le d(h),
\end{equation}
where $\mathcal{A}_t=\sigma(S_{n:n\le t},W_{n:n\le t},\epsilon_{n:n\le t})$, $\mathcal{B}_t=\sigma(S_{n:n\ge t},W_{n:n\ge t},\epsilon_{n:n\ge t})$, and $d:\ZZ_{\ge 0 }\to[0,1]$ satisfies that $d(h)\to 0$ as $h\to\infty$. Again by the assumption on mixing time, for $h$ large enough such that $t$ and $t+h$ are in different blocks (i.e., when $h>l-b$ so that $W_t$ and $W_{t+h}$ are independent),
\begin{equation}
\begin{split}
\sup_{A\in\mathcal{A}_{t},B\in\mathcal{B}_{t+h}}\abs{\PP{A\cap B}-\PP{A}\PP{B}}
\le& \sup_{A\in\mathcal{A}_{i},B\in\mathcal{B}_{i+h}}\abs{\PP{B|A}-\PP{B}}\\
\le& \sup_{s\in\mathcal{S},B\in\mathcal{B}_{T,i+h}}\abs{\PP{B|S_{il}=s}-\PP{B}}\\
\le& c_1\exp\p{-h/t_\mix },
\end{split}
\end{equation}
where $c_1<\infty$ is a constant.
Thus, (\ref{eq:score_seq}) satisfies the strong mixing condition.

Next, we show that $\Var{k(l-b)\Delta^{(l,b)}}\to \infty$ as $n\to\infty$, with
\begin{equation*}
\begin{split}
\Var{k(l-b)\Delta^{(l,b)}}&=\Var{\sum_{t\in \mathcal{I}^{(l, b)}}D_t}\ge \Var{\sum_{t\in \mathcal{I}^{(l, b)}}\epsilon_t }\\
&\ge k(l-b)\sigma_0^2\\
&= \Omega\p{k(l-b)}.
\end{split}
\end{equation*}

Finally, we upper bound the fourth moment of $k(l-b)\Delta^{(l,b)}$ and show it's of order $\oo\p{{k^2(l-b)^2}}$. Note that
\begin{equation}
\begin{split}
k^4(l-b)^4\EE{\p{\Delta^{(l,b)} }^4}&=\EE{\cb{\sum_{t\in \mathcal{I}^{(l, b)}}D_t}^4}\\
&=\sum_t\EE{D_t^4}+\sum_{t\ne s}\EE{D_t^3D_s}+\sum_{t\ne s}\EE{D_t^2D_s^2}+\\
&\qquad\qquad\sum_{t\ne s\ne m}\EE{D_t^2D_sD_m}+\sum_{t\ne s\ne m\ne n}\EE{D_tD_sD_mD_n}.
\end{split}
\label{eq:fourth_mom}
\end{equation}
To calculate the fourth order terms above, we start by noticing that $\cb{Y_t(w)-\mu_t(w)}^x=\oo\p{1}$, $x\le 4$, $w\in\cb{0,1}$, due to the assumption that $\abs{Y_t}\le \Gamma_0$. Furthermore,
for all $t<s$, $w\in\cb{0,1}$, similar to the calculation of upper bound on the variance in Section \ref{proof:Theorem1},
\begin{enumerate}
\item if $t$ and $s$ are from the same block $i$, then
\begin{equation*}
\begin{split}
 \abs{\EE{Y_s(w)\cond \ff_t}- \EE{\mu_s(w)\cond \ff_t}}
&=\abs{\EE{Y_s(w)\cond \ff_t}- \EE{Y_s(w)\cond \ff_t}}\\
&=0= \oop\p{\exp(-(s-t)/t_\mix)};   
\end{split}
\end{equation*}
\item if $t$ and $s$ are from the two different blocks $i$ and $j$, then
\begin{equation*}
\begin{split}
 \abs{\EE{Y_s(w)\cond \ff_t}- \EE{\mu_s(w)\cond \ff_t}}
&=\abs{\EE{Y_s(w)\cond \ff_t}- \EE{Y_s(w)\cond \ff_{(j-1)l}}}\\
&= \oop\p{\exp(-((j-1)l-t)/t_\mix})= \oop\p{\exp(-(s-t)/t_\mix}). 
\end{split}
\end{equation*}
\end{enumerate}

Now we return to the calculation of (\ref{eq:fourth_mom}).
Without loss of generality, we consider the case where $t< s< m< n$ \footnote{In other cases, it is easy to show that those moments can also be bounded with similar terms of the same order, e.g., $\EE{D_tD_sD_mD_n}=\oo\p{\exp\p{-(n-s)/t_\mix}}$ when $s< t< m< n$,}. From the two properties we obtained above,
\begin{equation*}
\EE{D_t^4} =\oo\p{1},
\end{equation*}
\begin{equation*}
\EE{D_t^3D_s}=\EE{D_t^3\EE{D_s\cond \ff_t}} =\oo\p{\exp\p{-(s-t)/t_\mix}},
\end{equation*}
\begin{equation*}
\EE{D_t^2D_s^2}=\EE{D_t^2\EE{D_s^2\cond \ff_t}} =\oo\p{\exp\p{-(s-t)/t_\mix}},
\end{equation*}
\begin{equation*}
\begin{split}
\EE{D_t^2D_sD_m}&=\EE{D_t^2\EE{D_s\EE{D_m\cond \ff_s}\cond \ff_t}}\\ 
&=\oo\p{\exp\p{-(m-s)/t_\mix}}\cdot\oo\p{\exp\p{-(s-t)/t_\mix}}\\
&=\oo\p{\exp\p{-(m-t)/t_\mix}},
\end{split}
\end{equation*}
\begin{equation*}
\begin{split}
\EE{D_tD_sD_mD_n}&=\EE{D_t\EE{D_s\EE{D_m\EE{D_n\cond \ff_m}\cond \ff_s}\cond \ff_t}}\\
&=\oo\p{\exp\p{-(n-m)/t_\mix}}\cdot\oo\p{\exp\p{-(m-s)/t_\mix}}\cdot\oo\p{\exp\p{-(s-t)/t_\mix}}\\
&=\oo\p{\exp\p{-(n-t)/t_\mix}}.
\end{split}
\end{equation*}
Putting everything together, we have $(\ref{eq:fourth_mom}) =\oo\p{k(l-b)}=\oo\p{{k^2(l-b)^2}}$.
It then follows directly from \cite{rosenblatt1956central} and \cite{davis2011rosenblatt} that
\[\Delta^{(l,b)}/\sqrt{\Var{\Delta^{(l,b)}}}\to_{d} N (0,1)\]
as $k(l-b)\to\infty$, where
\begin{equation}
\begin{split}
\Var{\Delta^{(l,b)}}
&= \frac{1}{k^2(l-b)^2}\sum_t\Var{D_t}+ \frac{1}{k^2(l-b)^2}\sum_{t\ne s}\EE{D_tD_s}\\
&= \oo\p{\frac{1}{k(l-b)}}+\frac{1}{k^2(l-b)^2}\sum_{t\ne s}\oo\p{\exp\p{-\abs{t-s}/t_\mix}}\\
&=\oo\p{\frac{1}{k(l-b)} }.
\label{eq:var_delta}
\end{split}
\end{equation}

\end{proof}

\end{document}